\documentclass[acmsmall,screen,natbib=false]{acmart}

\setcopyright{none}

\RequirePackage[
  datamodel=acmdatamodel,
  style=acmauthoryear,
  sortcites=true,
  sorting=nyt,
  backend=biber,
  ]{biblatex}

\usepackage[frozencache=true,cachedir=minted-cache]{minted}

\usepackage[appendix=append,bibliography=common]{apxproof}

\usepackage[capitalise,noabbrev]{cleveref}
\usepackage{wrapfig}
\usepackage{hyperref}
\usepackage{svg}
\usepackage{physics}
\usepackage{enumitem}
\usepackage{adjustbox}
\usepackage{bbold}
\usepackage{yfonts}

\usepackage{quiver}
\RequirePackage{mathtools}
\usepackage{newunicodechar}
\usepackage{stmaryrd}
\usepackage{qcircuit}
\usepackage{multicol}
\allowdisplaybreaks

\def\fcmp{\mathbin{\raise 0.6ex\hbox{\oalign{\hfil$\scriptscriptstyle      \mathrm{o}$\hfil\cr\hfil$\scriptscriptstyle\mathrm{9}$\hfil}}}}

\newcommand{\cat}[1]{\ensuremath{\mathbf{#1}}}
\newcommand{\Unitary}{\cat{Unitary}}
\newcommand{\id}{\ensuremath{\mathrm{id}}}

\newcommand{\V}{\ensuremath{\mathsf{V}}}
\newcommand{\Z}{\ensuremath{\mathsf{Z}}}
\newcommand{\Pg}{\ensuremath{\mathsf{P}}}
\newcommand{\Sg}{\ensuremath{\mathsf{S}}}
\newcommand{\X}{\ensuremath{\mathsf{X}}}
\newcommand{\K}{\ensuremath{\mathsf{K}}}
\newcommand{\SX}{\ensuremath{\sqrt{\mathsf{X}}}}
\newcommand{\T}{\ensuremath{\mathsf{T}}}
\newcommand{\Had}{\ensuremath{\mathsf{H}}}
\newcommand{\CX}{\ensuremath{\mathsf{CX}}}
\newcommand{\CSX}{\ensuremath{\mathsf{CSX}}}
\newcommand{\CSXI}{\ensuremath{\mathsf{CSXdg}}}
\newcommand{\CZ}{\ensuremath{\mathsf{CZ}}}
\newcommand{\CCX}{\ensuremath{\mathsf{CCX}}}

\newcommand{\Swap}{\ensuremath{\mathsf{SWAP}}}
\newcommand{\Midswap}{\ensuremath{\mathsf{Midswap}}}

\newcommand{\Mat}{\ensuremath{\mathsf{Mat}}}
\newcommand{\ThreeMat}{\ensuremath{\mathsf{3Mat}}}
\newcommand{\Ctrl}{\ensuremath{\mathsf{Ctrl}}}
\newcommand{\nCtrl}{\ensuremath{\mathsf{nCtrl}}}

\newcommand{\of}{\mathbin{:}}
\newcommand{\defeq}{\mathbin{::=}}
\newcommand{\fromto}{\leftrightarrow}
\DeclarePairedDelimiter{\sem}{\llbracket}{\rrbracket}
\DeclarePairedDelimiter{\stsem}{\llparenthesis}{\rrparenthesis}

\newtheoremrep{theorem}{Theorem}
\newtheoremrep{lemma}[theorem]{Lemma}
\newtheoremrep{proposition}[theorem]{Proposition}
\newtheoremrep{corollary}[theorem]{Corollary}

\newcommand{\PiLang}{\ensuremath{\Pi}}
\newcommand{\SPiLang}{\ensuremath{\sqrt{\Pi}}}

\newcommand{\zt}{\ensuremath{\mathbb{0}}}
\newcommand{\ot}{\ensuremath{\mathbb{1}}}
\newcommand{\bool}{\ensuremath{\mathbb{2}}}

\newcommand{\cm}{\mathit}
\newcommand{\pid}{\cm{id}}
\newcommand{\swapp}{\cm{swap}^{+}}
\newcommand{\assocp}{\cm{assocr}^+}
\newcommand{\associp}{\cm{assocl}^+}

\newcommand{\unitepl}{\cm{unite}^{+}l}
\newcommand{\unitipl}{\cm{uniti}^{+}l}

\newcommand{\swapt}{\cm{swap}^{\times}}
\newcommand{\assoct}{\cm{assocr}^{\times}}
\newcommand{\associt}{\cm{assocl}^{\times}}

\newcommand{\unitetl}{\cm{unite}^{\times}l}
\newcommand{\unititl}{\cm{uniti}^{\times}l}
\newcommand{\dist}{\cm{dist}}
\newcommand{\factor}{\cm{factor}}
\newcommand{\absorbl}{\cm{absorbl}}
\newcommand{\factorzr}{\cm{factorzr}}

\newcommand{\seqq}{\fcmp}

\newcommand{\one}{\textsf{1}}
\newcommand{\sprod}{\bullet}

\newcommand{\ctrlgate}{\textsc{ctrl}}
\newcommand{\vgate}{\textsc{v}}

\newcommand{\vigate}{\textsc{vi}}
\newcommand{\ogate}{\textsc{w}}
\newcommand{\oigate}{\textsc{wi}}
\newcommand{\xgate}{\textsc{x}}

\newcommand{\cxgate}{\textsc{cx}}
\newcommand{\ccxgate}{\textsc{ccx}}

\newunicodechar{ˡ}{${}^{l}$}
\newunicodechar{ʳ}{${}^{r}$}
\newunicodechar{Π}{$\Pi$}
\newunicodechar{ℕ}{$\mathbb{N}$}
\newunicodechar{⊡}{$\boxdot$}
\newunicodechar{≡}{$\equiv$}
\newunicodechar{⟦}{$\llbracket$}
\newunicodechar{⟧}{$\rrbracket$}
\newunicodechar{⦇}{$\llparenthesis$}
\newunicodechar{⦈}{$\rrparenthesis$}
\newunicodechar{𝟚}{$\mathbb{2}$}
\newunicodechar{ϕ}{$\phi$}
\newunicodechar{ᵤ}{$_{u}$}
\newunicodechar{₀}{$_{0}$}
\newunicodechar{₁}{$_{1}$}
\newunicodechar{₂}{$_{2}$}
\newunicodechar{⇔}{$\Leftrightarrow$}
\newunicodechar{₃}{$_{3}$}
\newunicodechar{₄}{$_{4}$}
\newunicodechar{₅}{$_{5}$}
\newunicodechar{₆}{$_{6}$}
\newunicodechar{∘}{$\circ$}
\newunicodechar{⟷}{$\leftrightarrow$}
\newunicodechar{◎}{$\fatsemi$} 
\newunicodechar{⊗}{$\otimes$}
\newunicodechar{₊}{${}_{+}$}
\newunicodechar{⊕}{$\oplus$}
\newunicodechar{λ}{$\lambda$}
\newunicodechar{⁻}{${}^{-}$}
\newunicodechar{⊚}{$@$}
\newunicodechar{π}{$\pi$}
\newunicodechar{∀}{$\forall$}
\newunicodechar{⋆}{$\star$}
\newunicodechar{∎}{$\blacksquare$}
\newunicodechar{𝔽}{$\mathbb{F}$}
\newunicodechar{𝕋}{$\mathbb{T}$}
\newunicodechar{◯}{$\odot$} 

\begin{document}
\title{With a Few Square Roots, Quantum Computing is as Easy as \PiLang} 
\author{Jacques Carette} 
\email{carette@mcmaster.ca}
\orcid{0000-0001-8993-9804}
\affiliation{
  \institution{McMaster University}
  \city{Hamilton}
  \state{Ontario}
  \country{Canada}
}

\author{Chris Heunen}
\email{Chris.Heunen@ed.ac.uk}
\orcid{0000-0001-7393-2640}
\affiliation{
  \institution{University of Edinburgh}
  \city{Edinburgh}
  \country{United Kingdom}
}

\author{Robin Kaarsgaard}
\email{kaarsgaard@imada.sdu.dk}
\orcid{0000-0002-7672-799X}
\affiliation{
  \institution{University of Southern Denmark}
  \city{Odense}
  \country{Denmark}
}

\author{Amr Sabry}
\email{sabry@indiana.edu}
\orcid{0000-0002-1025-7331}
\affiliation{
  \institution{Indiana University}
  \city{Bloomington}
  \state{Indiana}
  \country{United States of America}
}

\begin{abstract}
  Rig groupoids provide a semantic model of \PiLang, a universal
  classical reversible programming language over finite types. We
  prove that extending rig groupoids with just two maps and three
  equations about them results in a model of quantum computing that
  is computationally universal and equationally sound and complete for
  a variety of gate sets. The first map corresponds to an
  $8^{\text{th}}$ root of the identity morphism on the unit $1$. The
  second map corresponds to a square root of the symmetry on $1+1$. As
  square roots are generally not unique and can sometimes even be
  trivial, the maps are constrained to satisfy a nondegeneracy axiom,
  which we relate to the Euler decomposition of the Hadamard gate. The
  semantic construction is turned into an extension of \PiLang, called
  \SPiLang, that is a computationally universal quantum programming
  language equipped with an equational theory that is sound and
  complete with respect to the Clifford gate set, the standard gate
  set of Clifford+T restricted to $\le 2$ qubits, and the
  computationally universal Gaussian Clifford+T gate set.
\end{abstract}

\maketitle
\thispagestyle{empty}
\pagestyle{plain}

\section{Introduction} 

Just like in the classical case, quantum computing can be built up
from booleans and associated operations. The quantum version of
boolean negation is the \X\ gate defined by $\X \ket{0} = \ket{1}$ and
$\X \ket{1} = \ket{0}$.  The quantum circuit model also includes a
gate \SX\ (also known as the \V\ gate) that is the ``square root of \X.''
Informally \SX\ performs half of the action of the \X\ gate, \textit{i.e.},
if we imagine a trajectory from $\ket{0}$ to $\ket{1}$ and another
trajectory from $\ket{1}$ to $\ket{0}$, then one application of \SX\
follows half the relevant trajectory. The standard approach to model
this behaviour is to explicitly express the intermediate midpoints
as complex vectors~\cite{hayes:sqrtnot,satohetal:ctrlv}:
\[
  \SX \ket{0} = \frac{1+i}{2} \ket{0} + \frac{1-i}{2} \ket{1}
  \qquad\qquad
  \SX \ket{1} = \frac{1-i}{2} \ket{0} + \frac{1+i}{2} \ket{1}
\]
One can verify that:
\begin{equation*}\label{eq:VVX}
\begin{array}{rcl}
\SX (\SX~\ket{0}) &=&
  \SX (\frac{1+i}{2} \ket{0} + \frac{1-i}{2} \ket{1}) \\[1.0ex]
&=& 
    \frac{1+i}{2} \SX \ket{0} + \frac{1-i}{2} \SX \ket{1} \\[1.0ex]
&=& 
    \frac{1+i}{2} (\frac{1+i}{2} \ket{0} + \frac{1-i}{2} \ket{1}) + 
    \frac{1-i}{2} (\frac{1-i}{2} \ket{0} + \frac{1+i}{2} \ket{1}) \\[1.0ex]
&=& 
    \frac{i}{2} \ket{0} + \frac{1}{2} \ket{1} - \frac{i}{2} \ket{0} + \frac{1}{2} \ket{1} \\[1.0ex]
&=& 
    \ket{1} 
\end{array}
\end{equation*}
and similarly that $\SX (\SX~\ket{1}) = \ket{0}$.  As is
evident in this tiny example, reasoning this way about quantum programs
is overwhelmed by complex numbers and linear algebra.

Our first insight is that we do \emph{not} need to explicitly
represent the intermediate points. All we need to know about them are
two things: (i) they exist, and (ii) they satisfy one critical
axiom. Technically, we demonstrate that the following categorical
model is, not only computationally universal for quantum computing,
but also sound and complete for several modes of unitary quantum
computing.  

\medskip
\begin{defno}
  The model consists of a rig groupoid $(\cat{C},\otimes,\oplus,O,I)$
  equipped with maps $\omega \colon I \to I$ and
  $\V \colon I \oplus I \to I \oplus I$ satisfying the equations:
  \[
    (\text{E1}) ~\omega^8 = \id
    \qquad
    (\text{E2}) ~\V^2 = \sigma_\oplus 
    \qquad
    (\text{E3}) ~\V \circ \Sg \circ \V = \omega^2 \sprod \Sg \circ \V \circ \Sg
  \]
  where $\circ$ is sequential composition, $\sprod$ is scalar
  multiplication (cf.\ Def.~\ref{def:sprod}), $\sigma_\oplus$ is the
  symmetry on $I \oplus I$, exponents are iterated sequential
  compositions, and $\Sg \colon I \oplus I \to I \oplus I$ is defined
  as $\Sg = \id \oplus \omega^2$.
\end{defno}

\smallskip\noindent
In the definition, the rig groupoid $\cat{C}$ models an underlying
reversible classical programming language. By convention, booleans in
this language are represented as values of type $I \oplus I$ with one
injection representing \texttt{false}, the other representing
\texttt{true}, and the symmetry $\sigma_\oplus \colon I \oplus I \to I
\oplus I$
\begin{wrapfigure}{r}{0.32\textwidth}
\begin{center}
\includegraphics[scale=0.2]{"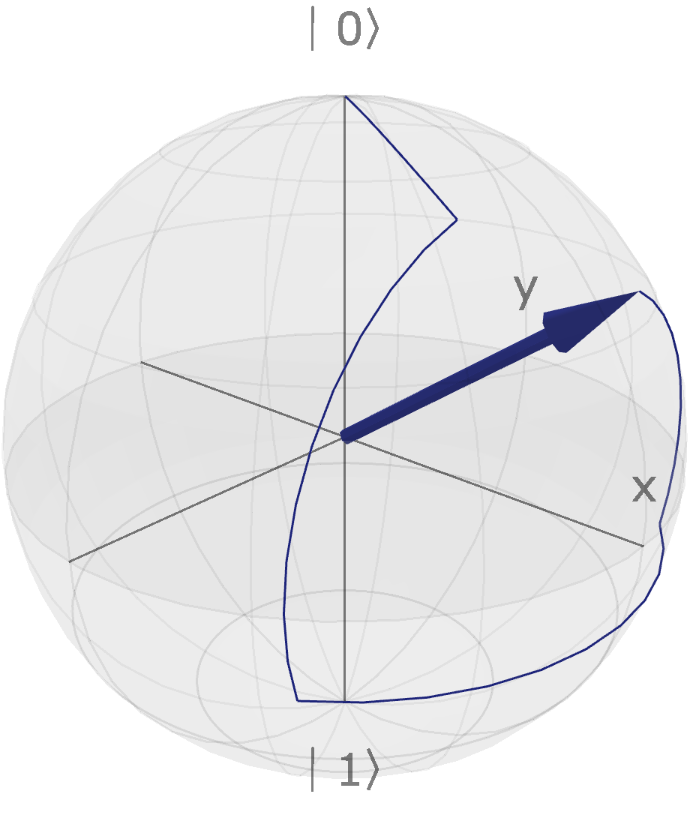"}
\qquad\qquad
\includegraphics[scale=0.2]{"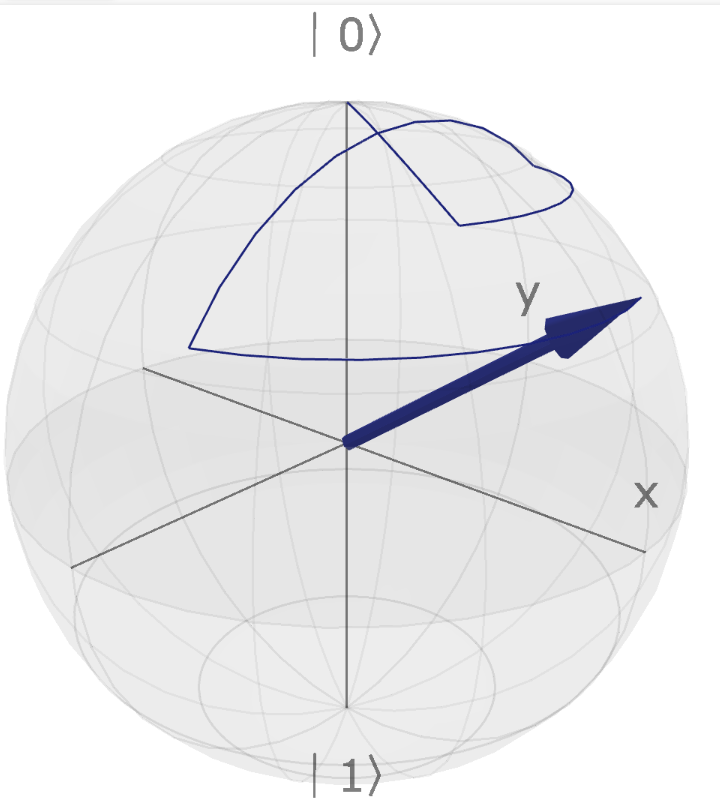"}
\end{center}
\caption{$XZX$ and $ZXZ$ rotations with all angles at $\pi/2$.}
\label{fig:rotations}
\end{wrapfigure}
representing boolean negation. The quantum model has two  
 additional morphisms $\omega$ and \V. The map $\omega$ is a primitive 
$8^{\text{th}}$ root of the identity; its semantics is partially  
 specified by (E1). The map \V\ is the square root of boolean  
negation; its semantics is partially specified by (E2).  So far, we  
 have postulated the existence of square roots but without needing to  
 write any actual complex numbers: they are just morphisms partially
specified by (E1) and (E2). At this point, it would be consistent to
choose $\omega = \id$ but this would not lead to a universal quantum
model. To understand how (E3) selects just the ``right'' square
root, we recall that the \emph{Euler decomposition} expresses any
1-qubit unitary gate as a product of a global phase and three
rotations along two fixed orthogonal axes, and that \Sg\ and \V\
correspond to rotations in complementary bases. In that light, axiom
(E3) picks the $Z$-basis and the $X$-basis as the two axes and
enforces that decompositions along $ZXZ$ or $XZX$ are equal (up to a
physically unimportant global phase). This ensures that it is
immaterial which of \Sg\ and \V\ rotations is mapped to the $Z$- or
$X$-basis and additionally ensures that the angle of the \Sg\ rotation
(induced by the $\omega^2$ in the definition of \Sg) is $\pi/2$. As a
helpful illustration, Fig.~\ref{fig:rotations} shows that, with the
standard choice for the computational basis in the $Z$-direction,
starting from an arbitrary state (near the North pole in the figure),
a sequence of $\pi/2$-$XZX$ rotations (top) is equivalent to a
sequence of $\pi/2$-$ZXZ$ rotations (bottom). Were the angle of the
$Z$-rotation different due to a different choice of~$\omega$, the two
sequences of rotations would not be equivalent.

This approach reduces reasoning about quantum programs to manipulating
the coherence conditions of rig
categories~\cite{laplaza:distributivity} extended with the axioms
(E1), (E2), and (E3). The calculation that $\SX \circ \SX = X$ follows
by (E2). Many quantum equivalences follow similarly. For example, the
proof that $\Sg \circ \Sg$ is equivalent to the $\Z$ gate defined as
$\id \oplus \omega^4$ follows by:
\[
\Sg \circ \Sg 
 = (\id \oplus \omega^2) \circ (\id \oplus \omega^2) 
 = (\id \circ \id) \oplus (\omega^2 \circ \omega^2) 
 = \id \oplus \omega^4 
 = \Z
\]
The proof uses just the coherence conditions of rig categories and is,
along with many other results, formalised in an extension of the
\texttt{agda-categories} library~\cite{10.1145/3437992.3439922}
included in the supplementary material.

The equational theory extracted from the semantic model is sound and
complete with respect to \emph{arbitrary Clifford circuits},
\emph{Clifford+T circuits of at most $2$ qubits}, and \emph{arbitrary
  Gaussian Clifford+T circuits}. These completeness theorems, 
Thms.~\ref{thm:ncliff}, \ref{thm:clifft2}, and \ref{thm:gclifft}, form 
our main technical results:
\begin{itemize}
\item Completeness for \emph{Arbitrary Clifford circuits}
  (cf. Thm~\ref{thm:ncliff}). Circuits built from Clifford gates are
  important in quantum computing for two related reasons. First,
  Clifford gates are exactly those quantum gates that normalise the
  Pauli matrices, which provide a linear-algebraic basis for a single
  qubit. Clifford gates include, and are in fact generated by, \Had,
  \Sg, and \CX. Second, although Clifford circuits may ``look
  quantum,'' they are in fact efficiently simulatable by a
  probabilistic classical computation, by the Gottesman-Knill
  theorem~\cite{gottesmanknill}. 
\item Completeness for \emph{Clifford+T circuits of at most 2 qubits}
  (cf. Thm~\ref{thm:clifft2}). To move beyond classical probabilistic
  machines in computational power, other quantum gates need to be
  considered. One popular choice is to extend the Clifford set with
  the \T\ gate. The restriction to $\leq 2$ qubits is a stepping stone
  to the next result. 
\item Completeness for \emph{Arbitrary Gaussian Clifford+T circuits}
  (cf. Thm~\ref{thm:comp_unz}). Another universal quantum gate set is
  given by $\{\X,\CX,\CCX,\Sg,\K\}$~\cite{Amy2020numbertheoretic,
    bianselinger:unz}. Such circuits can be characterised
  algebraically as those unitary matrices with entries in the ring
  $\mathbb{Z}[\frac{1}{2},i]$ of Gaussian dyadic
  rationals~\cite{Amy2020numbertheoretic}.
\end{itemize}
To summarise, we have developed a vastly simplified axiomatic
treatment of quantum computation using the coherence conditions of rig
categories extended with morphisms modeling roots of the identity and
a square root of the symmetry $\sigma_\oplus \colon I \oplus I \to I
\oplus I$.

This formalism provides, to our knowledge, the first sound and
complete equational theory for a computationally universal unitary
quantum programming language. As this approach avoids imposing
specific assumptions about gate sets or implementation details, it
could serve to bridge the gap between quantum programming languages
and the various gate sets used in the quantum circuit model. Further,
it could serve as a "theory of equational theories" capable of
describing and analyzing various modes of quantum computing, such as
different gate sets, without preference to any specific
approach. While this paper primarily focuses on qubit circuits due to
the abundance of finite presentation results, it does not reflect an
inherent limitation or assumption within the formalism. In fact, we
propose that this formalism could be used equally well to represent
and analyse circuits from qudit gate sets (\textit{e.g.}, qutrit
Clifford+T~\cite{yehwetering:quditcliffordt}).

\paragraph*{Related work}
Our result is distinguished from other calculi based on
ZX~\cite{coeckeduncan:zx}, notably ZH~\cite{backenskissinger:zh} and
PBS/LOv~\cite{clementetal:lov} in two fundamental aspects. First, ZX
and ZH describe quantum theory, not quantum computation. That is, they
are complete for all linear maps, not for unitary ones only. Indeed,
one of the major problems associated with the ZX calculus is circuit
extraction: to ensure that rewriting a quantum circuit ends up with a
quantum circuit again. This problem is
\#P-hard~\cite{beaudrapkissingerwetering:circuitextraction}. Second,
these calculi do not have universal equational theories,
as some of the axiom schemas involve existential quantifiers,
resulting from the Euler decomposition, that cannot be
eliminated~\cite{duncanperdrix:euler}. The theory presented here
builds on a different line of research that led to advances in
reversible quantum computing (\textit{e.g.}, \cite{10.1145/3498667,
  heunenkaarsgaard:qie, heunenkaarsgaardkarvonen:arrows,
  glueckkaarsgaardyokoyama:reversible}) and equational theories of
quantum circuits and unitaries~\cite{selinger:clifford,
  bianselinger:cliffordt, bianselinger:unz} (see also
\cite{thomsenkaarsgaardsoeken:ricercar}) arising from number-theoretic
insights (\textit{e.g.},
\cite{Amy2020numbertheoretic,gilesselinger:exact}). Our resulting
theory is sound, complete and universal, never considers more general
linear maps (unlike ZH/ZX), and relies only on universally quantified
equations (unlike PBS/LOv). Our work complements the work of
\citet{staton:effects}, which provides a sound and complete equational
theory of state preparation and measurement (which we do not consider
here), but does not consider an equational theory of unitaries.

\paragraph*{Outline.} We assume familiarity with category theory (in
particular rig categories, monoidal categories, and string diagrams)
and with the fundamentals of quantum computing. We provide a brief
review in the next section for the necessary notation and
conventions. Sec~\ref{sec:comb} motivates the use of combinator-based
languages to reason about quantum circuits. Sec.~\ref{sec:sqrtpi}
introduces the formal syntax of the combinator language \SPiLang\ used
as a technical device in this paper. Sec.~\ref{sec:den} gives the
denotational semantics of \SPiLang\ in extended rig
groupoids. Sec.~\ref{sec:sc} includes the main technical results that
establish soundness and completeness of \SPiLang\ for a variety of
gate sets. Sec.~\ref{sec:eq} describes the equational theory in
action. The concluding section puts the results in a larger context
and discusses their significance. Some of the proofs are relegated to
the appendix.



\section{Background} 

We recall here some basics of unitary quantum computing and rig
categories.

\subsection{Unitary quantum computing}
\label{subsec:unitaryqc}
For more details about this topic we refer to textbooks such as
\cite{nielsenchuang:qcqi, yanofskymannucci:qccs}.

Closed quantum systems are modelled mathematically by complex Hilbert
spaces $H$, which are complex vector spaces with an inner product
$\bra{-}\ket{-}$ that are complete as metric spaces (with respect to
the metric induced by the inner product). For example, a one-qubit
system is represented by $\mathbb{C}^2$, with vectors $\ket{0}
= (\begin{smallmatrix} 1 \\ 0
\end{smallmatrix})$ and $\ket{1} = (\begin{smallmatrix}
0 \\ 1
\end{smallmatrix})$ representing the two classical states. Hilbert
spaces $H$ and $K$ can be combined to form new ones using the
\emph{direct sum} $H \oplus K$ and \emph{tensor product} $H \otimes
K$: these can be seen as analogues of sum types and product types in
the sense that $\mathbb{C}^n \oplus \mathbb{C}^m \cong
\mathbb{C}^{n+m}$ and $\mathbb{C}^n \otimes \mathbb{C}^m \cong
\mathbb{C}^{nm}$.

Every linear map $f$ on a Hilbert space is associated with a
\emph{(Hermitian) adjoint} $f^\dagger$ satisfying $\bra{f
  \phi}\ket{\psi} = \bra{\phi}\ket{f^\dagger \psi}$. The
discrete time evolution of closed quantum systems is described by
\emph{unitaries}, which are linear isomorphisms $U$ satisfying $U^{-1}
= U^\dagger$. Some important examples of unitaries on $\mathbb{C}^2$
include the \emph{Hadamard} gate $\Had$, the $\X$ gate (the quantum
analogue of the classical \textsc{not} gate), and the \emph{phase
gates} $\Z$, $\Sg$, and $\T$, given by the matrices:
\[
\Had = \tfrac{1}{\sqrt{2}} \left(\begin{smallmatrix}
  1 & 1 \\    
  1 & -1 
\end{smallmatrix}\right)
\qquad
\X = \left(\begin{smallmatrix}
  0 & 1 \\    
  1 & 0 
\end{smallmatrix}\right)
\qquad
\Z = \left(\begin{smallmatrix}
  1 & 0 \\    
  0 & -1 
\end{smallmatrix}\right)
\qquad
\Sg = \left(\begin{smallmatrix}
  1 & 0 \\    
  0 & i 
\end{smallmatrix}\right)
\qquad
\T = \left(\begin{smallmatrix}
  1 & 0 \\    
  0 & \tfrac{1+i}{\sqrt{2}} 
\end{smallmatrix}\right)
\]
Any unitary $U$ acting on $H$ can be extended to a \emph{controlled}
variant acting on $\emph{C}^2 \otimes H$, given in matrix form by the
block diagonal matrix
\[
\left(\begin{matrix}
I & 0 \\ 0 & U  
\end{matrix}\right)
\]
where $I$ is the identity on $H$. This controlled-$U$ will apply $U$
to $H$ only if the given qubit was in the state $\ket{1}$; otherwise
it will do nothing. For example, the controlled-$X$ gate $\CX$ is
given by
\[
\CX = \left(\begin{smallmatrix}
  1 & 0 & 0 & 0 \\
  0 & 1 & 0 & 0 \\
  0 & 0 & 0 & 1 \\    
  0 & 0 & 1 & 0
\end{smallmatrix}\right)
\]
Similar to classical hardware description, low-level quantum
computations can be described at the level of qubits and gates using
quantum circuits, which we describe in further detail in
Sec.~\ref{sec:comb}, save for one crucial definition concerning
when a quantum gate set can be said to be universal:
\begin{definition}[Computational universality~\cite{aharonov:toffolihadamard}]
  A set of quantum gates $G$ is said to be \emph{strictly universal}
  if there exists a constant $n_0$ such that for any $n \ge n_0$, the
  subgroup generated by $G$ is dense in $\mathrm{SU}(2^n)$. The set $G$
  is said to be \emph{computationally universal} if it can be used to
  simulate to within $\epsilon$ error any quantum circuit which uses
  $n$ qubits and $t$ gates from a strictly universal set with only
  polylogarithmic overhead in $(n,t,1/\epsilon)$.
\end{definition}

\subsection{Rig categories}
\label{subsec:rigcats}
We refer to \cite{awodey:cat,heunenvicary:cqt} for more on (monoidal)
categories, and to \cite{johnsonyau:bimonoidal} for a recent textbook
on rig categories and their applications.

A category $\mathbf{C}$ is an algebraic structure capturing typed
processes: a category consists of some types (\emph{objects}) $X,Y,Z$
and some processes (\emph{morphisms}) $f,g,h$ such that each process
$f$ is assigned an input type (\emph{domain}) $X$ and an output type
(\emph{codomain}) $Y$, written $f : X \to Y$. Processes $f : X \to Y$
and $g : Y \to Z$ can be composed to form a new process $g \circ f : X
\to Z$ in such a way that composition is associative and unital
(\textit{i.e.}, every object $X$ is associated with an \emph{identity}
$\id_X : X \to X$ such that $f \circ \id_X = f = \id_Y \circ f$ for
all $f : X \to Y$). Thus, categories describe theories of processes
that can be composed in sequence: if a morphism $f$ has an inverse
$f^{-1}$ such that $f \circ f^{-1} = \id$ and $f^{-1} \circ f = \id$,
we say that $f$ is an \emph{isomorphism}. A category which contains
only isomorphisms is called a \emph{groupoid}.

A \emph{symmetric monoidal category} $(\mathbf{C},\otimes,I)$ is a
category that also permits parallel composition of objects and
morphisms: whenever one has objects $X$ and $Y$, there exists an
object $X \otimes Y$; similarly, morphisms $f : X \to Y$ and $g : Z
\to W$ give rise to $f \otimes g : X \otimes Z \to Y \otimes
W$. Further, we require that there is a distinguished object $I$ and
families of isomorphisms (indexed by objects $X,Y,Z$) $\lambda_\otimes
: I \otimes X \to X$ and $\rho_\otimes : I \otimes X \to X$ (the
\emph{unitors}); $\alpha_\otimes : (X \otimes Y) \otimes Z \to X
\otimes (Y \otimes Z)$ (the \emph{associator}); and $\sigma_\otimes :
X \otimes Y \to Y \otimes X$ (the \emph{symmetry}), satisfying some
equations (see, \textit{e.g.}, \cite[Chapter 1]{heunenvicary:cqt}).

A \emph{rig category} (or \emph{bimonoidal category})
$(\mathbf{C},\otimes,\oplus,I,O)$ is a category which is symmetric
monoidal in two different ways, such that one monoidal structure
distributes over the other. Precisely, it is a category such that
$(\mathbf{C},\otimes,I)$ and $(\mathbf{C},\oplus,O)$ are both
symmetric monoidal categories, and there are families of isomorphisms
(indexed by objects $X,Y,Z$) $\delta_L : X \otimes (Y \oplus Z) \to (X
\otimes Y) \oplus (X \otimes Z)$ and $\delta_R : (X \oplus Y) \otimes
Z \to (X \otimes Z) \oplus (Y \otimes Z)$ (the \emph{distributors})
and $\delta_0^L : O \otimes X \to O$ and $\delta_0^R : X \otimes O \to
O$ (the \emph{annihilators}), subject again to some equations (see
\cite{laplaza:distributivity}). A rig category which is simultaneously
a groupoid is called a \emph{rig groupoid}. The category \Unitary{}
of finite-dimensional Hilbert spaces and unitaries forms a rig groupoid
with its tensor product $\otimes$ and direct sum $\oplus$.

\section{Reasoning about Quantum Circuits with Combinators}
\label{sec:comb}
 
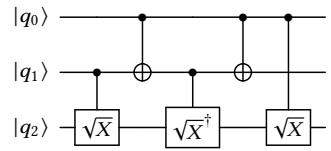
\begin{wrapfigure}{r}{0.32\textwidth}
  \vspace{-0.75\baselineskip}
  {\footnotesize
  \begin{center}
  $\Qcircuit @C=.7em @R=0.7em @!R {
      & \lstick{\ket{q_0}} & \qw             & \ctrl{1} & \qw
      & \ctrl{1}
      & \ctrl{2} & \qw \\
      & \lstick{\ket{q_1}} & \ctrl{1}        & \targ    & \ctrl{1}
      & \targ
      & \qw & \qw \\
      & \lstick{\ket{q_2}} & \gate{\sqrt{X}} & \qw      & \gate{\sqrt{X}^\dagger}
      & \qw
      & \gate{\sqrt{X}} & \qw
    } 
  $
  \end{center}}
\caption{\label{fig:sw}Quantum circuit for \CCX.}
\end{wrapfigure}
The \textit{lingua franca} of quantum computing is that of quantum
circuits.  Like boolean circuits consisting of bit-carrying wires
connecting boolean gates, quantum circuits consist of wires carrying
qubits connecting quantum gates.  For example, the circuit in
Fig.~\ref{fig:sw} has~5 controlled unitary gates acting on 3
qubits. In order, the first three gates are: controlled-\SX\ (aka
\CSX), controlled-not (aka \CX), and controlled-inverse-\SX\ (aka
\CSXI).

\subsection{Circuits as Matrices} 

Quantum circuits have a canonical reading as complex matrices. The
quantum gates stand for specific unitary matrices which are combined
by matrix multiplication when gates are composed sequentially, and by
tensor product when gates are composed in parallel. For example, the
controlled gates used in the circuit above denote the following
matrices:
\[
  \CSX = \tfrac{1}{2}\left(\begin{smallmatrix}
    2 & 0 & 0 & 0 \\
    0 & 2 & 0 & 0 \\
    0 & 0 & -1+i & -1-i \\    
    0 & 0 & -1-i & -1+i 
  \end{smallmatrix}\right)
  \qquad
  \CX = \left(\begin{smallmatrix}
    1 & 0 & 0 & 0 \\
    0 & 1 & 0 & 0 \\
    0 & 0 & 0 & 1 \\    
    0 & 0 & 1 & 0 
  \end{smallmatrix}\right)
  \qquad
    \CSXI = \tfrac{1}{2}\left(\begin{smallmatrix}
    2 & 0 & 0 & 0 \\
    0 & 2 & 0 & 0 \\
    0 & 0 & -1-i & -1+i \\    
    0 & 0 & -1+i & -1-i 
  \end{smallmatrix}\right)
\]
which when all multiplied following the layout of the circuit produce:
$$\left(\begin{smallmatrix}
     1 & 0 & 0 & 0 & 0 & 0 & 0 & 0 \\    
     0 & 1 & 0 & 0 & 0 & 0 & 0 & 0 \\    
     0 & 0 & 1 & 0 & 0 & 0 & 0 & 0 \\    
     0 & 0 & 0 & 1 & 0 & 0 & 0 & 0 \\    
     0 & 0 & 0 & 0 & 1 & 0 & 0 & 0 \\    
     0 & 0 & 0 & 0 & 0 & 1 & 0 & 0 \\    
     0 & 0 & 0 & 0 & 0 & 0 & 0 & 1 \\    
     0 & 0 & 0 & 0 & 0 & 0 & 1 & 0 
\end{smallmatrix}\right)$$
The reader may recognise the resulting matrix as the denotation of the
Toffoli (aka \CCX) gate~\cite{10.1007/3-540-10003-2104}. Indeed the
equivalence of \CCX\ to the circuit in Fig.~\ref{fig:sw} is an
instance of the Sleator-Weinfurter~[\citeyear{PhysRevLett.74.4087}]
construction. Evidently, one way to establish the equivalence is to
reduce both circuits to a common matrix. If such a low-level algebraic
manipulation is undesirable, a high-level, but informal proof, would
proceed by case analysis on the possible values of $q_0q_1$:
\begin{itemize}
\item if both $q_0q_1$ are 0, then no control gate is activated and
  the circuit behaves like the identity; 
\item if one of $q_0q_1$ is 1 and the other is 0, then both \SX\ and its
  inverse are activated and the circuit is again equivalent to the identity; 
\item if both $q_0q_1$ are 1, then two instances of \SX\ are
  activated which negates $q_2$.
\end{itemize}
To summarise, the circuit in Fig.~\ref{fig:sw} negates $q_2$ exactly
when both $q_0q_1$ are 1, which is exactly the behaviour of the
Toffoli gate. We will formalise this example using our calculus in
Sec.~\ref{sec:eq}.

\subsection{Circuits as Rig Morphisms} 
\label{sub:rig}

It is relatively easy to find \emph{some} collection of local rewrite
rules that are sound for quantum circuits composed of particular gate
sets. It is much harder to find a \emph{complete}
collection that guarantee that any equivalent
quantum circuits can be transformed to one another. We solve this
problem as follows. First, we build on the completeness result for
classical reversible circuits~\cite{10.1145/3498667} by including all
the coherence conditions for rig categories as a foundation for
reasoning about the classical subset of gates (\textit{e.g.}, \X, \CX, \CCX,
etc.) To reason about the purely quantum gates (\textit{e.g.}, \SX, \Had, \T,
etc.) we build on a collection of insights explained below.

The first insight is to not worry about gates at all but instead
exploit the rig groupoid structure that provides two constructors
$\oplus$ and $\otimes$ that behave in a distributive way, like $+$ and
$\times$ in the rig of natural numbers. The $\oplus$ construct, which
is not present in formalisms such as the
ZX-calculus~\cite{coeckeduncan:zx} provides a way to build quantum
gates from first principles by exploiting the fact that a qubit is a
two-dimensional additive structure $\ot \oplus \ot$. For example, the
rig structure provides, among others, the natural isomorphisms
$\lambda_\otimes \of I \otimes A \to A$, $\sigma_\oplus \of A \oplus B
\to B \oplus A$, and $\delta_R \of (A \oplus B) \otimes C \to (A
\oplus C) \otimes (B \oplus C)$ which can be used to define gates as
follows.  First, we isolate two patterns \Mat\ and \Ctrl\ to construct
simple gates and their controlled versions:
\begin{align*}
\Mat & \defeq \lambda_\otimes \oplus \lambda_\otimes \circ \delta_R \of
(I \oplus I) \otimes A \to A \oplus A \\
\Ctrl~m & \defeq \Mat^{-1} \circ (\id \oplus m) \circ \Mat \of
(I \oplus I) \otimes A \to (I \oplus I) \otimes A
\end{align*}
The definition of \Ctrl\ above is parametric in $m : I \oplus I \to I
\oplus I$, enabling the definitions of the classical gates:
\begin{align*}
\X & \defeq \sigma_\oplus \of I \oplus I \to I \oplus I
\\
\CX & \defeq \Ctrl~\X \of (I \oplus I) \otimes (I \oplus I) \to (I
\oplus I) \otimes (I \oplus I)
\\
\CCX & \defeq \Ctrl~\CX \of (I \oplus I) \otimes ((I \oplus I) \otimes
(I \oplus I)) \to (I \oplus I) \otimes ((I \oplus I) \otimes (I \oplus
I))
\end{align*}
These patterns would also provide controlled versions of single qubit
quantum gates if we managed to express them. To that end, we use the
insight that, by the Euler decomposition, single qubit quantum gates
can be expressed as a product $\phi \cdot P Q P’$, where $\phi$ is a
phase, $P$ and $P’$ are rotations in one basis, and $Q$ is a rotation
in a complementary basis. Thus, the categorical framework ``only''
needs to express phase gates in two complementary bases such as the
canonical $Z$ and $X$ bases; it turns out that this is relatively
straightforward once the framework includes roots of unity and a
square root of $\sigma_\oplus$. Each root of unity $\omega$ directly
provides phase gate $\id \oplus \omega$ in the \Z-basis; phase gates
in the \X-basis are obtained by the change of basis induced by \Had\
which itself can be defined using roots of unity and the square root
of $\sigma_\oplus$ (cf. Fig.~\ref{fig:abbrev}). The technical
challenge is that square roots are not unique, so for example
postulating some \V\ such that $\V \circ \V = \sigma_\oplus$ is not
sufficient to determine \V.  Axiom $(E_3)$, however, is sufficient to
completely determine all the required square roots. The final product
is an equational theory that provides (formalisable) proofs for
circuit equivalences that only require a modest extension of
conventional categorical reasoning.

\section{A Universal Quantum Language: \SPiLang}
\label{sec:sqrtpi}

We present the syntax of \SPiLang, whose underlying language is the
classical reversible language \PiLang\ that is universal for
reversible computing over finite types and whose semantics is
expressed in the rig groupoid of finite sets and
bijections~\cite{jamessabry:infeff} . After reviewing the design of
\PiLang\, we introduce the extension \SPiLang.

\subsection{The Core Language: \PiLang}
\label{sub:ctrl}

In reversible boolean circuits, the number of input bits matches the
number of output bits. Thus, a key insight for a programming language
of reversible circuits is to ensure that each primitive operation
preserves the number of bits, which is just a natural number. The
algebraic structure of natural numbers as the free commutative
semiring (or, commutative rig), with $(0,+)$ for addition, and
$(1,\times)$ for multiplication then provides sequential, vertical,
and horizontal circuit composition. Generalising these ideas, a typed
programming language for reversible computing should ensure that every
primitive expresses an isomorphism of finite types, \textit{i.e.}, a
permutation. 

\begin{figure}  
\begin{align*}
  b & \defeq \zt \mid \ot \mid b+b \mid b \times b & \text{(value types)} \\
  t & \defeq b \fromto b & \text{(combinator types)} \\
  \mathit{iso} & \defeq \pid \mid \swapp \mid \assocp \mid
  \associp \mid \unitepl \mid \unitipl \mid \absorbl \mid \factorzr 
& \text{(isomorphisms)} \\
    & \mid \swapt \mid \assoct \mid \associt 
      \mid \unitetl \mid \unititl 
     \mid \dist \mid \factor \\
  c & \defeq \mathit{iso} \mid c \seqq c \mid c + c \mid c \times c & \text{(combinators)}
\end{align*}
\caption{The syntax of~\PiLang.}
\label{fig:pisyntax}
\end{figure}

The syntax of the language~\PiLang, shown in Fig.~\ref{fig:pisyntax},
captures this concept. Type expressions $b$ are built from the empty
type (\zt), the unit type (\ot), the sum type ($+$), and the product
type ($\times$). A type isomorphism $c : b_1 \fromto b_2$ models a
reversible circuit that permutes the values in $b_1$ and $b_2$. These
type isomorphisms are built from the primitive identities
$\mathit{iso}$ and their compositions. The \PiLang-isomorphisms are
not ad hoc: they correspond exactly to the laws of a \emph{rig}
operationalised into invertible
transformations~\cite{10.1007/978-3-662-49498-1,CARETTE202215} which
have the types in Fig.~\ref{fig:pitypes}. Each line in the top part of
the figure has the pattern $c_1 : b_1 \fromto b_2 : c_2$ where~$c_1$
and~$c_2$ are self-duals; $c_1$ has type $b_1 \fromto b_2$ and $c_2$
has type $b_2 \fromto b_1$.

\begin{figure}[t]
\begin{equation*}
  \begin{array}{rcrclcl}
    \pid & \of & b &\fromto& b & \of & \pid \\
    \swapp & \of & b_1 + b_2 &\fromto& b_2 + b_1 & \of & \swapp \\
    \assocp & \of & (b_1 + b_2) + b_3 &\fromto& b_1 + (b_2 + b_3) & 
    \of & \associp \\
    \unitepl & \of & \zt + b &\fromto& b & \of & \unitipl \\
    \swapt & \of & b_1 \times b_2 &\fromto& b_2 \times b_1 & \of &
    \swapt \\
    \assoct & \of & (b_1 \times b_2) \times b_3 &\fromto& b_1 \times
    (b_2 \times b_3) & \of & \associt \\
    \unitetl & \of & \ot \times b &\fromto& b & \of & \unititl \\
    \dist & \of & (b_1 + b_2) \times b_3 &\fromto& (b_1 \times b_3)
    + (b_2 \times b_3) & \of & \factor \\
    \absorbl & \of & b \times \zt &\fromto& \zt & \of & \factorzr 
  \end{array}
\end{equation*}
\begin{equation*}
  \frac{c_1 \of b_1 \fromto b_2 \quad c_2 \of b_2 \fromto b_3}{c_1 \seqq c_2 
    \of b_1 \fromto b_3}
\qquad
  \frac{c_1 \of b_1 \fromto b_3 \quad c_2 \of b_2 \fromto b_4}{c_1 + c_2 
  \of b_1 + b_2 \fromto b_3 + b_4} 
\qquad
  \frac{c_1 \of b_1 \fromto b_3 \quad c_2 \of b_2 \fromto b_4}{c_1 \times c_2 
  \of b_1 \times b_2 \fromto b_3 \times b_4}
\end{equation*}
\caption{\label{fig:pitypes}Types for \PiLang\ combinators}
\end{figure}

\begin{figure}[t]
\begin{align*}
  \ctrlgate~c &= \dist \seqq \pid + (\pid \times c ) \seqq \factor \\
  \one : \ot \fromto \ot &= \pid \\
  \xgate : \bool \fromto \bool &= \swapp \\
  \cxgate : \bool \times \bool \fromto \bool \times \bool &= \ctrlgate~\swapp \\
  \ccxgate : \bool \times \bool \times \bool \fromto \bool \times \bool \times \bool
       &= \ctrlgate~\cxgate
\end{align*}
\caption{\label{fig:pid}Derived \PiLang\ constructs.}
\end{figure}

The instance of $\pid$ at type $\ot \fromto \ot$ plays an important
role as it will induce \emph{scalars}; it is given the distinguished
name \one\ when used as a scalar value. To see how this language
expresses reversible circuits, we first define types that describe
sequences of booleans ($\bool^{n}$). We use the type $\bool = \ot +
\ot$ to represent booleans with the left injection representing
\texttt{false} and the right injection representing
\texttt{true}. Boolean negation (the \xgate-gate) is straightforward
to define using the primitive combinator $\swapp$. We can represent
$n$-bit words using an $n$-ary product of boolean values. To express
the \cxgate- and \ccxgate-gates we need to encode a notion of
conditional expression. Such conditionals turn out to be expressible
using the distributivity and factoring identities of rigs as shown in
Fig.~\ref{fig:pid}. An input value of type $\bool \times b$ is
processed by the $\dist$ operator, which converts it into a value of
type $(\ot \times b) + (\ot \times b)$. Only in the right branch,
which corresponds to the case when the boolean is \texttt{true}, is
the combinator~\ensuremath{c} applied to the value of
type~\ensuremath{b}.  The inverse of $\dist$, namely $\factor$ is
applied to get the final result. Using this conditional, \cxgate\ is
defined as $\ctrlgate~\xgate$ and the Toffoli \ccxgate\ gate is
defined as $\ctrlgate~\cxgate$. Because \PiLang\ can express the
Toffoli gate and can generate ancilla values of type~$\ot$ as needed,
it is universal for classical reversible circuits.

\begin{theorem}[\PiLang\ Expressivity]
  \PiLang\ is universal for classical reversible circuits, \textit{i.e.}, boolean
  bijections $\bool^n \to \bool^n$ (for any natural number $n$).
\end{theorem}

\subsection{Classical Completeness} 

A crucial fact for the rest of the paper is the existence of an
equational theory for \PiLang\ that is sound and complete for the
permutation semantics. The equations for the theory were collected in
a second level of \PiLang\ syntax as level-2
combinators~\cite{10.1007/978-3-662-49498-1}. Each level-2 combinator
is of the form $c_1 \fromto_2 c_2$ for appropriate $c_1$ and $c_2$ of
the same type $b_1 \fromto b_2$ and asserts that $c_1$ and~$c_2$
denote the same bijection. For example, among the large number of
equations, we have the following level-2 combinators dealing with
associativity:

\begin{minted}{agda}
  assoc◎l  : c₁ ◎ (c₂ ◎ c₃) ⟷₂ (c₁ ◎ c₂) ◎ c₃
  assoc◎r  : ((c₁ ◎ c₂) ◎ c₃) ⟷₂ (c₁ ◎ (c₂ ◎ c₃))
  assocl+l : ((c₁ + (c₂ + c₃)) ◎ assocl₊) ⟷₂ (assocl₊ ◎ ((c₁ + c₂) + c₃))
  assocl+r : (assocl₊ ◎ ((c₁ + c₂) + c₃)) ⟷₂ ((c₁ + (c₂ + c₃)) ◎ assocl₊)
\end{minted}

\begin{theorem}[\PiLang\ Full Abstraction and Adequacy~\cite{10.1145/3498667}]
The equational theory of \PiLang\ expressed using the level-2
combinators $\fromto_2$ is sound and complete with respect to its
semantics in the weak symmetric rig groupoid of finite sets and
permutations. 
\end{theorem}

\noindent
As a consequence, we may use any classical reversible circuit identity
(\textit{i.e.}, any identity involving only rig terms in the category of finite
sets and permutations) without explicit proof, as such a proof can be
reconstructed using the theorem above. In particular, we will freely
use the classical identities below involving various combinations of
\CX\ and \Swap\ gates (which can all be straightforwardly verified by
explicit computation):
{\small
\begin{align}\tag{P1}
    \Qcircuit @C=.5em @R=.7em {
     & \targ & \qw & \qw \\
     & \ctrl{-1} & \ctrl{1} & \qw \\
     & \qw & \targ & \qw
    } 
    & \qquad \raisebox{-5mm}{=} \qquad
    \Qcircuit @C=.5em @R=.7em {
     & \qw      & \targ     & \qw \\
     & \ctrl{1} & \ctrl{-1} & \qw \\
     & \targ    & \qw       & \qw
     } \label{eq:p1}
\\[0.5\baselineskip] \tag{P2}
  \Qcircuit @C=.7em @R=0.5em @!R {
    & \targ     & \ctrl{1} & \qw      & \ctrl{1} & \targ     & \qw
    \\
    & \ctrl{-1} & \targ    & \ctrl{1} & \targ    & \ctrl{-1} & \qw
    \\
    & \qw       & \qw      & \targ    & \qw      & \qw       & \qw
  } 
  & \qquad \raisebox{-4mm}{=} \qquad
  \Qcircuit @C=.7em @R=0.5em @!R {
    & \qswap     & \qw      & \qswap     & \qw \\
    & \qswap\qwx & \ctrl{1} & \qswap\qwx & \qw \\
    & \qw        & \targ    & \qw        & \qw
  } \label{eq:p2}
\\[0.5\baselineskip] \tag{P3}
  \Qcircuit @C=.7em @R=0.5em @!R {
    & \qswap      & \qw       & \qswap     & \qw \\
    & \qswap\qwx  & \targ     & \qswap\qwx & \qw \\
    & \qw         & \ctrl{-1} & \qw   & \qw
  } 
  & \qquad \raisebox{-4mm}{=} \qquad
  \Qcircuit @C=.7em @R=0.5em @!R {
    & \qw      & \qw       & \targ     & \qw       & \qw
    & \qw \\
    & \ctrl{1} & \targ     & \ctrl{-1} & \targ     & \ctrl{1}
    & \qw \\
    & \targ    & \ctrl{-1} & \qw       & \ctrl{-1} & \targ
    & \qw
  } \label{eq:p3}
\\[0.5\baselineskip] \tag{P4}
  \Qcircuit @C=.7em @R=0.7em @!R {
    & \ctrl{1} & \targ     & \qw
    & \ctrl{1} & \targ     & \qw
    & \ctrl{1} & \targ     & \qw
    & \qw \\
    & \targ    & \ctrl{-1} & \targ
    & \targ    & \ctrl{-1} & \targ
    & \targ    & \ctrl{-1} & \targ
    & \qw \\
    & \qw      & \qw       & \ctrl{-1}
    & \qw      & \qw       & \ctrl{-1}
    & \qw      & \qw       & \ctrl{-1}
    & \qw 
  } 
  & \qquad \raisebox{-5mm}{=} \qquad
  \Qcircuit @C=.7em @R=1.5em {
    & \qw & \qw & \qw \\
    & \qw & \qw & \qw \\
    & \qw & \qw & \qw    
  } \label{eq:p4}
  \\[0.5\baselineskip] \tag{P5}
  \Qcircuit @C=.7em @R=0.7em @!R {
    & \qw       & \qw      & \ctrl{1} 
    & \qw       & \qw      & \ctrl{1} 
    & \qw       & \qw      & \ctrl{1}
    & \qw \\
    & \targ     & \ctrl{1} & \targ 
    & \targ     & \ctrl{1} & \targ 
    & \targ     & \ctrl{1} & \targ
    & \qw \\
    & \ctrl{-1} & \targ    & \qw
    & \ctrl{-1} & \targ    & \qw
    & \ctrl{-1} & \targ    & \qw
    & \qw 
  }
  & \qquad \raisebox{-5mm}{=} \qquad
  \Qcircuit @C=.7em @R=1.5em {
    & \qw & \qw & \qw \\
    & \qw & \qw & \qw \\
    & \qw & \qw & \qw    
  } \label{eq:p5}
  \\[0.5\baselineskip] \tag{P6}
  \Qcircuit @C=.7em @R=0.7em @!R {
    & \ctrl{1} & \targ     & \ctrl{1} & \qw \\
    & \targ    & \ctrl{-1} & \targ    & \qw
  }
  & \qquad \raisebox{-2.5mm}{=} \qquad
  \Qcircuit @C=.7em @R=1.5em {
    & \qswap     & \qw \\
    & \qswap\qwx & \qw
  } \label{eq:p6}
\end{align}}

\subsection{Adding Square Roots}

The remarkable fact is that all it takes for a programming language to
be universal for quantum computing with a sound and complete
equational theory is the modest extension to \PiLang\ in
Fig.~\ref{fig:spisyntaxtypes}.

The extension consists of a square root \vgate\ of \xgate\ and an
$8^{\text{th}}$ root \ogate\ of the identity combinator \one. To
maintain reversibility, we add not just these square roots but their
inverses \vigate\ and \oigate\ as well. The semantics of the new
combinators is partially specified by Eqs.~\ref{eq:sqrt1}
and~\ref{eq:sqrt2}. From these equations and the original level-2
combinators, we can derive properties of the inverses, \textit{e.g.}:
\[\begin{array}{rcl@{\qquad}l}
  \xgate &\fromto_2& \vgate ~\seqq~ \vgate & \textrm{(by 2-reversibility)}\\
  \vigate ~\seqq~ \xgate ~\seqq~ \xgate &\fromto_2&
    \vigate ~\seqq~ \vgate ~\seqq~ \vgate ~\seqq~ \xgate & \textrm{(by compatibility)}\\
  \vigate &\fromto_2& \vgate ~\seqq~ \xgate & \textrm{(by inverses and unit)} \\
  \\
  1 &\fromto_2& \ogate^8 & \textrm{(by 2-reversibility)}\\
  \oigate ~\seqq~  1 &\fromto_2& \oigate ~\seqq~ \ogate^8 & \textrm{(by compatibility)}\\
  \oigate &\fromto_2& \ogate^7 & \textrm{(by inverses and unit)} 
  \end{array}\]
As discussed earlier, Eqs. \ref{eq:sqrt1} and \ref{eq:sqrt2} do not completely
determine the meaning of the new combinators, however. In particular,
they do not exclude the trivial square root $\ogate = \one$. To get a
non-trivial semantics, we also impose Eq.~\ref{eq:e3}.

\begin{figure}
\noindent
\textbf{Syntax}
\begin{align*}
  \mathit{iso} & ::= \cdots \mid \vgate \mid \vigate \mid \ogate \mid \oigate
  & \text{(isomorphisms)}
\end{align*}
\noindent 
\textbf{Types}
\[\begin{array}{rcrclcl}
    \vgate & \of & \bool & \fromto & \bool & \of & \vigate \\
    \ogate & \of & 1 & \fromto & 1 & \of & \oigate 
  \end{array}\]
\noindent 
\textbf{Equations}
  \begin{enumerate}[label={(E\arabic*)}]
  \item \label{eq:sqrt1} $\vgate^2  \fromto_2 \xgate$
  \item \label{eq:sqrt2} $\ogate^8  \fromto_2 \one$
  \item \label{eq:e3}
  $\vgate ~\seqq~ (\pid + \ogate^2) ~\seqq~ \vgate
  \fromto_2
  \unititl ~\seqq~ \ogate^2 \times ((\pid + \ogate^2) ~\seqq~ \vgate ~\seqq~ (\pid + \ogate^2))
   ~\seqq~ \unitetl$
\end{enumerate}
\caption{The \SPiLang\ extension of \PiLang.}
\label{fig:spisyntaxtypes}
\end{figure}

\section{Denotational Semantics}
\label{sec:den}

By design, \PiLang\ has a natural model in \emph{rig
  groupoids}~\cite{10.1007/978-3-662-49498-1,10.1145/3498667}. Indeed,
every atomic isomorphism of \PiLang\ corresponds to a coherence
isomorphism in a rig category, while sequencing corresponds to
composition, and the two parallel compositions are handled by the two
monoidal structures. Inversion corresponds to the canonical dagger
structure of groupoids. This interpretation is summarised in
the top part of Fig.~\ref{fig:semantics}.

\subsection{Postulating Square Roots} 

We will postulate the existence of certain square roots to a rig
groupoid to obtain models of \SPiLang. Ideally, there would be a
universal categorical construction that formally adjoins $n$th roots
of specified (endo)morphisms to a given (rig) category. The
traditional way in commutative algebra to adjoin a square root of $r$
to a ring $R$ is to first move to the polynomial ring $R[x]$ in one
variable $x$, and then to quotient out the ideal generated by $x^2-r$
to force $x^2=r$. This method is fraught with problems in the
categorical case, because there is no analogue of the polynomial ring,
no good analogue of quotients by ideals, and because it only works for
endomorphisms.

\begin{figure}
  \begin{equation*}
    \begin{array}{rclrcl}
      \multicolumn{6}{l}{\text{\textbf{Types}}} \\
      \sem{\zt} &=& O & \sem{\ot} &=& I \\
      \sem{b_1 + b_2} &=& \sem{b_1} \oplus \sem{b_2} &
      \sem{b_1 \times b_2} &=& \sem{b_1} \otimes \sem{b_2} \\ \\
      \multicolumn{6}{l}{\text{\textbf{\PiLang\ Terms}}} \\

      \sem{\pid} &=& \id &
      \sem{c_1 \seqq c_2} &=& \sem{c_2} \circ \sem{c_1} \\
      \sem{c_1 + c_2} &=& \sem{c_1} \oplus \sem{c_2} &
      \sem{c_1 \times c_2} &=& \sem{c_1} \otimes \sem{c_2} \\ \\

      \sem{\assocp} &=& \alpha_\oplus &
      \sem{\associp} &=& \alpha_\oplus^{-1} \\

      \sem{\unitipl} &=& \lambda_\oplus^{-1} &
      \sem{\unitepl} & = & \lambda_{\oplus} \\

      \sem{\assoct} &=& \alpha_\otimes & 
      \sem{\associt} &=& \alpha_\otimes^{-1} \\

      \sem{\unititl} &=& \lambda_{\otimes}^{-1} &
      \sem{\unitetl} &=& \lambda_{\otimes} \\

      \sem{\swapp} &=& \sigma_\oplus &
      \sem{\swapt} &=& \sigma_\otimes \\

      \sem{\dist} &=& \delta_R &
      \sem{\factor} &=& \delta_R^{-1} \\
 
      \sem{\absorbl} &=& \delta_0 & 
      \sem{\factorzr} &=& \delta_0^{-1} \\ \\

      \multicolumn{6}{l}{\text{\textbf{\SPiLang\ Terms}}} \\
      \sem{\ogate} &=& \omega &
      \sem{\oigate} &=& \omega^7 \\
      \sem{\vgate} &=& \V &
      \sem{\vigate} &=& \V^3
    \end{array}
  \end{equation*}
  \caption{Semantics of \PiLang\ in rig groupoids
    $(\cat{C},\otimes,\oplus,O,I)$ and of \SPiLang\ in models of
    \SPiLang.}
  \label{fig:semantics}
\end{figure}
 
Another way to formally adjoin a square root of
$A \stackrel{f}{\to} B$ is to add a new object and two new morphisms
$A \stackrel{{}^{1/2}f}{\to} \bullet \stackrel{f^{1/2}}{\to} B$, to
take the free category on the resulting directed graph, and then
quotient out composition that already existed in the base category, as
well as quotienting out $f \sim f^{1/2} \circ {}^{1/2}f$. This does
work in arbitrary categories, satisfies a universal property, and can
be applied to arbitrary sets of morphisms $f$ simultaneously. The new
square roots automatically interact well with inverses in
groupoids. However, to respect rig structure we would have to take
free combinations of $\oplus$ and $\otimes$, and the benefit of the
universal property would be lost to bureaucracy.

Instead of pursuing general constructions, we will therefore simply
postulate what we need of a categorical model. It will be clear that
at least one model exists.

\begin{definition}
\label{def:sprod}
Given a scalar $s : I \to I$ and a morphism $f : X \to Y$, define the
\emph{scalar multiplication} of $f$ by $s$ on the left, written $s
\sprod f$, as $\lambda_\otimes \circ s \otimes f \circ
\lambda_\otimes^{-1} : X \to Y$. One similarly defines scalar
multiplication on the right, $f \sprod s$, by replacing left unitors
in the above by right unitors.
\end{definition}

\begin{definition}\label{def:model}
  A \emph{model of \SPiLang} consists of a rig category
  $(\cat{C},\otimes,\oplus,O,I)$ equipped with maps
  $\omega : I \to I$ and $\V : I \oplus I \to I \oplus I$ satisfying the
  equations:
  \begin{enumerate}[label={(E\arabic*)}]
    \item $\omega^8 = \id$, \label{ax:omega}
    \item $\V^2 = \sigma_\oplus$, \label{ax:v}
    \item $\V \circ \Sg \circ \V = \omega^2 \sprod \Sg \circ \V \circ \Sg$
      \label{ax:sqy}
  \end{enumerate}
  where $\Sg : I \oplus I \to I \oplus I$ is given by $\Sg = \id \oplus \omega^2$.
\end{definition}

This model is strong enough to express the standard gate set of
Clifford+T. It is not a minimal universal model, however: for
example, the (computationally universal) gate set of Gaussian
Clifford+T only requires a fourth root of unity, \textit{i.e.}, the use 
of $\omega : I \to I$ with $\omega^8 = \id$ can be replaced by $i : I
\to I$ with $i^4 = \id$ while still retaining computational
universality.

\begin{proposition}\label{prop:unitarymodel}
  The rig groupoid \Unitary\ of finite-dimensional Hilbert
  spaces and unitaries is a model of \SPiLang.
\end{proposition}
\begin{proof}
  Choosing $\omega = \exp(i\pi/4)$ and $\V = \Had (\begin{smallmatrix}
      -1 & 0 \\ 0 & i
  \end{smallmatrix}) \Had$ (with $\Had$ the usual Hadamard gate, \textit{i.e.},
  $\Had = \frac{1}{\sqrt{2}} (\begin{smallmatrix}
      1 & 1 \\ 1 & -1
  \end{smallmatrix})$), it is verified by straightforward calculation 
    that the three equations are satisfied.
\end{proof}

We will consider \Unitary\ to be the standard model of \SPiLang. A
semantics of \SPiLang\ can, more generally, be given in any model
satisfying Def.~\ref{def:model} by interpreting all the ``classical''
morphisms as in $\PiLang$, and additionally interpreting the
additional combinators as shown at the bottom of Fig.~\ref{fig:semantics}.

\begin{definition}[Models]
  We use $\sem{-}$ to denote the interpretation of a \SPiLang\ term
  in an arbitrary model of \SPiLang, and $\stsem{-}$ to denote its
  interpretation in the standard model \Unitary.
\end{definition}

In this way, given \SPiLang\ terms $c_1$ and $c_2$, we can only ever
establish $\sem{c_1} = \sem{c_2}$ if this holds from the axioms of
models of \SPiLang\ alone. On the other hand, we can establish
$\stsem{c_1} = \stsem{c_2}$ by any means sound for unitaries (\textit{e.g.},
matrix computation, circuit rewriting rules, ZX-calculus derivations,
etc.).

\begin{figure}
\begin{tabular}{c p{65mm} p{50mm}}
  \textbf{Name} & \textbf{Signature} & \textbf{Meaning} \\ \hline
  $i$ & \small $I \to I$ & $\omega^2$ \\
  $-1$ & \small $I \to I$ & $\omega^4$ \\
  $-i$ & \small $I \to I$ & $\omega^6$ \\
  $\X$ & \small $I \oplus I \to I \oplus I$ & $\sigma_\oplus$ \\
  $P(s)$ & \small $I \oplus I \to I \oplus I$ (for $s : I \to I$) & $\id \oplus s$
  \\
  $\Z$ & \small $I \oplus I \to I \oplus I$ & $P(-1)$ \\
  $\Sg$ & \small $I \oplus I \to I \oplus I$ & $P(i)$ \\
  $\T$ & \small $I \oplus I \to I \oplus I$ & $P(\omega)$ \\
  $\Had$ & \small $I \oplus I \to I \oplus I$ & $\omega \sprod \X \circ \Sg
  \circ \V \circ \Sg \circ \X$ \\
  $\K$ & \small $I \oplus I \to I \oplus I$ & $\omega^{-1} \sprod \Had$ \\
  $\Midswap$ & \small $(A \oplus B) \oplus (C \oplus D) \to
  (A \oplus C) \oplus (B \oplus D)$ & $ \alpha^{-1}_\oplus \circ (\id \oplus \alpha_\oplus) \circ (\id \oplus (\sigma_\oplus \oplus \id)) \circ (\id \oplus \alpha^{-1}_\oplus) \circ \alpha_\oplus$ \\
  $\Mat$ & \small $(I \oplus I) \otimes A \to A \oplus A$ & $\lambda_\otimes
  \oplus \lambda_\otimes \circ \delta_R$ \\
  $\Ctrl~m$ & $(I \oplus I) \otimes A \to (I \oplus I) \otimes A$ given $m : A \to A$ & $\Mat^{-1} \circ (\id \oplus m) \circ \Mat$ \\
  $\nCtrl~m$ & $(I \oplus I) \otimes A \to (I \oplus I) \otimes A$ given $m : A \to A$ & $\Mat^{-1} \circ (m \oplus \id) \circ \Mat$ \\
  $\Swap$ & \small $(I \oplus I) \otimes (I \oplus I) \to (I \oplus I) \otimes (I \oplus I)$ & $\sigma_\otimes$ \\
  $\CX$ & \small $(I \oplus I) \otimes (I \oplus I) \to (I \oplus I) \otimes (I \oplus I)$ & $\Ctrl~\X$ \\
  $\CZ$ & \small $(I \oplus I) \otimes (I \oplus I) \to (I \oplus I) \otimes (I \oplus I)$ & $\Ctrl~\Z$ \\
  $\CCX$ & \small $(I \oplus I) \otimes ((I \oplus I) \otimes (I \oplus I)) \to (I \oplus I) \otimes ((I \oplus I) \otimes (I \oplus I))$ & $\Ctrl~\CX$
\end{tabular}
\caption{\label{fig:abbrev}Shorthands for some maps in models of \SPiLang.}
\end{figure}

\subsection{Representing Quantum Gates}
\label{subsec:prelims}

Let $(\cat{C},\otimes,\oplus,O,I)$ be a model of \SPiLang. We
demonstrate that, in any such model, all the familiar quantum gates
can be represented \emph{internally} as shown in
Fig.~\ref{fig:abbrev}. We can combine these gates into circuits using
the tensor product and composition as usual. For example, the circuit
{\small $$
\Qcircuit @C=.5em @R=.7em {
  & \qw      & \ctrl{1} & \qw      & \qw \\
  & \gate{H} & \targ    & \gate{H} & \qw \\
} 
$$}

\noindent is represented by the morphism
$ \id \otimes \Had \circ \Ctrl~\X \circ \id \otimes \Had$ in a model
of \SPiLang. Besides familiar gates, Fig.~\ref{fig:abbrev} also
defines the convenient map \Mat\ which is so named because it can be
seen as a way to construct maps from \emph{matrix
  representations}. This powerful technique was implicitly used in the
definition of \Ctrl-gates in Sec.~\ref{sub:rig}. More generally, we
think of $g$ as an \emph{abstract block matrix representation} of $f$
when $g \circ \Mat = \Mat \circ f$, as this means in turn that
$\Mat^{-1} \circ g \circ \Mat = f$.

It is straightforward to confirm that the internal gates correspond to
their usual definitions in \Unitary, the standard model of
\SPiLang. Here, we focus on properties that are valid in every model.

We begin by establishing some basic facts about \emph{scalars}
(morphisms $I \to I$) in a rig (or, more generally, monoidal)
category.

\begin{proposition}\label{prop:scalars}
  Let $s$ and $t$ be scalars and $f$ and $g$ be morphisms.
  \begin{enumerate}[label={(\roman*)}]
    \item $s \circ t = t \circ s$,
    \item if $s^2 = t$ then $s^{-1} = t^{-1} \circ s$
    \item $s \sprod f = f \sprod s$
    \item $1 \sprod f = f$,
    \item $s \sprod (t \sprod f) = (s \circ t) \sprod f$,
    \item $s \sprod (f \oplus g) = (s \sprod f) \oplus (s \sprod g)$,
    \item $s \sprod (g \circ f) = (s \sprod g) \circ f$,
    \item $s \sprod (g \circ f) = g \circ (s \sprod f)$.
  \end{enumerate}
\end{proposition}
\begin{proof}
  All but the second property are shown in the literature, \textit{e.g.},
  \cite{heunenvicary:cqt}. For (ii), we see that $t^{-1}
  \circ s \circ s = t^{-1} \circ t = \id_I$ and $s \circ t^{-1} \circ
  s = t^{-1} \circ s \circ s = t^{-1} \circ t = \id_I$ using commutativity
  of scalars, so $s^{-1} = t^{-1} \circ s$ follows by unicity of inverses.
\end{proof}

The next three lemmas establish basic properties of the internal
gates and scalars; the straightforward but tedious proofs are
collected in Appendix~\ref{ap:lem:gates}.

\begin{lemmarep}\label{lem:gates}
  Let $s$ and $t$ be scalars.
  \begin{enumerate}[label={(\roman*)}]
    \item $-1^2 = \id$ and $i^2 = -1$,
    \item $\X^2 = \id$,
    \item $\Pg(s)^2 = \Pg(s^2)$,
    \item $\Pg(s)^{-1} = \Pg(s^{-1})$,
    \item $\Pg(s) \circ \Pg(t) = \Pg(s \circ t) = \Pg(t) \circ \Pg(s)$,
    \item $\Pg(s) \circ \X \circ \Pg(s) = s \sprod \X$,
    \item $\X \circ \V = \V \circ \X$,
    \item $\CX^2 = \id$,
    \item $\CZ^2 = \id$,
    \item $\CCX^2 = \id$,
    \item $\X \circ \Pg(s) = s \sprod \Pg(s^{-1}) \circ \X$.
  \end{enumerate}
\end{lemmarep}
\begin{proof}
  \label{ap:lem:gates}
We consider each property in turn:
\begin{enumerate}[label={(\roman*)}]
\item $i^2 = (\omega^2)^2 = \omega^4 = -1$ and $(-1)^2 =
  (\omega^4)^2 = \omega^8 = \id$ by \ref{ax:omega}.
\item $\X^2
= \sigma_\oplus \circ \sigma_\oplus = \id$ by laws of rig
categories.
\item $\Pg(s)^2= (\id \oplus s) \circ (\id \oplus s)
= (\id \circ \id) \oplus (s \circ s) = \id \oplus s^2 = \Pg(s^2)$ by
bifunctoriality.
\item $\Pg(s) \circ \Pg(s^{-1}) = (\id \oplus s)
\circ (\id \oplus s^{-1}) = (\id \circ \id) \oplus (s \circ s^{-1}) =
\id \oplus \id = \id$ by bifunctoriality, and similarly $\Pg(s^{-1})
\circ \Pg(s) = (\id \circ \id) \oplus (s^{-1} \circ s) = \id \oplus \id
= \id$, so $\Pg(s^{-1}) = \Pg(s)^{-1}$ by unicity of inverses.
\item $\Pg(s) \circ \Pg(t) = (\id \oplus s) \circ (\id \oplus t) = \id \oplus (s \circ t) = \Pg(s \circ t) = \id \oplus (s \circ t) = \id \oplus (t \circ s) = (\id
\oplus t) \circ (\id \oplus s) = \Pg(t) \circ \Pg(s)$ by bifunctoriality
and commutativity of scalars.
\item $\Pg(s) \circ \X \circ \Pg(s) = (\id \oplus s) \circ \sigma_\oplus \circ (\id \oplus s) = (\id \oplus s) \circ (s \oplus \id) \circ \sigma_\oplus = (s \oplus s) \circ \sigma_\oplus = (s \sprod (\id \oplus \id)) \circ \sigma_\oplus = s \sprod ((\id \oplus \id) \circ \sigma_\oplus) = s \sprod \sigma_\oplus = s \sprod \X$ by naturality of $\sigma_\oplus$, bifunctoriality, and Prop.~\ref{prop:scalars}.
\item $\X \circ \V = (\V \circ
\V) \circ \V = \V \circ (\V \circ \V) = \V \circ X$ by $(E2)$.
\item We compute:
\begin{align*}
  \CX^2
  & = \Mat^{-1} \circ (\id \oplus \X) \circ \Mat \circ \Mat^{-1}
  \circ (\id \oplus \X) \circ \Mat \\
  & = \Mat^{-1} \circ (\id \oplus \X) \circ (\id \oplus \X) \circ
  \Mat \\
  & = \Mat^{-1} \circ ((\id \circ \id) \oplus (\X \circ \X)) \circ
  \Mat \\
  & = \Mat^{-1} \circ (\id \oplus \id) \circ \Mat \\
  & = \Mat^{-1} \circ \Mat \\
    & = \id
\end{align*}
\item By analogous argument.
\item By analogous argument.
\item We compute:
\begin{align*}
  \X \circ \Pg(s)
  & = \sigma_\oplus \circ (\id \oplus s) \\
  & = (s \oplus \id) \circ \sigma_\oplus \\
  & = ((s \circ \id) \oplus (s \circ s^{-1})) \circ \sigma_\oplus \\
  & = s \sprod (\id \oplus s^{-1}) \circ \sigma_\oplus \\
  & = s \sprod \Pg(s^{-1}) \circ \X
\end{align*}
\end{enumerate}
\end{proof}

\begin{lemmarep}\label{lem:mat}
  Let $f : X \to Y$, $g : X \to X$, and $h : X \to X$ be maps, and $s$
  and $t$ be scalars. Then:
  \begin{enumerate}[label={(\roman*)}]    
    \item $\Mat \circ (\id_{I \oplus I} \otimes f) = (f \oplus f) \circ \Mat$,
    \item $\Mat \circ \Swap = \Midswap \circ \Mat$,
    \item $\Swap \circ \Mat^{-1} = \Mat^{-1} \circ \Midswap$,
    \item $\Mat \circ (f \otimes \id_{I \oplus I}) = \Midswap \circ (f
      \oplus f) \circ \Midswap \circ \Mat$,
    \item $\Swap \circ \Ctrl~\Pg(s) \circ \Swap = \Ctrl~\Pg(s)$,
    \item $\Ctrl~\Pg(s) \circ \Ctrl~\Pg(t) = \Ctrl~\Pg(t) \circ \Ctrl~\Pg(s)$,
    \item $\Ctrl~\Pg(s) \circ (\id_{I \oplus I} \otimes \Pg(t)) = (\id_{I
      \oplus I} \otimes \Pg(t)) \circ \Ctrl~\Pg(s)$,
    \item $\Mat \circ (\X \otimes \id_{I \oplus I}) = \sigma_\oplus \circ \Mat$,
    \item $\Mat \circ (\Pg(s) \otimes \id_{I \oplus I}) = (\id_{I
      \oplus I} \oplus (s \sprod \id)) \circ \Mat$.
    \item $\Ctrl~g \circ \Ctrl~h = \Ctrl(g \circ h)$
  \end{enumerate}
\end{lemmarep}
\begin{proof}
  \label{ap:lem:mat}
  Below, the word Laplaza followed by a numeral refers to the
  coherence conditions of rig categories, first described in
  \cite{laplaza:distributivity}.
  
We consider each property in turn:  

(i) follows by commutativity of the diagram

\bigskip
\[\begin{tikzcd}
	{(I \oplus I) \otimes X} && {X \oplus X} \\
	& {(I \otimes X) \oplus (I \otimes X)} \\
	& {(I \otimes Y) \oplus (I \otimes Y)} \\
	{(I \oplus I) \otimes Y} && {Y \oplus Y}
	\arrow["\Mat", from=1-1, to=1-3]
	\arrow["{\delta_R}"', from=1-1, to=2-2]
	\arrow["{\lambda_\otimes \oplus \lambda_\otimes}"', from=2-2, to=1-3]
	\arrow["{f \oplus f}", from=1-3, to=4-3]
	\arrow["{\id \otimes f}"', from=1-1, to=4-1]
	\arrow["\Mat"', from=4-1, to=4-3]
	\arrow["{\delta_R}", from=4-1, to=3-2]
	\arrow["{\lambda_\otimes \oplus \lambda_\otimes}", from=3-2, to=4-3]
	\arrow["{(\id \otimes f) \oplus (\id \otimes f)}", from=2-2, to=3-2]
\end{tikzcd}\]

\bigskip
where the left and right cells commute by naturality, and the top and bottom cells by definition. 

(ii) then follows by chasing

\bigskip
\adjustbox{scale=0.7,center}{\begin{tikzcd}
	{(I \oplus I) \otimes (I \oplus I)} & {(I \oplus I) \otimes (I \oplus I)} && {(I \oplus I) \oplus (I \oplus I)} \\
	&& { (I \otimes (I \oplus I)) \oplus (I \otimes (I \oplus I))} \\
	& { (I \otimes (I \oplus I)) \oplus (I \otimes (I \oplus I))} & {((I \otimes I) \oplus (I \otimes I)) \oplus ((I \otimes I) \oplus (I \otimes I))} \\
	&& {((I \otimes I) \oplus (I \otimes I)) \oplus ((I \otimes I) \oplus (I \otimes I))} \\
	{(I \oplus I) \oplus (I \oplus I)} &&& {(I \oplus I) \oplus (I \oplus I)}
	\arrow[""{name=0, anchor=center, inner sep=0}, "\Swap", from=1-1, to=1-2]
	\arrow["{\delta_R}", from=1-2, to=2-3]
	\arrow[""{name=1, anchor=center, inner sep=0}, "{\delta_L \oplus \delta_L}"', from=2-3, to=3-3]
	\arrow["{\lambda_\otimes \oplus \lambda_\otimes}", from=2-3, to=1-4]
	\arrow[""{name=2, anchor=center, inner sep=0}, "{(\lambda_\otimes \oplus \lambda_\otimes) \oplus ((\lambda_\otimes \oplus \lambda_\otimes)}"', from=3-3, to=1-4]
	\arrow[""{name=3, anchor=center, inner sep=0}, "\Midswap"', from=5-4, to=1-4]
	\arrow[""{name=4, anchor=center, inner sep=0}, "\Midswap"', from=4-3, to=3-3]
	\arrow[""{name=5, anchor=center, inner sep=0}, "{(\lambda_\otimes \oplus \lambda_\otimes) \oplus ((\lambda_\otimes \oplus \lambda_\otimes)}"', from=4-3, to=5-4]
	\arrow["{\delta_R}"', from=1-1, to=3-2]
	\arrow[""{name=6, anchor=center, inner sep=0}, "{\delta_L \oplus \delta_L}"', from=3-2, to=4-3]
	\arrow[""{name=7, anchor=center, inner sep=0}, "{\lambda_\otimes \oplus \lambda_\otimes}"', from=3-2, to=5-1]
	\arrow["{=}"', from=5-1, to=5-4]
	\arrow["\Mat"', from=1-1, to=5-1]
	\arrow["\Mat", from=1-2, to=1-4]
	\arrow["{(i)}"', draw=none, from=6, to=0]
	\arrow["{(ii)}"{description}, draw=none, from=7, to=5]
	\arrow["{(iv)}"{description}, draw=none, from=4, to=3]
	\arrow["{(iii)}"{description}, draw=none, from=1, to=2]
\end{tikzcd}}

\bigskip 
where $(i)$ commutes by coherence (Laplaza (II) + (IX)), $(ii)$ and
$(iii)$ by coherence (Laplaza (XXIII)), and $(iv)$ by naturality. But then

(iii) follows by
\begin{align*}
  \Swap \circ \Mat^{-1}
  & = \Swap^{-1} \circ \Mat^{-1} & (\Swap~\text{involutive})\\
  & = (\Mat \circ \Swap)^{-1} & ((-)^{-1}\text{ contravariant functorial})\\
  & = (\Midswap \circ \Mat)^{-1} & \text{Lem.~\ref{lem:mat} (2)}\\
  & = \Mat^{-1} \circ \Midswap^{-1} & ((-)^{-1}\text{ contravariant functorial})\\
  & = \Mat^{-1} \circ \Midswap & (\Midswap\text{ involutive})
\end{align*}

(iv) by
\begin{align*}
  \Mat \circ (f \otimes \id) & = \Mat \circ \Swap \circ (\id \otimes f) \circ \Swap & (\text{naturality }\Swap) \\
  & = \Midswap \circ \Mat \circ (\id \otimes f) \circ \Swap & \text{Lem.~\ref{lem:mat} (2)} \\
  & = \Midswap \circ (f \oplus f) \circ \Mat \circ \Swap & \text{Lem.~\ref{lem:mat} (1)} \\
  & = \Midswap \circ (f \oplus f) \circ \Midswap \circ \Mat & \text{Lem.~\ref{lem:mat} (2)}
\end{align*}

(v) by
\begin{align*}
  & \Swap \circ \Ctrl~\Pg(s) \circ \Swap \\
  & \quad = \Swap \circ \Mat^{-1} \circ (\id \oplus \Pg(s)) \circ \Mat \circ
  \Swap & (\text{def. }\Ctrl) \\
  & \quad = \Mat^{-1} \circ \Midswap \circ (\id \oplus \Pg(s)) \circ
  \Midswap \circ \Mat & (\text{Lem.~\ref{lem:mat} (2)+(3)}) \\
  & \quad = \Mat^{-1} \circ \Midswap \circ
  ((\id \oplus \id) \oplus (\id \oplus s)) \circ
  \Midswap \circ \Mat & (\text{def. }\Pg(s)) \\
  & \quad = \Mat^{-1} \circ \Midswap \circ \Midswap \circ
  ((\id \oplus \id) \oplus (\id \oplus s))
  \circ \Mat & (\text{naturality }\Midswap) \\
  & \quad = \Mat^{-1} \circ
  ((\id \oplus \id) \oplus (\id \oplus s))
  \circ \Mat & (\Midswap\text{ involutive}) \\
  & \quad = \Mat^{-1} \circ
  (\id \oplus \Pg(s))
  \circ \Mat & (\text{def. }\Pg(s)) \\
  & \quad = \Ctrl~\Pg(s) & (\text{def. }\Ctrl)
\end{align*}

(vi) by
\begin{align*}
  \Ctrl~\Pg(s) \circ \Ctrl~\Pg(t)
  & = \Mat^{-1} \circ (\id \oplus \Pg(s)) \circ \Mat \circ \Mat^{-1}
  \circ (\id \oplus \Pg(t)) \circ \Mat & (\text{def.}~\Ctrl)\\
  & = \Mat^{-1} \circ (\id \oplus \Pg(s))
  \circ (\id \oplus \Pg(t)) \circ \Mat & (\Mat~\text{invertible})\\
  & = \Mat^{-1} \circ (\id \oplus (\Pg(s) \circ \Pg(t))) \circ \Mat &
  ({\oplus}~\text{bifunctoriality})\\
  & = \Mat^{-1} \circ (\id \oplus (\Pg(t) \circ \Pg(s))) \circ \Mat &
  (\text{Lem.}~\ref{lem:gates}(v))\\
  & = \Mat^{-1} \circ (\id \oplus \Pg(t))
  \circ (\id \oplus \Pg(s)) \circ \Mat & ({\oplus}~\text{bifunctoriality})\\
  & = \Mat^{-1} \circ (\id \oplus \Pg(t)) \circ \Mat \circ \Mat^{-1}
  \circ (\id \oplus \Pg(s)) \circ \Mat & (\Mat~\text{invertible})\\
  & = \Ctrl~\Pg(t) \circ \Ctrl~\Pg(s) & (\text{def.}~\Ctrl)
\end{align*}

(vii) by
\begin{align*}
  \Ctrl~\Pg(s) \circ (\id \otimes \Pg(t))
  & = \Mat^{-1} \circ (\id \oplus \Pg(s)) \circ \Mat \circ (\id
  \otimes \Pg(t)) & (\text{def.}~\Ctrl)\\
  & = \Mat^{-1} \circ (\id \oplus \Pg(s)) \circ (\Pg(t) \oplus \Pg(t))
  \circ \Mat & (\text{Lem.}~\ref{lem:mat}(1))\\
  & = \Mat^{-1} \circ ((\id \circ \Pg(t)) \oplus (\Pg(s) \circ
  \Pg(t))) \circ \Mat & ({\oplus}~\text{bifunctoriality})\\
  & = \Mat^{-1} \circ ((\Pg(t) \circ \id) \oplus (\Pg(t) \circ
  \Pg(s))) \circ \Mat & (\text{Lem.}~\ref{lem:gates}(v))\\
  & = \Mat^{-1} \circ (\Pg(t) \oplus \Pg(t)) \circ (\id \oplus \Pg(s))
  \circ \Mat & ({\oplus}~\text{bifunctoriality})\\
  & = ((\Pg(t)^{-1} \oplus \Pg(t)^{-1}) \circ \Mat)^{-1} \circ (\id \oplus \Pg(s))
  \circ \Mat & ((-)^{-1}~\text{contrav. funct.})\\
  & = (\Mat \circ (\id \otimes \Pg(t)^{-1}))^{-1} \circ (\id \oplus \Pg(s))
  \circ \Mat & (\text{Lem.}~\ref{lem:mat}(1)) \\
  & = (\id \otimes \Pg(t)) \circ \Mat^{-1} \circ (\id \oplus \Pg(s))
  \circ \Mat & ((-)^{-1}~\text{contrav. funct.})\\
  & = (\id \otimes \Pg(t)) \circ \Ctrl~\Pg(s) & (\text{def.}~\Ctrl)
\end{align*}

(viii) by commutativity of the diagram

\adjustbox{scale=0.8,center}{\begin{tikzcd}
	{(I \oplus I) \otimes (I \oplus I)} &&& {(I \oplus I) \otimes (I \oplus I)} \\
	\\
	& {(I \otimes (I \oplus I)) \oplus (I \otimes (I \oplus I))} & {(I \otimes (I \oplus I)) \oplus (I \otimes (I \oplus I))} \\
	\\
	{(I \oplus I) \oplus (I \oplus I)} &&& {(I \oplus I) \oplus (I \oplus I)}
	\arrow[""{name=0, anchor=center, inner sep=0}, "{\delta_R}", from=1-1, to=3-2]
	\arrow["{\lambda_\otimes \oplus \lambda_\otimes}", from=3-2, to=5-1]
	\arrow[""{name=1, anchor=center, inner sep=0}, "\Mat"', from=1-1, to=5-1]
	\arrow[""{name=2, anchor=center, inner sep=0}, "{\sigma_\oplus}"', from=5-1, to=5-4]
	\arrow[""{name=3, anchor=center, inner sep=0}, "{\delta_R}"', from=1-4, to=3-3]
	\arrow["{\lambda_\otimes \oplus \lambda_\otimes}"', from=3-3, to=5-4]
	\arrow[""{name=4, anchor=center, inner sep=0}, "\Mat", from=1-4, to=5-4]
	\arrow[""{name=5, anchor=center, inner sep=0}, "{\sigma_\oplus}", from=3-2, to=3-3]
	\arrow[""{name=6, anchor=center, inner sep=0}, "{\X \otimes \id}", from=1-1, to=1-4]
	\arrow["{(i)}", draw=none, from=6, to=5]
	\arrow["{(ii)}", draw=none, from=5, to=2]
	\arrow["{(iii)}", draw=none, from=0, to=1]
	\arrow["{(iii)}"', draw=none, from=3, to=4]
\end{tikzcd}}

\bigskip
where $(i)$ commutes by Laplaza (I)+(II) (recalling that $\X$ is just
defined to be $\sigma_\oplus$ on $I \oplus I$), $(ii)$ by naturality
of $\sigma_\oplus$, and $(iii)$ by definition of $\Mat$.

(ix) follows by
\begin{align*}
  \Mat \circ (\Pg(s) \otimes \id)
  & = \Midswap \circ (\Pg(s) \oplus \Pg(s)) \circ \Midswap \circ \Mat
  & (\text{Lem.~\ref{lem:mat} (4)}) \\
  & = \Midswap \circ ((\id \oplus s) \oplus (\id \oplus s)) \circ
  \Midswap \circ \Mat & (\text{def.}~\Pg(s)) \\
  & = \Midswap \circ \Midswap \circ ((\id \oplus \id) \oplus (s \oplus
  s)) \circ \Mat & (\text{naturality}~\Midswap) \\
  & = ((\id \oplus \id) \oplus (s \oplus s)) \circ \Mat &
  (\Midswap~\text{involutive}) \\
  & = ((\id \oplus \id) \oplus ((s \sprod \id) \oplus (s \sprod \id)) \circ \Mat &
  (\text{Prop.~\ref{prop:scalars}}) \\
  & = ((\id \oplus \id) \oplus (s \sprod (\id \oplus \id)) \circ \Mat &
  (\text{Prop.~\ref{prop:scalars}}) \\
  & = (\id \oplus (s \sprod \id)) \circ \Mat &
  (\text{bifunctoriality}~{\oplus}) \\
\end{align*}

(x) follows by
\begin{align*}
  \Ctrl~g \circ \Ctrl~h
  & = \Mat^{-1} \circ \id \oplus g \circ \Mat \circ \Mat^{-1} \circ
  \id \oplus h \circ \Mat & (\text{def.}~\Ctrl) \\
  & = \Mat^{-1} \circ \id \oplus g \circ \id \oplus h \circ \Mat &
  (\Mat~\text{invertible}) \\
  & = \Mat^{-1} \circ \id \oplus (g \circ h) \circ \Mat &
  (\text{bifunctoriality}~\oplus) \\
  & = \Mat^{-1} \circ \id \oplus (g \circ h) \circ \Mat &
  (\text{def.}~\Ctrl)
\end{align*}
\end{proof}

\begin{lemmarep}\label{lem:had}
  Any model of \SPiLang\ satisfies $\Had \circ \X \circ \Had = \Z$
  and $\Had \circ \Z \circ \Had = \X$. 
\end{lemmarep}
\begin{proof}
\label{ap:lem:had}
\begin{align*}
  \Had \circ \X \circ \Had
  & = (\omega \sprod \X \circ \Sg
  \circ \V \circ \Sg \circ \X) \circ \X \circ (\omega \sprod \X \circ \Sg
  \circ \V \circ \Sg \circ \X) & (\text{def.}~\Had) \\
  & = \omega^2 \sprod (\X \circ \Sg
  \circ \V \circ \Sg \circ \X \circ \X \circ \X \circ \Sg
  \circ \V \circ \Sg \circ \X) & (\text{Prop.}~\ref{prop:scalars}) \\
  & = i \sprod (\X \circ \Sg
  \circ \V \circ \Sg \circ \X \circ \Sg
  \circ \V \circ \Sg \circ \X) & (X^2 = \id, \omega^2 = i) \\
  & = i \sprod (\X \circ \Sg
  \circ \V \circ (i \sprod \X)
  \circ \V \circ \Sg \circ \X) & (\text{Prop.}~\ref{prop:scalars}) \\
  & = i^2 \sprod (\X \circ \Sg \circ \V \circ \X \circ \V \circ \Sg
  \circ \X) & (\text{Prop.}~\ref{prop:scalars}) \\
  & = -1 \sprod (\X \circ \Sg \circ \X \circ \V \circ \V \circ \Sg
  \circ \X) & (\text{Lem.}~\ref{lem:gates}, i^2 = -1) \\
  & = -1 \sprod (\X \circ \Sg \circ \X \circ \X \circ \Sg
  \circ \X) & (\V^2 = \X) \\
  & = -1 \sprod (\X \circ \Sg \circ \Sg \circ \X) & (\X^2 = \id) \\
  & = -1 \sprod (\X \circ \Z \circ \X) & (\Sg^2 = \Z) \\
  & = -1 \sprod ((-1 \sprod \Z \circ \X) \circ \X) &
  (\text{Lem.}~\ref{lem:gates}) \\
  & = (-1)^2 \sprod \Z \circ \X \circ \X &
  (\text{Prop.}~\ref{prop:scalars}) \\
  & = \Z & ((-1)^2 = \id, \X^2 = \id)
\end{align*}
\begin{align*}
  \Had \circ \Z \circ \Had
  & = (\omega \sprod \X \circ \Sg \circ \V \circ \Sg \circ \X) \circ
  \Z \circ (\omega \sprod \X \circ \Sg \circ \V \circ \Sg \circ \X) &
  (\text{def.}~\Had) \\
  & = \omega^2 \sprod (\X \circ \Sg \circ \V \circ \Sg \circ \X \circ
  \Z \circ \X \circ \Sg \circ \V \circ \Sg \circ \X) &
  (\text{Prop.}~\ref{prop:scalars}) \\
  & = i \sprod (\X \circ \Sg \circ \V \circ \Sg \circ (-1 \sprod \Z
  \circ \X) \circ \X \circ \Sg \circ \V \circ \Sg \circ \X) &
  (\text{Lem.}~\ref{lem:gates}, \omega^2 = i) \\
  & = -i \sprod (\X \circ \Sg \circ \V \circ \Sg \circ \Z \circ \X
  \circ \X \circ \Sg \circ \V \circ \Sg \circ \X) &
  (\text{Prop.}~\ref{prop:scalars}) \\
  & = -i \sprod (\X \circ \Sg \circ \V \circ \Sg \circ \Z \circ \Sg
  \circ \V \circ \Sg \circ \X) & (\X^2 = \id) \\
  & = -i \sprod (\X \circ \Sg \circ \V \circ \Z \circ \Sg \circ \Sg
  \circ \V \circ \Sg \circ \X) & (\text{Lem.}~\ref{lem:gates}) \\
  & = -i \sprod (\X \circ \Sg \circ \V \circ \Z \circ \Z
  \circ \V \circ \Sg \circ \X) & (\Sg^2 = \Z) \\
  & = -i \sprod (\X \circ \Sg \circ \V \circ \V \circ \Sg \circ \X) &
  (\Z^2 = \id) \\
  & = -i \sprod (\X \circ \Sg \circ \X \circ \Sg \circ \X) &
  (\V^2 = \X) \\
  & = -i \sprod (\X \circ (i \sprod \X) \circ \X) &
  (\text{Lem.}~\ref{lem:gates}) \\
  & = -i \circ i \sprod (\X \circ \X \circ \X) &
  (\text{Prop.}~\ref{prop:scalars}) \\
  & = \X & (-i \circ i = \id, \X^2 = \id)
\end{align*}
\end{proof}

\section{Soundness and Completeness}
\label{sec:sc}

We present our main technical development:
\SPiLang\ is \emph{equationally sound and
  complete} for a variety of gate sets, including the computationally
universal \emph{Gaussian
  Clifford+T}~\cite{Amy2020numbertheoretic}. This is expressed
in terms of a series of \emph{full abstraction} results, showing that
fragments of \SPiLang\ are fully abstract for certain classes of
unitaries.

To our knowledge, this is the first presentation of a computationally
universal quantum programming language with a sound and complete
equational theory.


\subsection{$\leq 2$-qubit Clifford Circuits}
\label{subsec:2clifford}
\begin{figure}
  \tiny\[
  \begin{array}{rcll rcll}
    \text{\normalsize $\omega \cdot A$} &=& \text{\normalsize $A \cdot \omega$}
    & \text{\normalsize (A1)} &
    \text{\normalsize $A_0 B_1$} &=& \text{\normalsize $A_1 B_0$}
    & \text{\normalsize (A2)} \\
    \text{\normalsize $\omega^8$} &=& \text{\normalsize $\id$}
    & \text{\normalsize (A3)} &
    \text{\normalsize $\Had^2$} &=& \text{\normalsize $\id$}
    & \text{\normalsize (A4)} \\
    \text{\normalsize $\Sg^4$} &=& \text{\normalsize $\id$}
    & \text{\normalsize (A5)} &
    \text{\normalsize $\Sg \Had \Sg \Had \Sg \Had$} &=&
    \text{\normalsize $\omega \cdot \id$}
    & \text{\normalsize (A6)}\\ \\
      \Qcircuit @C=.7em @R=1.2em {
        & \ctrl{1}    & \ctrl{1}    & \qw \\
        & \control\qw & \control\qw & \qw \\
      } &\raisebox{-1.5mm}{=}&
      \Qcircuit @C=.7em @R=1.2em {
        & \qw & \qw & \qw \\
        & \qw & \qw & \qw \\
      } & \text{\raisebox{-2mm}{\normalsize (A7)}} &
      \Qcircuit @C=.7em @R=1.2em {
        & \gate{S} & \ctrl{1}    & \qw \\
        & \qw      & \control\qw & \qw \\
      } &\raisebox{-2mm}{=}&
      \Qcircuit @C=.7em @R=1.2em {
        & \ctrl{1}    & \gate{S} & \qw \\
        & \control\qw & \qw      & \qw \\
      } & \text{\raisebox{-2mm}{\normalsize (A8)}} \\ \\
      \Qcircuit @C=.7em @R=1.2em {
        & \qw      & \ctrl{1}    & \qw \\
        & \gate{S} & \control\qw & \qw \\
      } &\raisebox{-2mm}{=}&
      \Qcircuit @C=.7em @R=1.2em {
        & \ctrl{1}    & \qw     & \qw \\
        & \control\qw & \gate{S}& \qw \\
      } & \text{\raisebox{-2mm}{\normalsize (A9)}} &
      \Qcircuit @C=.7em @R=1.2em {
        & \gate{H} & \gate{S} & \gate{S} & \gate{H} & \control\qw & \qw \\
        & \qw      & \qw      & \qw      & \qw      & \ctrl{-1}   & \qw
      } &\raisebox{-2mm}{=}&
      \Qcircuit @C=.7em @R=0.7em {
        & \control\qw  & \gate{H} & \gate{S} & \gate{S} & \gate{H} & \qw \\
        & \ctrl{-1}    & \gate{S} & \gate{S} & \qw      & \qw      & \qw
      } & \text{\raisebox{-2mm}{\normalsize (A10)}} \\ \\
      \Qcircuit @C=.7em @R=1.2em {
        & \qw      & \qw      & \qw      & \qw      & \ctrl{1}    & \qw \\
        & \gate{H} & \gate{S} & \gate{S} & \gate{H} & \control\qw & \qw
      } &\raisebox{-2mm}{=}&
      \Qcircuit @C=.7em @R=0.7em {
        & \ctrl{1}    & \gate{S} & \gate{S} & \qw      & \qw      & \qw \\
        & \control\qw & \gate{H} & \gate{S} & \gate{S} & \gate{H} & \qw
      } & \text{\raisebox{-2mm}{\normalsize (A11)}} &
      \Qcircuit @C=.7em @R=1.2em {
        & \ctrl{1}    & \qw      & \ctrl{1}    & \qw \\
        & \control\qw & \gate{H} & \control\qw & \qw
      } &\raisebox{-2mm}{=}&
      \Qcircuit @C=.7em @R=.7em {
        & \qw      & \qw      & \control\qw & \gate{S} & \qw      & \qw
        & \qw \\
        & \gate{S} & \gate{H} & \ctrl{-1}   & \gate{S} & \gate{H} & \gate{S}
        & \qw
      } \text{\normalsize \raisebox{-2mm}{$\phantom{l} \cdot \omega^{-1}$}}
      & \text{\raisebox{-2mm}{\normalsize (A12)}} \\ \\ & & & &
      \Qcircuit @C=.7em @R=1.2em {
        & \ctrl{1}    & \gate{H} & \ctrl{1}    & \qw \\
        & \control\qw & \qw      & \control\qw & \qw
      } &\raisebox{-2mm}{=}&
      \Qcircuit @C=.7em @R=.7em {
        & \gate{S} & \gate{H} & \ctrl{1}    & \gate{S} & \gate{H} & \gate{S}
        & \qw \\
        & \qw      & \qw      & \control\qw & \gate{S} & \qw      & \qw
        & \qw
      } \text{\normalsize \raisebox{-2mm}{$\phantom{l} \cdot \omega^{-1}$}}
      & \text{\raisebox{-2mm}{\normalsize (A13)}}
  \end{array}
  \]
  \caption{A sound and complete equational theory of $\le 2$-qubit
    Clifford circuits due to \citet{selinger:clifford}. What we call
    (A3)--(A13) refer to relations (C1)--(C11) in the original paper
    by \citet{selinger:clifford} (equations (A1) and (A2) become
    relevant once we consider $\leq 2$-qubit Clifford+T circuits~
    \cite{bianselinger:cliffordt}). Note that we swap the order of
    (A12) and (A13) compared to the original presentation by
    \citet{selinger:clifford}.}
  \label{fig:2qubitclifford}
\end{figure}
We begin by proving that models of \SPiLang\ satisfy the sound and
complete equational theory of $\leq 2$-qubit Clifford circuits shown
in Fig.~\ref{fig:2qubitclifford}.
Clifford circuits are those which
can be formed using the gates $\{\CZ,\Sg,\Had\}$ and the scalar
$\omega = e^{i \pi/4}$.
\begin{definition}
In a model of \SPiLang, a \emph{representation of a Clifford circuit}
is any morphism which can be written in terms of morphisms from the
sets $\{\omega, \Sg, \Had, \CZ\}$ and $\{\alpha_\otimes,
\alpha_\otimes^{-1}, \lambda_\otimes, \lambda_\otimes^{-1},
\rho_\otimes, \rho_\otimes^{-1}, \sigma_\otimes\}$, composed
arbitrarily in parallel (using $\otimes$) and in sequence (using
$\circ$). A representation of a $\le 2$-qubit Clifford circuit is one
with signature $I \oplus I \to I \oplus I$ or $(I \oplus I) \otimes (I
\oplus I) \to (I \oplus I) \otimes (I \oplus I)$.
\end{definition}
Note that this definition permits both scalar multiplication by powers
of $\omega$ (since this is formulated using the coherence
isomorphisms) and use of the $\Swap$ gate (since this is precisely
$\sigma_\otimes$).  This result relies on the generators and relations
for Clifford circuits due to \citet{selinger:clifford}, which we prove
are all satisfied in any model of \SPiLang:

\begin{enumerate}[label={(A\arabic*)}]
\item $\omega \sprod f = f \sprod \omega$ for all $f$ follows by
  Prop.~\ref{prop:scalars} (iii). \label{eq:a1}
\item That $(f \otimes \id) \circ (\id \otimes g) =
  (\id \otimes g) \circ (f \otimes \id)$ follows by bifunctoriality
  of $\otimes$. \label{eq:a2}
\item $\omega^8 = \id$ follows immediately by \ref{ax:omega}. \label{eq:a3}
\item We derive
\begin{align*}
  \Had \circ \Had
  & = (\omega \sprod \X \circ \Sg \circ \V \circ \Sg \circ \X) \circ
  (\omega \sprod \X \circ \Sg \circ \V \circ \Sg \circ \X) & (\text{def. }\Had) \\
  & = \omega^2 \sprod \X \circ \Sg \circ \V \circ \Sg \circ \X \circ
  \X \circ \Sg \circ \V \circ \Sg \circ \X & (\text{Prop. \ref{prop:scalars}})\\
  & = \omega^2 \sprod \X \circ \Sg \circ \V \circ \Sg \circ \Sg \circ
  \V \circ \Sg \circ \X & (\X^2 = \id)\\
  & = \omega^2 \sprod \X \circ (\omega^{-2} \sprod \V \circ \Sg \circ
  \V) \circ (\omega^{-2} \sprod \V \circ \Sg \circ \V) \circ \X & \ref{ax:sqy} \\
  & = \omega^{-2} \sprod \X \circ \V \circ \Sg \circ \V \circ \V \circ
  \Sg \circ \V \circ \X & (\text{Prop. \ref{prop:scalars}}) \\
  & = \omega^{-2} \sprod \X \circ \V \circ \Sg \circ \X \circ \Sg
  \circ \V \circ \X & \ref{ax:v}\\
  & = \omega^{-2} \sprod \X \circ \V \circ (\omega^2 \sprod \X)
  \circ \V \circ \X & (\text{Lem.~\ref{lem:gates} (vi)}) \\
  & = \X \circ \V \circ \X \circ \V \circ \X & (\text{Prop. \ref{prop:scalars}})\\
  & = \X \circ \X \circ \V \circ \V \circ \X & (\text{Lem.~\ref{lem:gates} (vii)}) \\
  & = \X \circ \X \circ \X \circ \X & \ref{ax:v} \\
  & = \id & (\X^2 = \id)
\end{align*} \label{eq:a4}
\item $\Sg^4 = (\id \oplus i)^4 = (\id \oplus \omega^2)^4 = \id^4 \oplus \omega^8 = \id \oplus \id = \id$ by bifunctoriality and \ref{ax:omega}. \label{eq:a5}
\item We compute
  \begin{align*}
    (\Sg \circ \Had)^3
    & = (\Sg \circ (\omega \sprod \X \circ \Sg \circ \V \circ \Sg
    \circ \X))^3 & (\text{def. }\Had) \\
    & = (\omega \sprod \Sg \circ \X \circ \Sg \circ \V \circ \Sg
    \circ \X)^3 & (\text{Prop. \ref{prop:scalars}}) \\
    & = (\omega \sprod (\omega^2 \sprod \X) \circ \V \circ \Sg
    \circ \X)^3 & (\text{Lem.~\ref{lem:gates} (vi)}) \\
    & = (\omega^3 \sprod \X \circ \V \circ \Sg
    \circ \X)^3 & (\text{Prop. \ref{prop:scalars}}) \\
    & = (\omega^3 \sprod \X \circ \V \circ \Sg
    \circ \X) \circ (\omega^3 \sprod \X \circ \V \circ \Sg
    \circ \X) \circ (\omega^3 \sprod \X \circ \V \circ \Sg
    \circ \X)  & (\text{expand}) \\
    & = \omega^9 \sprod \X \circ \V \circ \Sg \circ \X \circ \X \circ
    \V \circ \Sg \circ \X \circ \X \circ \V \circ \Sg \circ \X &
    (\text{Prop. \ref{prop:scalars}}) \\
    & = \omega \sprod \X \circ \V \circ \Sg \circ
    \V \circ \Sg \circ \V \circ \Sg \circ \X & (\ref{ax:omega}, \X^2 = \id) \\
    & = \omega \sprod \X \circ (\omega^2 \sprod \Sg \circ \V \circ
    \Sg) \circ \Sg \circ \V \circ \Sg \circ \X & \ref{ax:sqy} \\
    & = \omega^3 \sprod \X \circ \Sg \circ \V \circ
    \Sg \circ \Sg \circ \V \circ \Sg \circ \X & (\text{Prop. \ref{prop:scalars}})
    \\
    & = \omega^3 \sprod \X \circ \Sg \circ \V \circ
    \Sg \circ \X \circ \X \circ \Sg \circ \V \circ \Sg \circ \X & (\X^2 = \id)
    \\
    & = \omega \sprod (\omega \sprod \X \circ \Sg \circ \V \circ
    \Sg \circ \X) \circ (\omega \sprod \X \circ \Sg \circ \V \circ \Sg \circ \X) &
    (\text{Prop. \ref{prop:scalars}}) \\
    & = \omega \sprod (\Had \circ \Had) & (\text{def. }\Had) \\
    & = \omega \sprod \id & \ref{eq:a4}
  \end{align*} \label{eq:a6}
\item By Lem.~\ref{lem:gates} (ix). \label{eq:a7}
\item We have
  \begin{align*}
    \Ctrl~\Z \circ (\Sg \otimes \id)
    & = \Swap \circ \Ctrl~\Z \circ \Swap \circ (\Sg \otimes \id) &
    \text{(Lem.~\ref{lem:mat} (v))} \\
    & = \Swap \circ \Ctrl~\Z \circ (\id \otimes \Sg) \circ \Swap &
    (\text{naturality}~\Swap) \\
    & = \Swap \circ (\id \otimes \Sg) \circ \Ctrl~\Z \circ \Swap &
    (\text{Lem.}~\ref{lem:mat} (vii)) \\
    & = (\Sg \otimes \id) \circ \Swap \circ \Ctrl~\Z \circ \Swap &
    (\text{naturality}~\Swap) \\
    & = (\Sg \otimes \id) \circ \Ctrl~\Z & \text{(Lem.~\ref{lem:mat} (v))}
  \end{align*} \label{eq:a8}
\item By Lem.~\ref{lem:mat} (v). \label{eq:a9}

\item Since $\Sg \circ \Sg = \Z$ and $\Had \circ \Sg \circ \Sg \circ
  \Had = \Had \circ \Z \circ \Had = \X$ by Lems.~\ref{lem:gates} and
  \ref{lem:had}, it suffices to show
  $\Ctrl~\Z \circ (\X \otimes \id) = \X \otimes \Z \circ \Ctrl~\Z$.
  This follows by
  \begin{align*}
    \Ctrl~\Z \circ (\X \otimes \id)
    & = \Mat^{-1} \circ (\id \oplus \Z) \circ \Mat \circ (\X \otimes
    \id) & (\text{def.}~\Ctrl) \\
    & = \Mat^{-1} \circ (\id \oplus \Z) \circ \sigma_\oplus \circ \Mat
    & (\text{Lem.}~\ref{lem:mat} (viii)) \\
    & = \Mat^{-1} \circ \sigma_\oplus \circ (\Z \oplus \id) \circ \Mat
    & (\text{naturality}~\sigma_\oplus) \\
    & = (\X \otimes \id) \circ \Mat^{-1} \circ (\Z \oplus \id) \circ \Mat
    & (\text{Lem.}~\ref{lem:mat} (viii)) \\
    & = (\X \otimes \id) \circ \Mat^{-1} \circ (\Z \oplus (\Z \circ
    \Z)) \circ \Mat & (\Z^2 = \id) \\
    & = (\X \otimes \id) \circ \Mat^{-1} \circ (\Z \oplus \Z) \circ
    (\id \oplus \Z) \circ \Mat & (\text{bifunctoriality}~{\oplus}) \\
    & = (\X \otimes \id) \circ (\id \otimes \Z) \circ \Mat^{-1} \circ
    (\id \oplus \Z) \circ \Mat & (\text{Lem.}~\ref{lem:mat} (i)) \\
    & = (\X \otimes \Z) \circ \Mat^{-1} \circ
    (\id \oplus \Z) \circ \Mat & (\text{bifunctoriality}~{\otimes}) \\
    & = \X \otimes \Z \circ \Ctrl~\Z & (\text{def.}~\Ctrl)
  \end{align*} \label{eq:a10}

\item Similarly, since it has already been established that $\Had
  \circ \Sg \circ \Sg \circ \Had = \X$ and $\Sg \circ \Sg = \Z$, it
  suffices to show
  $\Ctrl~\Z \circ (\id \otimes \X) = \Z \otimes \X \circ \Ctrl~\Z$:
  \begin{align*}
    \Ctrl~\Z \circ (\id \otimes \X)
    & = \Swap \circ \Ctrl~\Z \circ \Swap \circ (\id \otimes \X) &
    (\text{Lem.}~\ref{lem:gates} (v)) \\
    & = \Swap \circ \Ctrl~\Z \circ (\X \otimes \id) \circ \Swap &
    (\text{naturality}~\Swap) \\
    & = \Swap \circ \X \otimes \Z \circ \Ctrl~\Z \circ \Swap &
    \ref{eq:a10} \\
    & = \Z \otimes \X \circ \Swap \circ \Ctrl~\Z \circ \Swap &
    (\text{naturality}~\Swap) \\
    & = \Z \otimes \X \circ \Ctrl~\Z &
    (\text{Lem.}~\ref{lem:gates} (v))
  \end{align*} \label{eq:a11}

\item\label{item:a12} We defer the derivation of this identity to Appendix~\ref{eq:a12}.
  \begin{appendixproof}[Proof of ~\ref{item:a12}]
\label{eq:a12}
  We start by showing some identities that will be helpful in showing this
  relation and the one that follows. We first observe that
\begin{align*}
  \Sg \circ \Had \circ \Sg \circ \Had \circ \Sg
  & = \Sg \circ \Had \circ \Sg \circ \Had \circ \Sg \circ \Had \circ
  \Had & \ref{eq:a4} \\
  & = (\omega \sprod \id) \circ \Had & \ref{eq:a6} \\
  & = \omega \sprod \Had & \text{(Prop.~\ref{prop:scalars})}
\end{align*}
and that
\begin{align*}
  i \sprod \Sg \circ \Had \circ \Sg \circ \Z \circ \Had \circ \Sg
  & = i \sprod \Sg \circ \Had \circ \Sg \circ \Had \circ \Had \circ \Z
  \circ \Had \circ \Sg & \ref{eq:a4} \\
  & = i \sprod \Sg \circ \Had \circ \Sg \circ \Had \circ \X \circ \Sg
  & (\text{Lem.~\ref{lem:had}}) \\
  & = i \sprod \Sg \circ \Had \circ \Sg \circ \Had \circ (i \sprod \Sg
  \circ \Z \circ X) & (\text{Lem.~\ref{lem:gates} (xi)}) \\
  & = i^2 \sprod \Sg \circ \Had \circ \Sg \circ \Had \circ \Sg
  \circ \Z \circ X & (\text{Prop.~\ref{prop:scalars}}) \\
  & = -1 \circ \omega \sprod \Had \circ \Z \circ \X & (\text{Lem.}, i^2 = -1) \\
  & = -1 \circ \omega \sprod \Had \circ (-1 \sprod \X \circ \Z) &
  (\text{Lem.~\ref{lem:gates} (xi)}) \\
  & = (-1)^2 \circ \omega \sprod \Had \circ \X \circ \Z &
  (\text{Prop.~\ref{prop:scalars}}) \\
  & = \omega \sprod \Had \circ \X \circ \Had \circ \Had \circ \Z &
  (\ref{eq:a4}, (-1)^2 = \id) \\
  & = \omega \sprod \Z \circ \Had \circ \Z & (\text{Lem.~\ref{lem:had}})
\end{align*}
But then we have
{\scriptsize\begin{align*}
  & \omega^{-1} \sprod (\Sg \otimes (\Sg \circ \Had \circ \Sg)) \circ
  \Ctrl~\Z \circ (\id \otimes (\Had \circ \Sg)) \\
  & \quad = \omega^{-1} \sprod (\Sg \otimes (\Sg \circ \Had \circ \Sg))
  \circ \Mat^{-1} \circ (\id \oplus \Z) \circ \Mat \circ (\id \otimes
  (\Had \circ \Sg)) & (\text{def.}~\Ctrl) \\
  & \quad = \omega^{-1} \sprod (\Sg \otimes \id) \circ (\id \otimes
  (\Sg \circ \Had \circ \Sg)) \circ \Mat^{-1} \circ (\id \oplus \Z)
  \circ \Mat \circ (\id \otimes (\Had \circ \Sg)) &
  (\text{bifunctoriality}~\otimes) \\
  & \quad = \omega^{-1} \sprod (\Sg \otimes \id) \circ \Mat^{-1} \circ
  ((\Sg \circ \Had \circ \Sg) \oplus (\Sg \circ \Had \circ \Sg)) \circ
  (\id \oplus \Z) \circ ((\Had \circ \Sg) \oplus (\Had \circ \Sg)) \circ
  \Mat & (\text{Lem.~\ref{lem:mat} (i) twice}) \\
  & \quad = \omega^{-1} \sprod (\Sg \otimes \id) \circ \Mat^{-1} \circ
  ((\Sg \circ \Had \circ \Sg \circ \Had \circ \Sg) \oplus (\Sg \circ
  \Had \circ \Sg \circ \Z \circ \Had \circ \Sg)) \circ \Mat &
  (\text{bifunctoriality}~\oplus) \\
  & \quad = \omega^{-1} \sprod \Mat^{-1} \circ (\id \oplus (i \sprod \id)) \circ
  ((\Sg \circ \Had \circ \Sg \circ \Had \circ \Sg) \oplus (\Sg \circ
  \Had \circ \Sg \circ \Z \circ \Had \circ \Sg)) \circ \Mat &
  (\text{Lem.~\ref{lem:mat} (ix)}) \\
  & \quad = \omega^{-1} \sprod \Mat^{-1} \circ ((\Sg \circ \Had \circ
  \Sg \circ \Had \circ \Sg) \oplus (i \sprod \Sg \circ \Had \circ \Sg
  \circ \Z \circ \Had \circ \Sg)) \circ \Mat
  &  (\text{bifunctoriality}~\oplus) \\
  & \quad = \omega^{-1} \sprod \Mat^{-1} \circ ((\omega \sprod \Had)
  \oplus (\omega \sprod \Z \circ \Had \circ \Z)) \circ \Mat
  & (\text{inner lemmas}) \\
  & \quad = \omega^{-1} \sprod \Mat^{-1} \circ \omega \sprod (\Had
  \oplus (\Z \circ \Had \circ \Z)) \circ \Mat &
  (\text{Prop.~\ref{prop:scalars}}) \\
  & \quad = \omega^{-1} \circ \omega \sprod \Mat^{-1} \circ (\Had \oplus
  (\Z \circ \Had \circ \Z)) \circ \Mat &
  (\text{Prop.~\ref{prop:scalars}}) \\
  & \quad = \Mat^{-1} \circ (\id \oplus \Z) \circ (\Had \oplus \Had)
  \circ (\id \oplus \Z) \circ \Mat
  & (\text{bifunctoriality}~\oplus) \\
  & \quad = \Mat^{-1} \circ (\id \oplus \Z) \circ \Mat \circ
  \Mat^{-1} \circ (\Had \oplus \Had) \circ \Mat \circ \Mat^{-1} \circ
  (\id \oplus \Z) \circ \Mat & (\Mat~\text{invertible}) \\
  & \quad = \Mat^{-1} \circ (\id \oplus \Z) \circ \Mat \circ
  \Mat^{-1} \circ \Mat \circ (\id \otimes \Had) \circ \Mat^{-1} \circ
  (\id \oplus \Z) \circ \Mat & (\text{Lem.~\ref{lem:mat} (i)}) \\
  & \quad = \Mat^{-1} \circ (\id \oplus \Z) \circ \Mat \circ (\id
  \otimes \Had) \circ \Mat^{-1} \circ (\id \oplus \Z) \circ \Mat &
  (\Mat~\text{invertible})\\
  & \quad = \Ctrl~\Z \circ (\id \otimes \Had) \circ \Ctrl~\Z
  & (\text{def.}~\Ctrl)
\end{align*}}
as desired.
\end{appendixproof}
\item This relation follows by the above since
\begin{align*}
  & \omega^{-1} \sprod ((\Sg \circ \Had \circ \Sg) \otimes \Sg) \circ
  \Ctrl~\Z \circ ((\Had \circ \Sg) \otimes \id) \\
  & \quad = \omega^{-1} \sprod ((\Sg \circ \Had \circ \Sg) \otimes \Sg) \circ
  \Swap \circ \Ctrl~\Z \circ \Swap \circ ((\Had \circ \Sg) \otimes \id)
  & (\text{Lem.~\ref{lem:mat} (v)}) \\
  & \quad = \omega^{-1} \sprod \Swap \circ (\Sg \otimes (\Sg \circ
  \Had \circ \Sg)) \circ \Ctrl~\Z \circ (\id \otimes (\Had \circ \Sg))
  \circ \Swap
  & (\text{naturality}~\Swap) \\
  & \quad = \Swap \circ (\omega^{-1} \sprod ((\Sg \otimes (\Sg \circ
  \Had \circ \Sg)) \circ \Ctrl~\Z \circ (\id \otimes (\Had \circ \Sg)))
  \circ \Swap
  & (\text{Prop.~\ref{prop:scalars}})\\
  & \quad = \Swap \circ \Ctrl~\Z \circ (\id
  \otimes \Had) \circ \Ctrl~\Z \circ \Swap
  & \ref{eq:a12} \\
  & \quad = \Swap \circ \Ctrl~\Z \circ \Swap \circ \Swap \circ (\id
  \otimes \Had) \circ \Ctrl~\Z \circ \Swap
  & (\Swap~\text{involutive}) \\
  & \quad = \Swap \circ \Ctrl~\Z \circ \Swap \circ (\Had
  \otimes \id) \circ \Swap \circ \Ctrl~\Z \circ \Swap
  & (\text{naturality}~\Swap) \\
  & \quad = \Ctrl~\Z \circ (\Had \otimes \id) \circ \Ctrl~\Z
  & (\text{Lem.~\ref{lem:mat} (v)})
\end{align*}\label{eq:a13}
\end{enumerate}
These derivations lead us, as a first step, to full abstraction for
$\le 2$-qubit Clifford circuits.
\begin{theorem}[Full abstraction for $\le 2$-qubit Clifford]\label{thm:2cliff}
  Let $c_1$ and $c_2$ be \SPiLang\ terms representing Clifford circuits of
  at most two qubits. Then $\sem{c_1} = \sem{c_2}$ iff $\stsem{c_1} =
  \stsem{c_2}$.
\end{theorem}
\begin{proof}
  The identities \ref{eq:a3}--\ref{eq:a13} are complete for $\le
  2$-qubit Clifford circuits by \cite[Prop. 7.1]{selinger:clifford}
  (see Remark 7.2 regarding the special case of $\le 2$-qubit
  circuits), and have been shown above to hold in any model of
  \SPiLang.
\end{proof}

\subsection{$n$-qubit Clifford Circuits}
\label{subsec:nclifford}

To extend Thm.~\ref{thm:2cliff} to Clifford circuits with an arbitrary
number of qubits, it suffices by a result of \citet{selinger:clifford}
to prove just four identities (shown in
Fig.~\ref{fig:nclifford}). Interestingly, by showing that models of
\SPiLang\ admit a few circuit rewriting rules and applying these, we
will see that the heavy lifting of these four identities can be done
entirely by \emph{classical} reasoning. This lets us exploit the
soundness and completeness of \PiLang{} with respect to its
permutation semantics, which greatly simplifies these proofs.

Recall that we interpret controlled gates in \SPiLang\ using the
$\Ctrl$ macro, such that, \textit{e.g.}, a controlled-$\X$ gate \raisebox{2mm}{$\scalebox{0.5}{\Qcircuit @C=.3em @R=.5em {
     & \ctrl{1} & \qw \\
     & \targ & \qw
}}$}
becomes $\Ctrl~\X$. If we're interested in a controlled gate where the
target line is above rather than below, we can simply conjugate it by
a swap, \textit{e.g.},
\[
\Qcircuit @C=.7em @R=.7em {
     & \targ & \qw \\
     & \ctrl{-1} & \qw
} \quad\raisebox{-2mm}{=}\quad
\Qcircuit @C=.7em @R=.7em {
     & \qswap & \ctrl{1} & \qswap & \qw \\
     & \qswap \qwx & \targ & \qswap \qwx & \qw
} \quad\raisebox{-2mm}{.}
\]
Thus a ``bottom-controlled'' $\X$ is interpreted in
\SPiLang\ as $\Swap \circ \Ctrl~\X \circ \Swap$. We first collect some
useful additional properties of $\Ctrl~X$ and $\Ctrl~\Z$, with proofs
located in Appendix~\ref{ap:lem:ctrlh}.

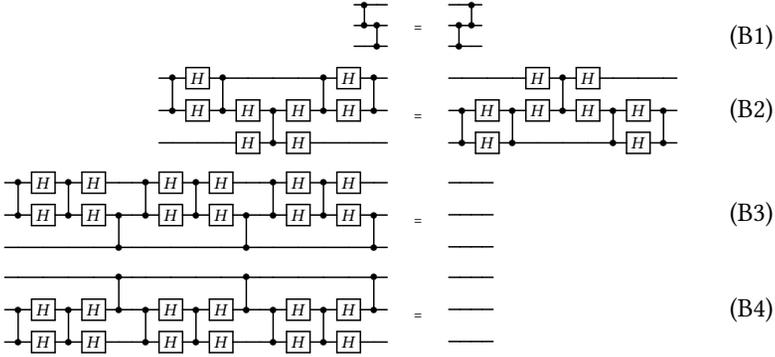
\begin{figure}
  \tiny
  \[
  \begin{array}{rcll}
    \Qcircuit @C=.5em @R=1em {
     & \ctrl{1} & \qw & \qw \\
     & \control\qw & \ctrl{1} & \qw \\
     & \qw & \control\qw & \qw
    } & \raisebox{-3mm}{=} &
    \Qcircuit @C=.5em @R=1em {
     & \qw & \control\qw & \qw \\
     & \ctrl{1} & \ctrl{-1} & \qw \\
     & \control\qw & \qw & \qw
    } & \raisebox{-4.5mm}{\normalsize \quad\text{(B1)}} \\ \\
    \Qcircuit @C=.7em @R=0.7em @!R {
      & \control\qw & \gate{H} & \control\qw & \qw      & \qw      & \qw
      & \control\qw & \gate{H} & \control\qw & \qw \\
      & \ctrl{-1}   & \gate{H} & \ctrl{-1}   & \gate{H} & \ctrl{1} & \gate{H}
      & \ctrl{-1}   & \gate{H} & \ctrl{-1}   & \qw \\
      & \qw         & \qw      & \qw         & \gate{H} & \control\qw & \gate{H}
      & \qw         & \qw      & \qw         & \qw
    } & \raisebox{-5mm}{=} &
    \Qcircuit @C=.7em @R=0.7em @!R {
      & \qw         & \qw      & \qw
      & \gate{H}    & \ctrl{1} & \gate{H}
      & \qw         & \qw      & \qw & \qw \\
      & \ctrl{1}    & \gate{H} & \ctrl{1}
      & \gate{H}    & \control\qw & \gate{H}
      & \control\qw & \gate{H} & \control\qw & \qw \\
      & \control\qw & \gate{H} & \control\qw
      & \qw         & \qw      & \qw
      & \ctrl{-1}   & \gate{H} & \ctrl{-1} & \qw
    } & \raisebox{-4.5mm}{\normalsize \quad\text{(B2)}} \\ \\
    \Qcircuit @C=.7em @R=0.7em @!R {
      & \control\qw & \gate{H} & \control\qw & \gate{H} & \qw & \qw
      & \control\qw & \gate{H} & \control\qw & \gate{H} & \qw & \qw
      & \control\qw & \gate{H} & \control\qw & \gate{H} & \qw & \qw \\
      & \ctrl{-1}   & \gate{H} & \ctrl{-1}   & \gate{H} & \ctrl{1} & \qw
      & \ctrl{-1}   & \gate{H} & \ctrl{-1}   & \gate{H} & \ctrl{1} & \qw
      & \ctrl{-1}   & \gate{H} & \ctrl{-1}   & \gate{H} & \ctrl{1} & \qw \\
      & \qw         & \qw      & \qw         & \qw & \control\qw & \qw
      & \qw         & \qw      & \qw         & \qw & \control\qw & \qw
      & \qw         & \qw      & \qw         & \qw & \control\qw & \qw
    } & \raisebox{-5mm}{=} &
    \Qcircuit @C=.7em @R=1em {
      & \qw & \qw & \qw & \qw \\ \\
      & \qw & \qw & \qw & \qw \\ \\
      & \qw & \qw & \qw & \qw
    } & \raisebox{-4.5mm}{\normalsize \quad\text{(B3)}} \\ \\
    \Qcircuit @C=.7em @R=0.7em @!R {
      & \qw         & \qw      & \qw         & \qw & \control\qw & \qw
      & \qw         & \qw      & \qw         & \qw & \control\qw & \qw
      & \qw         & \qw      & \qw         & \qw & \control\qw & \qw \\
      & \ctrl{1}   & \gate{H} & \ctrl{1}   & \gate{H} & \ctrl{-1} & \qw
      & \ctrl{1}   & \gate{H} & \ctrl{1}   & \gate{H} & \ctrl{-1} & \qw
      & \ctrl{1}   & \gate{H} & \ctrl{1}   & \gate{H} & \ctrl{-1} & \qw \\
      & \control\qw & \gate{H} & \control\qw & \gate{H} & \qw & \qw
      & \control\qw & \gate{H} & \control\qw & \gate{H} & \qw & \qw
      & \control\qw & \gate{H} & \control\qw & \gate{H} & \qw & \qw
    } & \raisebox{-5mm}{=} &
    \Qcircuit @C=.7em @R=1em {
      & \qw & \qw & \qw & \qw \\ \\
      & \qw & \qw & \qw & \qw \\ \\
      & \qw & \qw & \qw & \qw
    } & \raisebox{-4.5mm}{\normalsize \quad\text{(B4)}}
  \end{array}
  \]
  \caption{The 3-qubit identities of Clifford circuits due to
    \citet{selinger:clifford} which, together with (A3)--(A13) of
    Fig.~\ref{fig:2qubitclifford}, form a sound and complete
    equational theory of Clifford circuits.}
  \label{fig:nclifford}
\end{figure}

\begin{lemmarep}\label{lem:ctrlh}
  The following identities hold in any model of \SPiLang:
  \begin{enumerate}[label={(\roman*)}]
  \item $\id \otimes \Had \circ \Ctrl~\X \circ \id \otimes \Had =
    \Ctrl~\Z$,
  \item $\Had \otimes \id \circ \Swap \circ \Ctrl~\X \circ \Swap \circ
    \Had \otimes \id = \Ctrl~\Z$,
  \item $\id \otimes \Had \circ \Ctrl~\Z \circ \id \otimes \Had =
    \Ctrl~\X$,
  \item $\Had \otimes \id \circ \Ctrl~\Z \circ \Had \otimes \id =
    \Swap \circ \Ctrl~\X \circ \Swap$,
  \item $\Had \otimes \id \circ \Ctrl~\X \circ \Had \otimes \id = \id
    \otimes \Had \circ \Swap \circ \Ctrl~\X \circ \Swap \circ \id
    \otimes \Had$
  \end{enumerate}
\end{lemmarep}
\begin{proof}
  \label{ap:lem:ctrlh}
  For (i),
  \begin{align*}
    \id \otimes \Had \circ \Ctrl~\X \circ \id \otimes \Had
    & = \id \otimes \Had \circ \Mat^{-1} \circ \id \oplus \X \circ
    \Mat \circ \id \otimes \Had & (\text{def.}~\Ctrl) \\
    & = \Mat^{-1} \circ (\Had \oplus \Had) \circ (\id \oplus \X) \circ
    (\Had \oplus \Had) \circ \Mat & (\text{Lem.~\ref{lem:gates} (i)}) \\
    & = \Mat^{-1} \circ (\Had \circ \Had) \oplus (\Had \circ \X \circ
    \Had) \circ \Mat & (\text{bifunctoriality}~\oplus) \\
    & = \Mat^{-1} \circ (\id \oplus (\Had \circ \X \circ
    \Had)) \circ \Mat & (\text{A4}) \\
    & = \Mat^{-1} \circ (\id \oplus \Z) \circ \Mat &
    (\text{Lem.~\ref{lem:had}})\\
    & = \Ctrl~\Z & (\text{def.}~\Ctrl)
  \end{align*}
  and (ii),
  \begin{align*}
    &\Had \otimes \id \circ \Swap \circ \Ctrl~\X \circ \Swap \circ \Had
    \otimes \id \\
    &\quad  = \Swap \circ \id \otimes \Had \circ \Ctrl~\X \circ \id \otimes
    \Had \circ \Swap & (\text{naturality}~\Swap) \\
    &\quad = \Swap \circ \Ctrl~\Z \circ \Swap &
    (\text{Lem.~\ref{lem:ctrlh} (i)}) \\
    &\quad = \Ctrl~\Z & \text{(Lem.~\ref{lem:mat} (v))}
  \end{align*}
  Points (iii) and (iv) follow entirely analogously to (i) and (ii)
  respectively. As for (v),
  \begin{align*}
    & \Had \otimes \id \circ \Ctrl~\X \circ \Had \otimes \id \\
    & \quad = \Had \otimes \id \circ \id \otimes \Had \circ \Ctrl~\Z
    \circ \id \otimes \Had \circ \Had \otimes \id &
    (\text{Lem.~\ref{lem:ctrlh} (iii)}) \\
    & \quad = \Had \otimes \Had \circ \Ctrl~\Z
    \circ \Had \otimes \Had &
    (\text{bifunctoriality}~\otimes) \\
    & \quad = \Had \otimes \Had \circ \Swap \circ \Ctrl~\Z \circ \Swap
    \circ \Had \otimes \Had &
    (\text{Lem.~\ref{lem:mat} (v)}) \\
    & \quad = \Swap \circ \Had \otimes \Had \circ \Ctrl~\Z \circ
    \Had \otimes \Had \circ \Swap &
    (\text{naturality}~\Swap) \\
    & \quad = \Swap \circ \Had \otimes \id \circ \Ctrl~\X \circ
    \Had \otimes \id \circ \Swap &
    (\text{Lem.~\ref{lem:ctrlh} (iii)}) \\
    & \quad = \id \otimes \Had \circ \Swap \circ \Ctrl~\X \circ
    \Swap \circ \id \otimes \Had &
    (\text{naturality}~\Swap)
  \end{align*}
\end{proof}

These have direct interpretations as circuit identities,
which we will use to simplify \ref{eq:b1}--\ref{eq:b4}.

\begin{corollary}\label{corr:circuitid}
  The following circuit identities hold in any model of \SPiLang:
\vspace{2mm}
 \begin{multicols}{2}
  \begin{enumerate}[label={(\roman*)}]\setlength\itemsep{1em}
    \item \raisebox{2.5mm}{\scriptsize $\Qcircuit @C=.5em @R=.8em {
     & \qw & \ctrl{1} & \qw & \qw \\
     & \gate{H} & \targ & \gate{H} & \qw
    } \quad\raisebox{-2mm}{=}\quad
    \Qcircuit @C=0.5em @R=1.2em {
     & \ctrl{1} & \qw \\
     & \control \qw & \qw
    }$\normalsize} \enspace,
    \item \raisebox{2.5mm}{\scriptsize$\Qcircuit @C=.5em @R=0.8em {
     & \gate{H} & \targ & \gate{H} & \qw \\
     & \qw & \ctrl{-1} & \qw & \qw
    } \quad\raisebox{-2mm}{=}\quad
    \Qcircuit @C=.5em @R=1.2em {
     & \control \qw & \qw \\
     & \ctrl{-1} & \qw
    }$\normalsize} \enspace,
    \item \raisebox{2.5mm}{\scriptsize$\Qcircuit @C=.5em @R=0.7em {
     & \qw & \ctrl{1} & \qw & \qw \\
     & \gate{H} & \control \qw & \gate{H} & \qw
    } \quad\raisebox{-2mm}{=}\quad
    \Qcircuit @C=.5em @R=1em {
     & \ctrl{1} & \qw \\
     & \targ & \qw
    }$\normalsize} \enspace,
    \item \raisebox{2.5mm}{\scriptsize$\Qcircuit @C=.5em @R=0.7em {
     & \gate{H} & \ctrl{1} & \gate{H} & \qw \\
     & \qw & \control \qw & \qw & \qw
    } \quad\raisebox{-2mm}{=}\quad
    \Qcircuit @C=.5em @R=1em {
     & \targ & \qw \\
     & \ctrl{-1} & \qw
    }$\normalsize} \enspace,
    \item \raisebox{2.5mm}{\scriptsize$\Qcircuit @C=.5em @R=.7em {
     & \gate{H} & \ctrl{1} & \gate{H} & \qw \\
     & \qw & \targ & \qw & \qw
    } \quad\raisebox{-2mm}{=}\quad
    \Qcircuit @C=.5em @R=.7em {
     & \qw & \targ & \qw & \qw \\
     & \gate{H} & \ctrl{-1} & \gate{H} & \qw
    }$\normalsize} \enspace,
    \item \raisebox{2.5mm}{\scriptsize$\Qcircuit @C=.5em @R=.6em @!R {
     & \gate{U} & \qswap & \qw & \qw  \\
     & \qw & \qswap\qwx & \qw & \qw
    } \quad\raisebox{-2mm}{=}\quad
    \Qcircuit @C=.5em @R=.6em @!R {
     & \qw & \qswap & \qw & \qw  \\
     & \qw & \qswap\qwx & \gate{U} & \qw
    }$\normalsize} \enspace and \enspace
    \raisebox{2.5mm}{\scriptsize$\Qcircuit @C=.5em @R=.6em @!R {
     & \qw & \qswap & \qw & \qw  \\
     & \gate{U} & \qswap\qwx & \qw & \qw
    } \quad\raisebox{-2mm}{=}\quad
    \Qcircuit @C=.5em @R=.6em @!R {
     & \qw & \qswap & \gate{U} & \qw  \\
     & \qw & \qswap\qwx & \qw & \qw
    }$\normalsize}
    \enspace for any gate \enspace \scriptsize$\Qcircuit @C=.5em @R=.7em {
    & \gate{U} & \qw
    }$\normalsize~.
  \end{enumerate}
  \end{multicols}
\end{corollary}
\begin{proof}
  Points (i)--(v) hold by Lem.~\ref{lem:ctrlh}, while (vi) is
  naturality of $\Swap$.
\end{proof}

We can now tackle the four 3-qubit rules for Clifford circuits, named
(C12)--(C15) in the presentation of \citet{selinger:clifford}, which
we call \ref{eq:b1}--\ref{eq:b4}.

\begin{enumerate}[label={(B\arabic*)}]
\item This rule is can be derived using the circuit
  identities and classical completeness.
  {\scriptsize\begin{align*}
    \Qcircuit @C=.5em @R=1.2em {
     & \ctrl{1} & \qw & \qw \\
     & \control\qw & \ctrl{1} & \qw \\
     & \qw & \control\qw & \qw
    } & \quad \raisebox{-3.5mm}{=} \quad
    \Qcircuit @C=.5em @R=.7em {
     & \gate{H} & \gate{H} & \ctrl{1} & \gate{H} & \gate{H} & \qw & \qw \\
     & \qw & \qw & \control\qw & \ctrl{1} & \qw & \qw & \qw \\
     & \qw & \gate{H} & \gate{H} & \control\qw & \gate{H} & \gate{H} & \qw
    } & \text{\normalsize \ref{eq:a4}} \\[0.5\baselineskip]
    & \quad \raisebox{-3.5mm}{=} \quad
    \Qcircuit @C=.5em @R=.7em {
     & \gate{H} & \targ & \gate{H} & \qw & \qw \\
     & \qw & \ctrl{-1} & \ctrl{1} & \qw & \qw \\
     & \qw & \gate{H} & \targ & \gate{H} & \qw
    } & \text{\normalsize{(Cor.~\ref{corr:circuitid})}} \\[0.5\baselineskip]
    & \quad \raisebox{-3.5mm}{=} \quad
    \Qcircuit @C=.5em @R=.7em {
     & \qw & \gate{H} & \targ & \gate{H} & \qw \\
     & \qw & \ctrl{1} & \ctrl{-1} & \qw & \qw \\
     & \gate{H} & \targ & \gate{H} & \qw & \qw
    } & \text{\normalsize \eqref{eq:p1}} \\[0.5\baselineskip]
    & \quad \raisebox{-3.5mm}{=} \quad
    \Qcircuit @C=.5em @R=1.2em {
     & \qw & \control\qw & \qw \\
     & \ctrl{1} & \ctrl{-1} & \qw \\
     & \control\qw & \qw & \qw
    } & \text{\normalsize{(Cor.~\ref{corr:circuitid})}}
  \end{align*}} \label{eq:b1}
  Notice how the essential argument of this proof is the classical
  identity \eqref{eq:p1}.

\item \label{item:b2} We defer the proof of this identity to
  Appendix~\ref{eq:b2}.
  \begin{appendixproof}[Proof of ~\ref{item:b2}]
\label{eq:b2}
  Once again, we derive this from the circuit identities and a
  classical lemma:
  {\scriptsize\begin{align*}
  \Qcircuit @C=.7em @R=0.7em @!R {
    & \control\qw & \gate{H} & \control\qw & \qw      & \qw      & \qw
    & \control\qw & \gate{H} & \control\qw & \qw \\
    & \ctrl{-1}   & \gate{H} & \ctrl{-1}   & \gate{H} & \ctrl{1} & \gate{H}
    & \ctrl{-1}   & \gate{H} & \ctrl{-1}   & \qw \\
    & \qw         & \qw      & \qw         & \gate{H} & \control\qw & \gate{H}
    & \qw         & \qw      & \qw         & \qw
  } & \qquad \raisebox{-5mm}{=} \qquad
  \Qcircuit @C=.7em @R=0.7em @!R {
    & \gate{H}    & \gate{H} & \control\qw & \gate{H} & \control\qw
    & \qw         & \qw      & \qw
    & \control\qw & \gate{H} & \control\qw & \gate{H} & \gate{H} & \qw\\
    & \qw         & \qw      & \ctrl{-1}   & \gate{H} & \ctrl{-1}
    & \gate{H}    & \ctrl{1} & \gate{H}
    & \ctrl{-1}   & \gate{H} & \ctrl{-1}   & \qw      & \qw & \qw \\
    & \qw         & \qw      & \qw         & \qw      & \qw
    & \gate{H}    & \control\qw & \gate{H}
    & \qw         & \qw      & \qw         & \qw      & \qw      &\qw
  } & \text{\normalsize \ref{eq:a4}} \\[0.5\baselineskip]
  & \qquad \raisebox{-4mm}{=} \qquad
  \Qcircuit @C=.7em @R=0.5em @!R {
    & \gate{H} & \targ     & \ctrl{1} & \qw      & \ctrl{1} & \targ     & \gate{H}
    & \qw \\
    & \qw      & \ctrl{-1} & \targ    & \ctrl{1} & \targ    & \ctrl{-1} & \qw
    & \qw \\
    & \qw      & \qw       & \qw      & \targ    & \qw      & \qw       & \qw
    & \qw
  } & \text{\normalsize{(Cor.~\ref{corr:circuitid})}} \\[0.5\baselineskip]
  & \qquad \raisebox{-4mm}{=} \qquad
  \Qcircuit @C=.7em @R=0.5em @!R {
    & \gate{H} & \qswap     & \qw      & \qswap     & \gate{H} & \qw \\
    & \qw      & \qswap\qwx & \ctrl{1} & \qswap\qwx & \qw      & \qw \\
    & \qw      & \qw        & \targ    & \qw        & \qw      & \qw
  } & \text{\normalsize \eqref{eq:p2}} \\[0.5\baselineskip]
  & \qquad \raisebox{-4mm}{=} \qquad
  \Qcircuit @C=.7em @R=0.5em @!R {
    & \qswap     & \qw      & \qw      & \qw      & \qswap & \qw \\
    & \qswap\qwx & \gate{H} & \ctrl{1} & \gate{H} & \qswap\qwx & \qw \\
    & \qw        & \qw      & \targ    & \qw      & \qw    & \qw
  } & \text{\normalsize{(Cor.~\ref{corr:circuitid})}} \\[0.5\baselineskip]
  & \qquad \raisebox{-4mm}{=} \qquad
  \Qcircuit @C=.7em @R=0.5em @!R {
    & \qswap      & \qw       & \qswap     & \qw \\
    & \qswap\qwx  & \targ     & \qswap\qwx & \qw \\
    & \gate{H}    & \ctrl{-1} & \gate{H}   & \qw
  } & \text{\normalsize{(Cor.~\ref{corr:circuitid})}} \\[0.5\baselineskip]
  & \qquad \raisebox{-4mm}{=} \qquad
  \Qcircuit @C=.7em @R=0.5em @!R {
    & \qw      & \qw      & \qw       & \targ     & \qw       & \qw
    & \qw      & \qw \\
    & \qw      & \ctrl{1} & \targ     & \ctrl{-1} & \targ     & \ctrl{1}
    & \qw      & \qw \\
    & \gate{H} & \targ    & \ctrl{-1} & \qw       & \ctrl{-1} & \targ
    & \gate{H} & \qw
  } & \text{\normalsize \eqref{eq:p3}} \\[0.5\baselineskip]
  & \qquad \raisebox{-5mm}{=} \qquad
  \Qcircuit @C=.7em @R=0.7em @!R {
    & \qw         & \qw      & \qw         & \qw      & \qw
    & \gate{H}    & \ctrl{1} & \gate{H}
    & \qw         & \qw      & \qw         & \qw & \qw & \qw\\
    & \qw         & \qw      & \ctrl{1}    & \gate{H} & \ctrl{1}
    & \gate{H}    & \control\qw & \gate{H}
    & \control\qw & \gate{H} & \control\qw & \qw      & \qw & \qw \\
    & \gate{H}    & \gate{H} & \control\qw & \gate{H} & \control\qw
    & \qw         & \qw      & \qw
    & \ctrl{-1}   & \gate{H} & \ctrl{-1}   & \gate{H} & \gate{H} & \qw
  } & \text{\normalsize{(Cor.~\ref{corr:circuitid})}} \\[0.5\baselineskip]
  & \qquad \raisebox{-5mm}{=} \qquad
  \Qcircuit @C=.7em @R=0.7em @!R {
    & \qw         & \qw      & \qw
    & \gate{H}    & \ctrl{1} & \gate{H}
    & \qw         & \qw      & \qw & \qw \\
    & \ctrl{1}    & \gate{H} & \ctrl{1}
    & \gate{H}    & \control\qw & \gate{H}
    & \control\qw & \gate{H} & \control\qw & \qw \\
    & \control\qw & \gate{H} & \control\qw
    & \qw         & \qw      & \qw
    & \ctrl{-1}   & \gate{H} & \ctrl{-1} & \qw
  } & \text{\normalsize \ref{eq:a4}}
  \end{align*}}
  \end{appendixproof}
\item This identity and the next follow by reducing the circuit to one
  with a large classical subcircuit, which turns out (by classical
  completeness) to be the identity circuit.
{\scriptsize\begin{align*}
  & \Qcircuit @C=.7em @R=0.7em @!R {
    & \control\qw & \gate{H} & \control\qw & \gate{H} & \qw & \qw
    & \control\qw & \gate{H} & \control\qw & \gate{H} & \qw & \qw
    & \control\qw & \gate{H} & \control\qw & \gate{H} & \qw & \qw \\
    & \ctrl{-1}   & \gate{H} & \ctrl{-1}   & \gate{H} & \ctrl{1} & \qw
    & \ctrl{-1}   & \gate{H} & \ctrl{-1}   & \gate{H} & \ctrl{1} & \qw
    & \ctrl{-1}   & \gate{H} & \ctrl{-1}   & \gate{H} & \ctrl{1} & \qw \\
    & \qw         & \qw      & \qw         & \qw & \control\qw & \qw
    & \qw         & \qw      & \qw         & \qw & \control\qw & \qw
    & \qw         & \qw      & \qw         & \qw & \control\qw & \qw
  } \\[0.5\baselineskip]
  & \qquad \raisebox{-5mm}{=} \qquad
  \Qcircuit @C=.7em @R=0.7em @!R {
    & \qw & \qw
    & \control\qw & \gate{H} & \control\qw & \gate{H} & \qw
    & \qw & \qw
    & \control\qw & \gate{H} & \control\qw & \gate{H} & \qw
    & \qw & \qw
    & \control\qw & \gate{H} & \control\qw & \gate{H} & \qw
    & \qw & \qw & \qw \\
    & \gate{H} & \gate{H}
    & \ctrl{-1}   & \gate{H} & \ctrl{-1}   & \gate{H} & \ctrl{1}
    & \gate{H} & \gate{H}
    & \ctrl{-1}   & \gate{H} & \ctrl{-1}   & \gate{H} & \ctrl{1}
    & \gate{H} & \gate{H}
    & \ctrl{-1}   & \gate{H} & \ctrl{-1}   & \gate{H} & \ctrl{1}
    & \gate{H} & \gate{H} & \qw \\
    & \qw & \qw
    & \qw         & \qw      & \qw         & \qw & \control\qw
    & \qw & \qw
    & \qw         & \qw      & \qw         & \qw & \control\qw
    & \qw & \qw
    & \qw         & \qw      & \qw         & \qw & \control\qw
    & \qw & \qw & \qw
  } & \text{\normalsize \ref{eq:a4}} \\[0.5\baselineskip]
  & \qquad \raisebox{-5mm}{=} \qquad
  \Qcircuit @C=.7em @R=0.7em @!R {
    & \qw
    & \ctrl{1} & \targ     & \qw
    & \ctrl{1} & \targ     & \qw
    & \ctrl{1} & \targ     & \qw
    & \qw & \qw \\
    & \gate{H}
    & \targ    & \ctrl{-1} & \targ
    & \targ    & \ctrl{-1} & \targ
    & \targ    & \ctrl{-1} & \targ
    & \gate{H} & \qw \\
    & \qw
    & \qw      & \qw       & \ctrl{-1}
    & \qw      & \qw       & \ctrl{-1}
    & \qw      & \qw       & \ctrl{-1}
    & \qw & \qw
  } & \text{\normalsize{(Cor.~\ref{corr:circuitid})}} \\[0.5\baselineskip]
  & \qquad \raisebox{-5mm}{=} \qquad
  \Qcircuit @C=.7em @R=0.7em @!R {
    & \qw & \qw & \qw \\
    & \gate{H} & \gate{H} & \qw \\
    & \qw & \qw & \qw
  } & \text{\normalsize \eqref{eq:p4}} \\[0.5\baselineskip]
  & \qquad \raisebox{-3mm}{=} \qquad
  \Qcircuit @C=.7em @R=0.7em {
    & \qw & \qw \\ \\
    & \qw & \qw \\ \\
    & \qw & \qw
  } & \text{\normalsize \ref{eq:a4}} \label{eq:b3}
\end{align*}}
\item We defer the proof of this identity to Appendix~\ref{proof:b4}.\label{eq:b4}
  \begin{appendixproof}[Proof of \ref{eq:b4}]
    \label{proof:b4}
{\scriptsize\begin{align*}
  & \Qcircuit @C=.7em @R=0.7em @!R {
    & \qw         & \qw      & \qw         & \qw & \control\qw & \qw
    & \qw         & \qw      & \qw         & \qw & \control\qw & \qw
    & \qw         & \qw      & \qw         & \qw & \control\qw & \qw \\
    & \ctrl{1}   & \gate{H} & \ctrl{1}   & \gate{H} & \ctrl{-1} & \qw
    & \ctrl{1}   & \gate{H} & \ctrl{1}   & \gate{H} & \ctrl{-1} & \qw
    & \ctrl{1}   & \gate{H} & \ctrl{1}   & \gate{H} & \ctrl{-1} & \qw \\
    & \control\qw & \gate{H} & \control\qw & \gate{H} & \qw & \qw
    & \control\qw & \gate{H} & \control\qw & \gate{H} & \qw & \qw
    & \control\qw & \gate{H} & \control\qw & \gate{H} & \qw & \qw
  } \\[0.5\baselineskip]
  & \qquad \raisebox{-5mm}{=} \qquad
  \Qcircuit @C=.7em @R=0.7em @!R {
    & \qw & \qw
    & \qw         & \qw      & \qw         & \qw & \control\qw
    & \qw & \qw
    & \qw         & \qw      & \qw         & \qw & \control\qw
    & \qw & \qw
    & \qw         & \qw      & \qw         & \qw & \control\qw
    & \qw & \qw & \qw \\
    & \gate{H} & \gate{H}
    & \ctrl{1}   & \gate{H} & \ctrl{1}   & \gate{H} & \ctrl{-1}
    & \gate{H} & \gate{H}
    & \ctrl{1}   & \gate{H} & \ctrl{1}   & \gate{H} & \ctrl{-1}
    & \gate{H} & \gate{H}
    & \ctrl{1}   & \gate{H} & \ctrl{1}   & \gate{H} & \ctrl{-1}
    & \gate{H} & \gate{H} & \qw \\
    & \qw & \qw
    & \control\qw & \gate{H} & \control\qw & \gate{H} & \qw
    & \qw & \qw
    & \control\qw & \gate{H} & \control\qw & \gate{H} & \qw
    & \qw & \qw
    & \control\qw & \gate{H} & \control\qw & \gate{H} & \qw
    & \qw & \qw & \qw
  } & \text{\normalsize \ref{eq:a4}} \\[0.5\baselineskip]
  & \qquad \raisebox{-5mm}{=} \qquad
  \Qcircuit @C=.7em @R=0.7em @!R {
    & \qw
    & \qw      & \qw       & \ctrl{1}
    & \qw      & \qw       & \ctrl{1}
    & \qw      & \qw       & \ctrl{1}
    & \qw & \qw \\
    & \gate{H}
    & \targ    & \ctrl{1} & \targ
    & \targ    & \ctrl{1} & \targ
    & \targ    & \ctrl{1} & \targ
    & \gate{H} & \qw \\
    & \qw
    & \ctrl{-1} & \targ     & \qw
    & \ctrl{-1} & \targ     & \qw
    & \ctrl{-1} & \targ     & \qw
    & \qw & \qw
  } & \text{\normalsize{(Cor.~\ref{corr:circuitid})}} \\[0.5\baselineskip]
  & \qquad \raisebox{-5mm}{=} \qquad
  \Qcircuit @C=.7em @R=0.7em @!R {
    & \qw & \qw & \qw \\
    & \gate{H} & \gate{H} & \qw \\
    & \qw & \qw & \qw
  } & \text{\normalsize \eqref{eq:p5}} \\[0.5\baselineskip]
  & \qquad \raisebox{-3mm}{=} \qquad
  \Qcircuit @C=.7em @R=0.7em {
    & \qw & \qw \\ \\
    & \qw & \qw \\ \\
    & \qw & \qw
  } & \text{\normalsize \ref{eq:a4}}
\end{align*}}
\end{appendixproof}
\end{enumerate}

From this follows an equational completeness result for Clifford circuits of
arbitrary size.

\begin{theorem}[Full abstraction for Clifford circuits]\label{thm:ncliff}
  Let $c_1$ and $c_2$ be \SPiLang\ terms representing Clifford circuits of
  arbitrary size. Then $\sem{c_1} = \sem{c_2}$ iff $\stsem{c_1} =
  \stsem{c_2}$.
\end{theorem}
\begin{proof}
  The identities \ref{eq:a3}--\ref{eq:a13} and
  \ref{eq:b1}--\ref{eq:b4} are complete for Clifford circuits of
  arbitrary size by \citet[Thm.~7.1]{selinger:clifford}, and have been
  shown above to hold in any model of \SPiLang.
\end{proof}

\subsection{$\le 2$-qubit Clifford+T}
\label{subsec:2clifft}
\begin{figure}
  \tiny\[
  \begin{array}{rcll rcll}
    \text{\normalsize $\T^2$} &=& \text{\normalsize $\Sg$}
    & \text{\normalsize (A14)} &
    \text{\normalsize $(\T \Had \Sg \Sg \Had)^2$} &=&
    \text{\normalsize $\omega \cdot \id$}
    & \text{\normalsize (A15)} \\ \\
    \Qcircuit @C=.7em @R=1.2em {
      & \gate{T} & \ctrl{1}    & \qw \\
      & \qw      & \control\qw & \qw
    } &\raisebox{-2mm}{=}&
    \Qcircuit @C=.7em @R=1.2em {
      & \ctrl{1}    & \gate{T} & \qw \\
      & \control\qw & \qw      & \qw
    } & \text{\raisebox{-2mm}{\normalsize (A16)}} &
    \Qcircuit @C=.7em @R=.7em {
      & \qw      & \ctrl{1}    & \gate{H} & \ctrl{1}    & \gate{H} & \gate{T}
      & \qw \\
      & \gate{H} & \control\qw & \gate{H} & \control\qw & \qw      & \qw
      & \qw
    } &\raisebox{-2mm}{=}&
    \Qcircuit @C=.7em @R=.7em {
      & \qw      & \qw      & \ctrl{1}    & \gate{H} & \ctrl{1}    & \gate{H}
      & \qw \\
      & \gate{T} & \gate{H} & \control\qw & \gate{H} & \control\qw & \qw
      & \qw
    } & \text{\raisebox{-2mm}{\normalsize (A17)}} \\
  \end{array}
  \]
  \tiny\[
  \begin{array}{rcll}
    \Qcircuit @C=.7em @R=.7em {
      & \qw & \ctrl{1} & \qw      & \qw      & \qw          & \ctrlo{1}
      & \qw      & \qw      & \qw          & \qw \\
      & \qw & \targ    & \gate{T} & \gate{H} & \gate{T^{-1}} & \targ
      & \gate{T} & \gate{H} & \gate{T^{-1}} & \qw
    } &\raisebox{-2mm}{=}&
    \Qcircuit @C=.7em @R=.7em {
      & \qw      & \qw      & \qw          & \ctrlo{1} & \qw      & \qw
      & \qw          & \ctrl{1} & \qw \\
      & \gate{T} & \gate{H} & \gate{T^{-1}} & \targ     & \gate{T} & \gate{H}
      & \gate{T^{-1}} & \targ    & \qw
    } & \text{\raisebox{-2mm}{\normalsize (A18)}} \\
    \Qcircuit @C=.7em @R=.7em {
      & \qw & \ctrl{1} & \qw      & \qw      & \qw      & \qw      & \qw
      & \ctrlo{1} & \qw      & \qw      & \qw          & \qw      & \qw & \qw \\
      & \qw & \targ    & \gate{T} & \gate{H} & \gate{T} & \gate{H} & \gate{T^{-1}}
      & \targ     & \gate{T} & \gate{H} & \gate{T^{-1}} & \gate{H} & \gate{T^{-1}}
      & \qw
    } &\raisebox{-2mm}{=}&
    \Qcircuit @C=.7em @R=.7em {
      & \qw & \qw      & \qw      & \qw      & \qw      & \qw
      & \ctrlo{1} & \qw      & \qw      & \qw          & \qw      & \qw
      & \ctrl{1} & \qw \\
      & \qw & \gate{T} & \gate{H} & \gate{T} & \gate{H} & \gate{T^{-1}}
      & \targ     & \gate{T} & \gate{H} & \gate{T^{-1}} & \gate{H} & \gate{T^{-1}}
      & \targ & \qw
    }
    & \text{\raisebox{-2mm}{\normalsize (A19)}} \\ \\
    \Qcircuit @C=.7em @R=.7em {
      & \qw & \ctrlo{1} & \qw      & \ctrl{1} & \gate{H}   & \gate{T}
      & \gate{H} & \qw \\
      & \qw & \gate{H}  & \gate{T} & \gate{H} & \ctrlo{-1} & \qw
      & \ctrl{-1} & \qw
    } &\raisebox{-2mm}{=}&
    \Qcircuit @C=.7em @R=.7em {
      & \qw & \gate{H}  & \gate{T} & \gate{H} & \ctrlo{1} & \qw
      & \ctrl{1} & \qw \\
      & \qw & \ctrlo{-1} & \qw      & \ctrl{-1} & \gate{H}   & \gate{T}
      & \gate{H} & \qw
    } & \text{\raisebox{-2mm}{\normalsize (A20)}}
  \end{array}
  \]
  \caption{The remaining identities which, along with (A1)--(A13) of
    Fig.~\ref{fig:2qubitclifford}, form a sound and complete
    equational theory of $\le 2$-qubit Clifford+T
    circuits~\cite{bianselinger:cliffordt}.}
  \label{fig:2qubitcliffordt}
\end{figure}
We extend Thm.~\ref{thm:2cliff} to show that
models of \SPiLang\ are sound and complete for all $\le 2$-qubit
Clifford+T circuits. We do this by showing the remaining identities of
\citet{bianselinger:cliffordt} (see Fig.~\ref{fig:2qubitcliffordt}),
which, together with \ref{eq:a1}--\ref{eq:a13} from
Sec.~\ref{subsec:2clifford}, are equationally sound and complete for
$\le 2$-qubit Clifford+T circuits.  Recall that Clifford+T circuits
are those which can be formed using the scalar $\omega$ and gates
$\{\Sg, \Had, \CZ, \T\}$. This leads us to the following definition
of representations of Clifford+T circuits in models of $\SPiLang$:
\begin{definition}
In a model of \SPiLang, a \emph{representation of a Clifford+T
circuit} is any morphism which can be written in terms of morphisms
from the sets $\{\omega, \Sg, \Had, \CZ, \T\}$ and $\{\alpha_\otimes,
\alpha_\otimes^{-1}, \lambda_\otimes, \lambda_\otimes^{-1},
\rho_\otimes, \rho_\otimes^{-1}, \sigma_\otimes\}$, composed
arbitrarily in parallel (using $\otimes$) and in sequence (using
$\circ$). A representation of a $\le 2$-qubit Clifford+T circuit is
one with signature $I \oplus I \to I \oplus I$ or $(I \oplus I)
\otimes (I \oplus I) \to (I \oplus I) \otimes (I \oplus I)$.
\end{definition}
We start by showing an equivalence of representations of negatively
controlled gates, as the definition of $\nCtrl$ in
Fig.~\ref{fig:abbrev} may be considered non-standard. One usually
thinks of a negatively controlled gate as a positively controlled one
conjugated by $\X$ on the control line, and we show that our
definition \nCtrl{} is a convenient reduced form for stating this.
\citet{bianselinger:cliffordt} uses yet another representation of
negatively controlled $\X$ and $\Had$, which we also show to be
equivalent.
\begin{lemma}[Negative control]\label{lem:nctrl}
  Let $f : X \to X$ be a map in a rig category. Then
  \begin{enumerate}[label={(\roman*)}]
  \item $\nCtrl~f = \X \otimes \id \circ \Ctrl~f \circ \X \otimes \id$,
  \item $\nCtrl~f = \Ctrl~f \circ \id \otimes f$ when $f$ is involutive.
  \end{enumerate}
\end{lemma}
\begin{proof}
  We derive (i) by
  \begin{align*}
    & \X \otimes \id \circ \Ctrl~f \circ \X \otimes \id \\
    & \quad = \X \otimes \id \circ \Mat^{-1} \circ (\id \oplus f)
    \circ \Mat \circ \X \otimes \id & (\text{definition}~\Ctrl) \\
    & \quad = \Mat^{-1} \circ \sigma_\oplus \circ (\id \oplus f) \circ
    \sigma_\oplus \circ \Mat & (\text{Lem.~\ref{lem:mat} (viii)}) \\
    & \quad = \Mat^{-1} \circ (f \oplus \id) \circ \sigma_\oplus \circ
    \sigma_\oplus \circ \Mat & (\text{naturality}~\sigma_\oplus) \\
    & \quad = \Mat^{-1} \circ (f \oplus \id) \circ \Mat &
    (\sigma_\oplus~\text{involutive}) \\
    & \quad = \nCtrl~f & (\text{definition}~\nCtrl)
  \end{align*}
  and we show (ii) by
  \begin{align*}
    \Ctrl~f \circ (\id \otimes f)
    & = \Mat^{-1} \circ (\id \oplus f) \circ \Mat \circ (\id \otimes
    f) & (\text{definition}~\Ctrl) \\
    & = \Mat^{-1} \circ (\id \oplus f) \circ (f \oplus f) \circ \Mat &
    (\text{Lem.~\ref{lem:mat} (i)}) \\
    & = \Mat^{-1} \circ (f \oplus (f \circ f)) \circ \Mat &
    (\text{bifunctoriality}~\oplus) \\
    & = \Mat^{-1} \circ (f \oplus \id) \circ \Mat & (f~\text{involutive}) \\
    & = \nCtrl~f & (\text{definition}~\nCtrl)
  \end{align*}
\end{proof}
We are now ready to derive the remaining identities.

\begin{enumerate}[label={(A\arabic*)}]\setcounter{enumi}{13}
\item By Lem.~\ref{lem:gates} and definition of $\Sg$ and $\T$,
  $\T^2 = \Pg(\omega)^2 = \Pg(\omega^2) = \Sg$. \label{eq:a14}
\item We derive
  \begin{align*}
    (\T \circ \Had \circ \Sg \circ \Sg \circ \Had)^2
    & = (\T \circ \Had \circ \Z \circ \Had)^2 & (\Sg^2 = \Z) \\
    & = (\T \circ \X)^2 & (\text{Lem.~\ref{lem:had}}) \\
    & = \T \circ \X \circ \T \circ \X & (\text{expand}) \\
    & = (\omega \sprod \X) \circ \X & (\text{Lem.~\ref{lem:gates}}) \\
    & = \omega \sprod (\X \circ \X) & (\text{Prop.~\ref{prop:scalars}}) \\
    & = \omega \sprod \id & (\X^2 = \id)
  \end{align*} \label{eq:a15}
\item This is a special case of commutativity of phase gates:
  \begin{align*}
    \Ctrl~\Z \circ (\T \otimes \id)
    & = \Swap \circ \Ctrl~\Z \circ \Swap \circ (\T \otimes \id) &
    (\text{Lem.~\ref{lem:mat}}) \\
    & = \Swap \circ \Ctrl~\Z \circ (\id \otimes \T) \circ \Swap &
    (\text{naturality}~\Swap) \\
    & = \Swap \circ (\id \otimes \T) \circ \Ctrl~\Z \circ \Swap &
    (\text{Lem.~\ref{lem:mat}}) \\
    & = (\T \otimes \id) \circ \Swap \circ \Ctrl~\Z \circ \Swap &
    (\text{naturality}~\Swap) \\
    & = (\T \otimes \id) \circ \Ctrl~\Z & (\text{Lem.~\ref{lem:mat}})
  \end{align*} \label{eq:a16}
\item By first applying circuit identities from
  Cor.~\ref{corr:circuitid}, this identity amounts to showing that
  {\scriptsize\begin{equation*}
  \Qcircuit @C=.7em @R=0.7em @!R {
    & \ctrl{1} & \targ & \gate{T} & \qw \\
    & \targ & \ctrl{-1} & \qw & \qw
  } \qquad \raisebox{-2mm}{=} \qquad
  \Qcircuit @C=.7em @R=0.7em @!R {
    & \qw      & \ctrl{1} & \targ     & \qw \\
    & \gate{T} & \targ    & \ctrl{-1} & \qw
  }
  \end{equation*}}
  We then derive this:
  {\small\begin{align*}
    &(\T \otimes \id) \circ \Swap \circ \Ctrl~\X \circ \Swap \circ \Ctrl~\X \\
    & \quad = (\T \otimes \id) \circ \Ctrl~\X \circ \Ctrl~\X \circ \Swap \circ
    \Ctrl~\X \circ \Swap \circ \Ctrl~\X & ((\Ctrl~\X)^2 = \id) \\
    & \quad = (\T \otimes \id) \circ (\id \otimes \Had) \circ \Ctrl~\Z
    \circ (\id \otimes \Had) \circ \Ctrl~\X \circ \Swap \circ \Ctrl~\X
    \circ \Swap \circ \Ctrl~\X & (\text{Lem.~\ref{lem:ctrlh}}) \\
    & \quad = (\id \otimes \Had) \circ (\T \otimes \id) \circ \Ctrl~\Z
    \circ (\id \otimes \Had) \circ \Ctrl~\X \circ \Swap \circ \Ctrl~\X
    \circ \Swap \circ \Ctrl~\X & (\text{bifunctoriality}~\oplus) \\
    & \quad = (\id \otimes \Had) \circ \Ctrl~\Z \circ (\T \otimes \id)
    \circ (\id \otimes \Had) \circ \Ctrl~\X \circ \Swap \circ \Ctrl~\X
    \circ \Swap \circ \Ctrl~\X & \ref{eq:a16} \\
    & \quad = (\id \otimes \Had) \circ \Ctrl~\Z \circ (\id \otimes
    \Had) \circ (\T \otimes \id) \circ \Ctrl~\X \circ \Swap \circ
    \Ctrl~\X \circ \Swap \circ \Ctrl~\X &
    (\text{bifunctoriality}~\oplus) \\
    & \quad = \Ctrl~\X \circ (\T \otimes \id) \circ \Ctrl~\X \circ \Swap \circ
    \Ctrl~\X \circ \Swap \circ \Ctrl~\X &
    (\text{Lem.~\ref{lem:ctrlh}}) \\
    & \quad = \Ctrl~\X \circ (\T \otimes \id) \circ \Swap &
    \eqref{eq:p6} \\
    & \quad = \Ctrl~\X \circ \Swap \circ (\id \otimes \T) &
    (\text{naturality}~\Swap) \\
    & \quad = \Ctrl~\X \circ \Ctrl~\X \circ \Swap \circ \Ctrl~\X \circ
    \Swap \circ \Ctrl~\X \circ (\id \otimes \T) &
    \eqref{eq:p6} \\
    & \quad = \Swap \circ \Ctrl~\X \circ
    \Swap \circ \Ctrl~\X \circ (\id \otimes \T) &
    ((\Ctrl~\X)^2 = \id)
  \end{align*}}\label{eq:a17}
\item \label{eq:a18} As noted by \citet{bianselinger:cliffordt}, this identity and
  the next are both of the form
  {\small\begin{equation*}
  \Qcircuit @C=.7em @R=0.7em @!R {
    & \ctrl{1} & \ctrlo{1} & \qw \\
    & \gate{U} & \gate{W} & \qw
  } \qquad \raisebox{-2mm}{=} \qquad
  \Qcircuit @C=.7em @R=0.7em @!R {
    & \ctrlo{1} & \ctrl{1} & \qw \\
    & \gate{W} & \gate{U} & \qw
  }
  \end{equation*}}
  for some $U : I \oplus I \to I \oplus
  I$ and $W : I \oplus I \to I \oplus I$. This is because
  \begin{align*}
    & \id \otimes g^{-1} \circ \nCtrl~f \circ \id \otimes g \\
    & \quad = \id \otimes g^{-1} \circ \Mat^{-1} \circ (f \oplus \id)
    \circ \Mat \circ \id \otimes g & (\text{definition}~\nCtrl) \\
    & \quad = \Mat^{-1} \circ (g^{-1} \oplus g^{-1}) \circ (f \oplus \id)
    \circ (g \oplus g) \circ \Mat & (\text{Lem.~\ref{lem:mat} (i)}) \\
    & \quad = \Mat^{-1} \circ \circ ((g^{-1} \circ f \circ g) \oplus
    (g^{-1} \circ g) \circ \Mat & (\text{bifunctoriality}~\oplus) \\
    & \quad = \Mat^{-1} \circ ((g^{-1} \circ f \circ g) \oplus
    \id) \circ \Mat & (g~\text{invertible})
  \end{align*}
  In other words, conjugating a negatively controlled $f$-gate by $g$
  on the target line yields a negatively controlled $g^{-1} \circ f
  \circ g$-gate (idem for positively
  controlled gates). Thus, it suffices to show that positively
  controlled gates commute with negatively controlled gates.
  \begin{align*}
    & \Ctrl~f \circ \nCtrl~g \\
    & = \Mat^{-1} \circ (\id \oplus f) \circ \Mat \circ \Mat^{-1}
    \circ (g \oplus \id) \circ \Mat & (\text{definition}~\Ctrl, \nCtrl) \\
    & = \Mat^{-1} \circ (\id \oplus f)
    \circ (g \oplus \id) \circ \Mat & (\Mat~\text{invertible}) \\
    & = \Mat^{-1} \circ (g \oplus \id) \circ (\id \oplus f) \circ \Mat
    & (\text{bifunctoriality}~\oplus) \\
    & = \Mat^{-1} \circ (g \oplus \id) \circ \Mat \circ \Mat^{-1}
    \circ (\id \oplus f) \circ \Mat & (\Mat~\text{invertible}) \\
    & = \nCtrl~g \circ \Ctrl~f &
    (\text{definition}~\Ctrl, \nCtrl)
  \end{align*}
\item \label{eq:a19} As above.
\item \label{eq:a20} We defer the derivation of this identity to Appendix~\ref{ap:a20}.
  \begin{appendixproof}[Proof of ~\ref{eq:a20}]
    \label{ap:a20}
  To start, this identity involves a controlled Hadamard
  gate, which by \cite{bianselinger:cliffordt} is taken as the shorthand
  {\scriptsize \begin{equation*}
  \Qcircuit @C=.7em @R=0.7em @!R {
  & \ctrl{1} & \qw \\
  & \gate{H} & \qw
  } \qquad \raisebox{-3mm}{=} \qquad
  \Qcircuit @C=.7em @R=0.7em @!R {
  & \qw & \qw & \qw & \ctrl{1} & \qw & \qw & \qw & \qw \\
  & \gate{S} & \gate{H} & \gate{T} & \targ & \gate{T^{-1}}
  & \gate{H} & \gate{S^{-1}} & \qw
  }
  \end{equation*}}
  Since this representation is very inconvenient, we start by
  showing that it is equal to the far simpler $\Ctrl~\Had$. Since,
  as previously observed regarding controlled gates conjugated by other
  gates on the target line,
  {\scriptsize \begin{equation*}
    \Qcircuit @C=.7em @R=0.7em @!R {
    & \qw & \qw & \qw & \ctrl{1} & \qw & \qw & \qw & \qw \\
    & \gate{S} & \gate{H} & \gate{T} & \targ & \gate{T^{-1}}
    & \gate{H} & \gate{S^{-1}} & \qw
    }
  \end{equation*}}
  is a controlled $\Sg^{-1} \circ \Had \circ \T^{-1} \circ \X \circ \T
  \circ \Had \circ \Sg$ gate, it suffices to show that $\Sg^{-1} \circ
  \Had \circ \T^{-1} \circ \X \circ \T \circ \Had \circ \Sg$ is
  nothing more than $\Had$, which follows by
  \begin{align*}
    \Sg^{-1} \circ \Had \circ \T^{-1} \circ \X \circ \T \circ \Had
    \circ \Sg & =
    \Sg \circ \Z \circ \Had \circ \Z \circ \Sg \circ \T \circ \X \circ
    \T \circ \Had \circ \Sg & (\text{Lem.~\ref{lem:gates} (iv)}) \\
    & = \Sg \circ \Z \circ \Had \circ \Z \circ \Sg \circ (\omega
    \sprod \X) \circ \Had \circ \Sg & (\text{Lem.~\ref{lem:gates} (vi)}) \\
    & = \omega \sprod \Sg \circ \Z \circ \Had \circ \Z \circ \Sg \circ
    \X \circ \Had \circ \Sg & (\text{Prop.~\ref{prop:scalars}}) \\
    & = \omega \sprod \Sg \circ \Z \circ \Had \circ \Z \circ \X \circ
    (i \sprod \Z \circ \Sg) \circ \Had \circ \Sg &
    (\text{Lem..~\ref{lem:gates} (vi)}) \\
    & = i \circ \omega \sprod \Sg \circ \Z \circ \Had \circ \Z \circ
    \X \circ \Z \circ \Sg \circ \Had \circ \Sg &
    (\text{Prop.~\ref{prop:scalars}}) \\
    & = i \circ \omega \sprod \Sg \circ \Z \circ \Had \circ (-1 \sprod
    \X) \circ \Sg \circ \Had \circ \Sg &
    (\text{Lem..~\ref{lem:gates} (vi)}) \\
    & = -1 \circ i \circ \omega \sprod \Sg \circ \Z \circ \Had \circ
    \X \circ \Sg \circ \Had \circ \Sg &
    (\text{Prop.~\ref{prop:scalars}}) \\
    & = \omega^{-1} \sprod \Sg \circ \Z \circ \Had \circ
    \X \circ \Sg \circ \Had \circ \Sg &
    (\text{Prop.~\ref{prop:scalars}}) \\
    & = \omega^{-1} \sprod \Sg \circ \Z \circ \Z \circ \Had \circ \Sg
    \circ \Had \circ \Sg & (\text{Lem.~\ref{lem:had}}) \\
    & = \omega^{-1} \sprod \Sg \circ \Had \circ \Sg
    \circ \Had \circ \Sg & (\text{Lem.~\ref{lem:gates} (v)}) \\
    & = \omega^{-1} \sprod (\omega \sprod \Had) & (\text{see \ref{eq:a12}}) \\
    & = (\omega^{-1} \circ \omega) \sprod \Had
    & (\text{Prop.~\ref{prop:scalars}}) \\
    & = \Had & (\omega~\text{invertible})
  \end{align*}
  With that shown, we can move on to showing the final identity. We do
  this by showing four smaller identities which, together, imply this
  last identity, namely equations (5)--(8) in \cite{bianselinger:cliffordt}.
  \begin{enumerate}[label={(\roman*)}]
  \item Directly by Lem.~\ref{lem:mat} (v).
  \item This is straightforwardly derived as
    \begin{align*}
      & \Ctrl~\T \circ \nCtrl~\T \\
      & \quad = \Mat^{-1} \circ (\id \oplus \T) \circ \Mat \circ \Mat^{-1}
      \circ (\T \oplus \id) \circ \Mat & (\text{definition}~\Ctrl,\nCtrl) \\
      & \quad = \Mat^{-1} \circ (\id \oplus \T)
      \circ (\T \oplus \id) \circ \Mat
      & (\Mat~\text{invertible}) \\
      & \quad = \Mat^{-1} \circ (\T \oplus \T) \circ \Mat &
      (\text{bifunctoriality}~\oplus) \\
      & \quad = \Mat^{-1} \circ \Mat \circ (\id \otimes \T) &
      (\text{Lem.~\ref{lem:mat} (1)}) \\
      & \quad = \id \otimes \T & (\Mat~\text{invertible})
    \end{align*}
  \item We first see that
    {\small\begin{align*}
      & \Swap \circ \nCtrl~\T \circ \Swap \circ (\id \otimes \T) \\
      & \quad = \Swap \circ \Mat^{-1} \circ (\T \oplus \id) \circ \Mat
      \circ \Swap \circ (\id \otimes \T) & (\text{def.}~\nCtrl) \\
      & \quad = \Mat^{-1} \circ \Midswap \circ (\T \oplus \id) \circ
      \Midswap \circ \Mat \circ (\id \otimes \T)
      & (\text{Lem.~\ref{lem:mat} (ii, iii)}) \\
      & \quad = \Mat^{-1} \circ \Midswap \circ (\T \oplus \id) \circ
      \Midswap \circ (\T \oplus \T) \circ \Mat
      & (\text{Lem.~\ref{lem:mat} (i)}) \\
      & \quad = \Mat^{-1} \circ \Midswap \circ ((\id \oplus \omega)
      \oplus \id) \circ \Midswap \circ (\T \oplus \T) \circ \Mat &
      (\text{def.}~\T) \\
      & \quad = \Mat^{-1} \circ ((\id \oplus \id) \oplus (\omega
      \oplus \id)) \circ \Midswap \circ \Midswap \circ (\T \oplus \T)
      \circ \Mat & (\text{naturality}~\Midswap) \\
      & \quad = \Mat^{-1} \circ ((\id \oplus \id) \oplus (\omega
      \oplus \id)) \circ (\T \oplus \T)
      \circ \Mat & (\Midswap~\text{invertible}) \\
      & \quad = \Mat^{-1} \circ ((\id \oplus \omega) \oplus (\omega
      \oplus \omega)) \circ \Mat & (\text{def.}~\T, \text{bifunctoriality}~\oplus)
    \end{align*}}
    and then derive
    {\begin{align*}
      & \nCtrl~\T \circ (\T \otimes \id) \\
      & \quad = \nCtrl~\T \circ \Swap \circ
      (\id \otimes \T) \circ \Swap & (\text{naturality}~\Swap)\\
      & \quad = \Swap \circ \Swap \circ \nCtrl~\T \circ \Swap \circ (\id
      \otimes \T) \circ \Swap & (\Swap~\text{invertible})\\
      & \quad = \Swap \circ \Mat^{-1} \circ ((\id \oplus \omega) \oplus
      (\omega \oplus \omega)) \circ \Mat \circ \Swap & (\text{above}) \\
      & \quad = \Swap \circ \Mat^{-1} \circ ((\id \oplus \omega) \oplus
      (\omega \oplus \omega)) \circ \Midswap \circ \Mat
      & (\text{Lem.~\ref{lem:mat} (ii)}) \\
      & \quad = \Swap \circ \Mat^{-1} \circ \Midswap \circ ((\id
      \oplus \omega) \oplus (\omega \oplus \omega)) \circ \Mat &
      (\text{naturality}~\Midswap) \\
      & \quad = \Swap \circ \Swap \circ \Mat^{-1} \circ ((\id \oplus
      \omega) \oplus (\omega \oplus \omega)) \circ \Mat
      & (\text{Lem.~\ref{lem:mat} (iii)}) \\
      & \quad = \Mat^{-1} \circ ((\id \oplus \omega) \oplus
      (\omega \oplus \omega)) \circ \Mat & (\Swap~\text{invertible}) \\
      & \quad = \Swap \circ \nCtrl~\T \circ \Swap \circ (\id \otimes \T)
      & (\text{above})
    \end{align*}}
  \item We derive
    \begin{align*}
      & \Ctrl~\Had \circ (\id \otimes \T) \circ \nCtrl~\Had \\
      & \quad = \Mat^{-1} \circ (\id \oplus \Had) \circ \Mat \circ (\id
      \otimes \T) \circ \Mat^{-1} \circ (\Had \oplus \id) \circ \Mat
      & (\text{def.}~\Ctrl, \nCtrl)\\
      & \quad = \Mat^{-1} \circ (\id \oplus \Had) \circ (\T \oplus \T)
      \circ \Mat \circ \Mat^{-1} \circ (\Had \oplus \id) \circ \Mat
      & (\text{Lem.~\ref{lem:mat} (1)})\\
      & \quad = \Mat^{-1} \circ (\id \oplus \Had) \circ (\T \oplus \T)
      \circ (\Had \oplus \id) \circ \Mat
      & (\Mat~\text{invertible})\\
      & \quad = \Mat^{-1} \circ ((\T \circ \Had) \oplus (\Had \circ
      \T)) \circ \Mat
      & (\text{bifunctoriality}~\oplus)\\
      & \quad = \Mat^{-1} \circ (\T \oplus \id) \circ (\Had \oplus
      \Had) \circ (\id \oplus \T) \circ \Mat
      & (\text{bifunctoriality}~\oplus)\\
      & \quad = \Mat^{-1} \circ (\T \oplus \id) \circ (\Had \oplus
      \Had) \circ \Mat \circ \Mat^{-1} \circ (\id \oplus \T) \circ \Mat
      & (\Mat~\text{invertible})\\
      & \quad = \Mat^{-1} \circ (\T \oplus \id) \circ \Mat \circ (\id
      \otimes \Had) \circ \Mat^{-1} \circ (\id \oplus \T) \circ \Mat
      & (\text{Lem.~\ref{lem:mat} (1)})\\
      & \quad = \nCtrl~\T \circ (\id \otimes \Had) \circ \Ctrl~\T
      & (\text{def.}~\Ctrl, \nCtrl)
    \end{align*}
  \end{enumerate}
  \end{appendixproof}
\end{enumerate}

Summing up:

\begin{theorem}\label{thm:clifft2}
  Let $c_1$ and $c_2$ be \SPiLang\ terms representing Clifford+T circuits
  of at most two qubits. Then $\sem{c_1} = \sem{c_2}$ iff $\stsem{c_1}
  = \stsem{c_2}$.
\end{theorem}
\begin{proof}
  \ref{eq:a1}--\ref{eq:a20} are sound and complete for Clifford+T
  circuits of at most two qubits~\cite{bianselinger:cliffordt}, and
  have been shown to hold in any model of \SPiLang{} (see also
  Thm.~\ref{thm:2cliff}).
\end{proof}

\subsection{Unitaries with entries in $\mathbb{Z}[\frac{1}{2},i]$}
\label{subsec:unitaries_unz}
\begin{figure}
\centering
{\small\[
\begin{array}{rcl l rcl l}
  i_{[j]}^4 &=& \id & \text{(D1)} &
  i_{[k]} X_{[j,k]} &=& X_{[j,k]} i_{[j]}  & \text{(D10)} \\
  \X^2_{[j,k]} &=& \id & \text{(D2)} &
  \X_{[k,l]} \X_{[j,k]} &=& \X_{[j,k]} \X_{[j,l]} & \text{(D11)} \\
  \K^8_{[j,k]} &=& \id & \text{(D3)} &
  \X_{[j,l]} \X_{[k,l]} &=& \X_{[k,l]} \X_{[j,k]} & \text{(D12)} \\
  &&&&
  \K_{[k,l]} \X_{[j,k]} &=& \X_{[j,k]} \K_{[j,l]} & \text{(D13)} \\
  i_{[j]} i_{[k]} &=& i_{[k]} i_{[j]} & \text{(D4)} &
  \K_{[j,l]} \X_{[k,l]} &=& \X_{[k,l]} \K_{[j,k]} & \text{(D14)} \\
  i_{[j]} \X_{[k,l]} &=& \X_{[k,l]} i_{[j]} & \text{(D5)} &
  \K_{[j,k]} i^2_{[k]} &=& \X_{[j,k]} \K_{[j,k]} & \text{(D15)} \\
  i_{[j]} \K_{[k,l]} &=& \K_{[k,l]} i_{[j]} & \text{(D6)} &
  \K_{[j,k]} i^3_{[k]} &=& i_{[k]} \K_{[j,k]} i_{[k]} \K_{[j,k]} & \text{(D16)} \\
  \X_{[j,k]} \X_{[l,m]} &=& \X_{[l,m]} \X_{[j,k]} & \text{(D7)} &
  \K_{[j,k]} i_{[j]} i_{[k]} &=& i_{[j]} i_{[k]} \K_{[j,k]} & \text{(D17)} \\
  \X_{[j,k]} \K_{[l,m]} &=& \K_{[l,m]} \X_{[j,k]} & \text{(D8)} &
  \K^2_{[j,k]} i_{[j]} i_{[k]} &=& \id & \text{(D18)} \\
  \K_{[j,k]} \K_{[l,m]} &=& \K_{[l,m]} \K_{[j,k]} & \text{(D9)} &
  \K_{[j,k]} \K_{[l,m]} \K_{[j,l]} \K_{[k,m]} &=&
  \K_{[j,l]} \K_{[k,m]} \K_{[j,k]} \K_{[l,m]} & \text{(D19)}
\end{array}
\]}
\caption{The sound and complete equational theory of Gaussian dyadic
  rational unitaries due to \cite{bianselinger:unz}.}
\label{fig:unitaries_unz}
\end{figure}
We now show that models of \SPiLang\ are equationally
sound and complete for unitaries with entries from
the ring $\mathbb{Z}[\frac{1}{2},i]$
(\textit{i.e.}, the ring of integers extended with $\frac{1}{2}$ and $i$).
We call these \emph{Gaussian dyadic rational unitaries}. It was shown by
\citet{Amy2020numbertheoretic} that every circuit in the
computationally universal \emph{Gaussian Clifford+T} gate set has an
\emph{exact} representation as a unitary matrix with entries in
$\mathbb{Z}[\frac{1}{2},i]$. A sound and complete equational theory
for these unitaries was given by \citet{bianselinger:unz} (see
Fig.~\ref{fig:unitaries_unz}). In other words, these unitaries are
enough to approximate any other finite quantum computation to any
desired degree of error, and they can be reasoned about using a sound
and complete equational theory.

In this section, we show that this equational theory is subsumed by
that of \SPiLang. Then we show that the
easy direction of \cite{Amy2020numbertheoretic} can also be
internalised in models of \SPiLang, thus proving equational soundness
and completeness for Gaussian Clifford+T circuits.

Unlike the previous results, which concerned circuits (formed using
$\otimes$), this result concerns only matrices (formed using
$\oplus$). This also means that the presentation (in
Fig.~\ref{fig:unitaries_unz}) is quite different. Gaussian dyadic
rational unitaries are generated by $i$, $\X$, and $\K$, where $\K$ is
a variant of the Hadamard gate given by $\K = \omega^{-1} \sprod
\Had$\footnote{Note the slight discrepancy in the literature that
\citet{bianselinger:unz} take $\K = \omega^{-1} \sprod \Had$ while
\citet{Amy2020numbertheoretic} use $\K = \omega \sprod \Had$. However,
since one definition is inverse to the other, and
$U_n(\mathbb{Z}[\frac{1}{2},i])$ is closed under inversion, the
particular choice doesn't matter so long as it is done
consistently.}. In Fig.~\ref{fig:unitaries_unz}, these are
additionally given indices, assumed distinct, corresponding to the
component(s) that the generator is applied to. When proving these
identities, we further assume indices to start from $1$ and to be
consecutive in the order written. We are free to do so since we can
simply conjugate by the appropriate permutation to make it so
(recalling that \PiLang{} can express all permutations).  Likewise, we
will assume identities to be minimal, and only consider the case that
uses the number of distinct indices; any other case reduces to this by
appending an identity morphism as necessary using the direct sum and
conjugating by a permutation. For example, in the context on an $n
\times n$ unitary (\textit{i.e.}, a morphism $I^{\oplus n} \to
I^{\oplus n}$, where $I^{\oplus n}$ is taken as usual to mean the
$n$-fold direct sum of $I$ with itself), $\X_{[2,3]}$ is taken to mean
$\id_I \oplus \X \oplus \id_{I^{\oplus n-3}}$ (up to
associativity). To form $\X_{[2,4]}$ would require us to conjugate
this by the permutation swapping the third and fourth components.
\begin{definition}
In a model of \SPiLang, a \emph{representation of a Gaussian dyadic
rational unitary} is any morphism which can be written in terms of
morphisms from the sets $\{i, \K\}$ and $\{\alpha_\oplus,
\alpha_\oplus^{-1}, \lambda_\oplus, \lambda_\oplus^{-1}, \rho_\oplus,
\rho_\oplus^{-1},\allowbreak \sigma_\oplus\}$, composed arbitrarily in parallel
(using $\oplus$) and in sequence (using $\circ$).
\end{definition}
Note that the above definition permits the use of $\X$ since $\X =
\sigma_\oplus$ by definition. It is additionally important to realise
that the notion of parallel composition is different between the above
the previous definitions concerning circuits, as this uses the direct
sum $\oplus$ for parallel composition whereas the circuits used the
tensor product $\otimes$.

We show that the identities of Fig.~\ref{fig:unitaries_unz} are all
satisfied in any model of \SPiLang.
\begin{enumerate}[label={(D\arabic*)}]
\item $i^4 = (\omega^2)^4 = \omega^8 = \id$ by \ref{ax:omega}. \label{eq:d1}
\item $\X^2 = \sigma_\oplus^2 = \id$ by the rig axioms. \label{eq:d2}
\item We start by seeing that \label{eq:d3}
  \begin{align*}
    \K^2 & = (\omega^{-1} \sprod \Had) \circ (\omega^{-1} \sprod \Had)
    & (\text{def.}~\K) \\
    & = (\omega^{-1} \circ \omega^{-1}) \sprod \Had \circ \Had
    & (\text{Prop.~\ref{prop:scalars}})\\
    & = (\omega^7 \circ \omega^7) \sprod \id
    & \text{\ref{eq:a4}}\\
    & = (\omega^8 \circ \omega^6) \sprod \id
    & (\circ~\text{associative})\\
    & = \omega^6 \sprod \id & \text{\ref{ax:omega}}
  \end{align*}
  and so $K^8 = (K^2)^4 = (\omega^6 \sprod \id)^4 = \omega^{24} \sprod
  \id = (\omega^8 \circ \omega^8 \circ \omega^8) \sprod \id = \id$ by
  \ref{ax:omega} and Prop.~\ref{prop:scalars}.
\item[(D4--9)] These are all instances of bifunctoriality for
  $\oplus$, \textit{i.e.}, $(f \oplus \id) \circ (\id \oplus g) = (\id \oplus
  g) \circ (f \oplus \id)$.\setcounter{enumi}{9}\label{eq:d4-d9}
\item We have \label{eq:d10}
  \begin{align*}
    (\id \oplus i) \circ \X
    & = (\id \oplus i) \circ \sigma_\oplus & (\text{definition}~\X) \\
    & = \sigma_\oplus \circ (i \oplus \id) & (\text{naturality}~\sigma_\oplus) \\
    & = \X \circ (i \oplus \id) & (\text{definition}~\X)
  \end{align*}
\item \label{eq:d11} We show the more general case for any $f$, from which this
  identity follows as the case of $f = \X$. Marking lines in the string
  diagram by indices, we see that this is nothing but
  \begin{center}
    \includegraphics[scale=0.8]{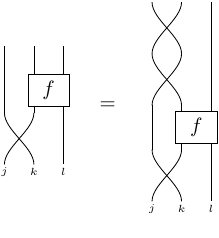}
  \end{center}
  which follows by invertibility of the symmetry.
\item \label{eq:d12} Likewise, we show the more general case for any
  $f$, from which this identity will follows as the case where $f =
  \X$. Marking lines in the string diagram by indices, we get
  \begin{center}
    \includegraphics[scale=0.8]{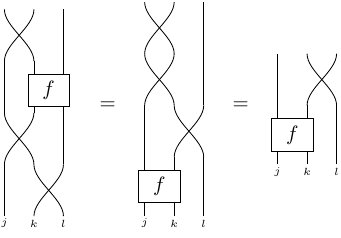}
  \end{center}
  which follows by (respectively) naturality and invertibility of the
  symmetry.
\item \label{eq:d13} This follows by the generalised form of
  \ref{eq:d11} with $f = \K$.
\item \label{eq:d14} This follows by the generalised form of
  \ref{eq:d12} with $f = \K$.
\item \label{eq:d15} We have
  \begin{align*}
    \K \circ \Z
    & = \K \circ \Z \circ \Had \circ \Had & \text{\ref{eq:a4}} \\
    & = \K \circ \Z \circ \Had \circ (\omega \sprod \K) &
    (\text{definition}~\Had) \\
    & = (\omega \sprod \K) \circ \Z \circ \Had \circ \K &
    (\text{Prop.~\ref{prop:scalars}}) \\
    & = \Had \circ \Z \circ \Had \circ \K & (\text{definition}~\Had) \\
    & = \X \circ \K & (\text{Lem.~\ref{lem:had}})
  \end{align*}
\item \label{eq:d16} We reduce
  \begin{align*}
    \K \circ \Z \circ \Sg
    & = \X \circ \K \circ \Sg & \text{\ref{eq:d15}} \\
    & = \X \circ \X \circ \Sg \circ \V \circ \Sg \circ \X \circ \Sg &
    (\text{definition}~\K) \\
    & = \Sg \circ \V \circ \Sg \circ \X \circ \Sg &
    (\X~\text{involutive}) \\
    & = \Sg \circ \V \circ (i \sprod \X) &
    (\text{Lem.~\ref{lem:gates} (vi)}) \\
    & = i \sprod \Sg \circ \V \circ \X &
    (\text{Prop.~\ref{prop:scalars}})
  \end{align*}
  and
  \begin{align*}
    \Sg \circ \K \circ \Sg \circ \K
    & = \Sg \circ \X \circ \Sg \circ \V \circ \Sg \circ \X \circ \Sg
    \circ \X \circ \Sg \circ \V \circ \Sg \circ \X &
    (\text{definition}~\K) \\
    & = (i \sprod \X) \circ \V \circ \Sg \circ \X \circ (i \sprod \X)
    \circ \V \circ \Sg \circ \X
    & (\text{Lem.~\ref{lem:gates} (vi)}) \\
    & = i^2 \sprod \X \circ \V \circ \Sg \circ \X \circ \X \circ \V
    \circ \Sg \circ \X
    & (\text{Prop.~\ref{prop:scalars}}) \\
    & = -1 \sprod \X \circ \V \circ \Sg \circ \V \circ \Sg \circ \X
    & (\X~\text{involutive}) \\
    & = -1 \sprod \X \circ \V \circ (-i \sprod \V \circ \Sg \circ \V) \circ \X
    & \text{\ref{ax:sqy}} \\
    & = -1 \circ -i \sprod \X \circ \V \circ \V \circ \Sg \circ \V \circ \X
    & (\text{Prop.~\ref{prop:scalars}}) \\
    & = i \sprod \X \circ \X \circ \Sg \circ \V \circ \X
    & \text{\ref{ax:v}} \\
    & = i \sprod \Sg \circ \V \circ \X
    & (\X~\text{involutive})
  \end{align*}
  so $\K \circ \Z \circ \Sg = i \sprod \Sg \circ \V \circ \X = \Sg
  \circ \K \circ \Sg \circ \K$.
\item \label{eq:d17} It follows that
  \begin{align*}
    \K \circ (i \oplus i) & = \K \circ (i \sprod (\id \oplus \id))
    & (\text{Prop.~\ref{prop:scalars}}) \\
    & = i \sprod \K \circ \id
    & (\text{bifunctoriality}~\oplus) \\
    & = i \sprod \K
    & (\text{Prop.~\ref{prop:scalars}}) \\
    & = i \sprod (\id \oplus \id) \circ \K
    & (\text{bifunctoriality}~\oplus) \\
    & = (i \oplus i) \circ \K
    & (\text{Prop.~\ref{prop:scalars}})
  \end{align*}
\item \label{eq:d18} We derive
  \begin{align*}
    \K^2 \circ (i \oplus i)
    & = \K^2 \circ (i \sprod (\id \oplus \id))
    & (\text{Prop.~\ref{prop:scalars}}) \\
    & = i \sprod \K^2
    & (\text{Prop.~\ref{prop:scalars}}) \\
    & = i \sprod (\omega^{-1} \sprod \Had) \circ (\omega^{-1} \sprod \Had)
    & (\text{definition}~\K) \\
    & = i \circ \omega^{-1} \circ \omega^{-1} \sprod \Had \circ \Had
    & (\text{Prop.~\ref{prop:scalars}}) \\
    & = i \circ -i \sprod \id
    & \text{\ref{eq:a4}} \\
    & = \id
    & \text{\ref{ax:omega}}
  \end{align*}
\item \label{eq:d19} We derive this final identity by showing that it
  is an instance of bifunctoriality of the tensor product in disguise:
  {\small\begin{align*}
    & \Midswap \circ (\K \oplus \K) \circ \Midswap \circ (\K \oplus \K) \\
    & \quad = \Mat \circ \Mat^{-1} \circ \Midswap \circ (\K \oplus \K) \circ
    \Midswap \circ (\K \oplus \K) \circ \Mat \circ \Mat^{-1}
    & (\Mat~\text{invertible})\\
    & \quad = \Mat \circ \Mat^{-1} \circ \Midswap \circ (\K \oplus \K) \circ
    \Midswap \circ \Mat \circ (\id \otimes \K) \circ \Mat^{-1}
    & (\text{Lem.~\ref{lem:mat}} (i)) \\
    & \quad = \Mat \circ \Mat^{-1} \circ \Midswap \circ (\K \oplus \K) \circ
    \Mat \circ \Swap \circ (\id \otimes \K) \circ \Mat^{-1}
    & (\text{Lem.~\ref{lem:mat}} (ii)) \\
    & \quad = \Mat \circ \Mat^{-1} \circ \Midswap \circ \Mat \circ
    (\id \otimes \K) \circ \Swap \circ (\id \otimes \K) \circ
    \Mat^{-1}
    & (\text{Lem.~\ref{lem:mat}} (i)) \\
    & \quad = \Mat \circ \Mat^{-1} \circ \Mat \circ \Swap \circ
    (\id \otimes \K) \circ \Swap \circ (\id \otimes \K) \circ
    \Mat^{-1}
    & (\text{Lem.~\ref{lem:mat}} (ii)) \\
    & \quad = \Mat \circ \Mat^{-1} \circ \Mat \circ \Swap \circ
    (\id \otimes \K) \circ (\K \otimes \id) \circ \Swap \circ
    \Mat^{-1}
    & (\text{naturality}~\Swap) \\
    & \quad = \Mat \circ \Mat^{-1} \circ \Mat \circ \Swap \circ
    (\K \otimes \id) \circ (\id \otimes \K) \circ \Swap \circ
    \Mat^{-1}
    & (\text{bifunctoriality}~\oplus) \\
    & \quad = \Mat \circ \Mat^{-1} \circ \Mat \circ (\id \otimes \K)
    \circ \Swap \circ (\id \otimes \K) \circ \Swap \circ \Mat^{-1}
    & (\text{naturality}~\Swap) \\
    & \quad = \Mat \circ \Mat^{-1} \circ (\K \oplus \K) \circ \Mat
    \circ \Swap \circ (\id \otimes \K) \circ \Swap \circ \Mat^{-1}
    & (\text{Lem.~\ref{lem:mat}} (i)) \\
    & \quad = \Mat \circ \Mat^{-1} \circ (\K \oplus \K) \circ \Midswap \circ \Mat
    \circ (\id \otimes \K) \circ \Swap \circ \Mat^{-1}
    & (\text{Lem.~\ref{lem:mat}} (ii)) \\
    & \quad = \Mat \circ \Mat^{-1} \circ (\K \oplus \K) \circ \Midswap
    \circ (\K \oplus \K) \circ \Mat \circ \Swap \circ \Mat^{-1}
    & (\text{Lem.~\ref{lem:mat}} (i)) \\
    & \quad = \Mat \circ \Mat^{-1} \circ (\K \oplus \K) \circ \Midswap
    \circ (\K \oplus \K) \circ \Midswap \circ \Mat \circ \Mat^{-1}
    & (\text{Lem.~\ref{lem:mat}} (ii)) \\
    & \quad = (\K \oplus \K) \circ \Midswap
    \circ (\K \oplus \K) \circ \Midswap
    & (\Mat~\text{invertible})
  \end{align*}}
\end{enumerate}

We obtain yet another equational completeness result:

\begin{theorem}[Full abstraction for Gaussian dyadic rational unitaries]
  \label{thm:comp_unz}
  Let $c_1$ and $c_2$ be \SPiLang\ terms representing unitaries with
  entries in the ring $\mathbb{Z}[\frac{1}{2},i]$. Then $\sem{c_1} =
  \sem{c_2}$ iff $\stsem{c_1} = \stsem{c_2}$.
\end{theorem}
\begin{proof}
  Identities \ref{eq:d1}--\ref{eq:d19} form a sound and complete
  equational theory for Gaussian dyadic rational
  unitaries~\cite{bianselinger:unz}.
\end{proof}

\subsection{Gaussian Clifford+T Circuits}
\label{subsec:gaussianclifft}

We mentioned in Sec.~\ref{subsec:unitaries_unz} the one-to-one
correspondence (due to \cite{Amy2020numbertheoretic}) between circuits
in the (computationally universal) Gaussian Clifford+T gate
set $\{\X, \CX, \CCX, \K, \Sg\}$ and Gaussian dyadic rational
unitaries. 
\begin{definition}
In a model of \SPiLang, a \emph{representation of a Gaussian
Clifford+T circuit} is any morphism which can be written in terms of
morphisms from the sets $\{\X, \CX, \CCX, \K, \Sg\}$ and
$\{\alpha_\otimes, \alpha_\otimes^{-1}, \lambda_\otimes,
\lambda_\otimes^{-1}, \allowbreak \rho_\otimes, \rho_\otimes^{-1},
\sigma_\otimes\}$, composed arbitrarily in parallel (using $\otimes$)
and in sequence (using $\circ$).
\end{definition}
We argue that we can reason about Gaussian Clifford+T circuits in
models of \SPiLang\ by reasoning about their matrices, using the
coherence theorem for rig categories. Recall that a
\emph{bipermutative category} is a rig category where both symmetric
monoidal structures are strict, and the annihilators and right
distributor are all identities. (The explicit definition can be found in
\cite{may:eringspectra}.)

The coherence theorem for rig categories can be stated in terms of
bipermutative categories as follows:
\begin{theorem}
  Any rig category is rig equivalent to a bipermutative category.
\end{theorem}
\begin{proof}
  See \cite[VI, Prop. 3.5]{may:eringspectra}.
\end{proof}
We can use this theorem to make the rig structure in any
model of \SPiLang\ bipermutative. This is very handy since we notice
that in a bipermutative category, the isomorphism $\Mat : (I \oplus I)
\otimes A \to A \oplus A$ is the identity, as it is composed of
the right distributor and some unitors; similarly, $\Midswap : (A
\oplus B) \oplus (C \oplus D) \to (A \oplus C) \oplus (B \oplus D)$ is
$\id \oplus \sigma_\oplus \oplus \id$ (we don't need to
worry about associativity due to strictness). Since in a general model
of \SPiLang\ we have
$$
\CX = \Ctrl~\X = \Mat^{-1} \circ (\id \oplus \X) \circ \Mat,
$$
in a bipermutative model of \SPiLang\ we have $\CX =
\id \oplus \X$; and $\CCX = (\id \oplus (\id \oplus
\X))$. As
$$
\Swap = \Mat^{-1} \circ \Mat \circ \Swap = \Mat^{-1} \circ \Midswap
\circ \Mat
$$ by invertibility of $\Mat$ and Lem.~\ref{lem:mat}, we have that
$\Swap = \Midswap = \id \oplus \X \oplus \id$ in the bipermutative
case, so even swapping two circuit lines reduces to applying $\X$. As
such, $\X$, $\CX$, $\CCX$, $\K$, $\Sg$, and $\Swap$ are all Gaussian
dyadic rational unitaries in a bipermutative model of \SPiLang.
This is the key observation in obtaining equational soundness and
completeness for Gaussian Clifford+T circuits (as it was for classical
reversible circuits as well~\cite{10.1145/3498667}).

We will need a small lemma (with proof in Appendix~\ref{ap:swapassoc}). Let
$\mathsf{SWAPASSOC} : (I \oplus I) \otimes ((I \oplus I) \otimes A)
\to (I \oplus I) \otimes ((I \oplus I) \otimes A)$ denote the natural
isomorphism $\alpha_\otimes \circ \Swap \otimes \id \circ
\alpha^{-1}_\otimes$.
\begin{lemmarep}\label{lem:swapassoc}
  In any model of \SPiLang, we have
  $$(\Mat \oplus \Mat) \circ \Mat \circ \mathsf{SWAPASSOC} =
  \Midswap \circ (\Mat \oplus \Mat) \circ \Mat.$$
\end{lemmarep}
\begin{proof}
\label{ap:swapassoc}
  This follows by commutativity of the diagram in Fig.~\ref{fig:diag}.

  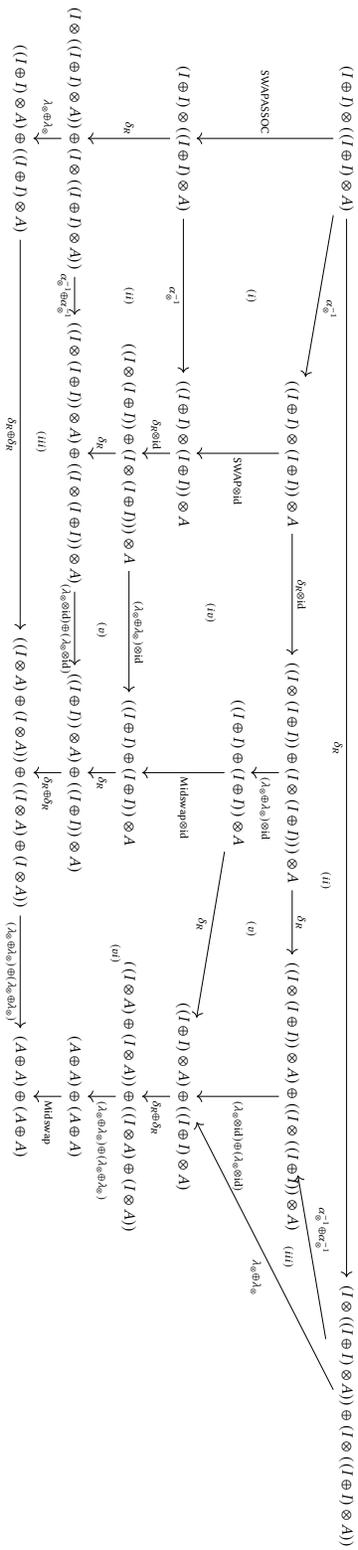
\begin{figure}
\adjustbox{scale=0.6,angle=-90,center}{\begin{tikzcd}
	{(I \oplus I) \otimes ((I \oplus I) \otimes A)} &&&& {(I \otimes ((I \oplus I) \otimes A)) \oplus (I \otimes ((I \oplus I) \otimes A))} \\
	& {((I \oplus I) \otimes (I \oplus I)) \otimes A} & {((I \otimes (I \oplus I)) \oplus (I \otimes (I \oplus I))) \otimes A} & {((I \otimes (I \oplus I)) \otimes A) \oplus ((I \otimes ((I \oplus I)) \otimes A)} \\
	&& {((I \oplus I) \oplus (I \oplus I)) \otimes A} \\
	{(I \oplus I) \otimes ((I \oplus I) \otimes A)} & {((I \oplus I) \otimes (I \oplus I)) \otimes A} && {((I \oplus I) \otimes A) \oplus ((I \oplus I) \otimes A)} \\
	& {((I \otimes (I \oplus I)) \oplus (I \otimes (I \oplus I))) \otimes A} & {((I \oplus I) \oplus (I \oplus I)) \otimes A} & {((I \otimes A) \oplus (I \otimes A)) \oplus ((I \otimes A) \oplus (I \otimes A))} \\
	{(I \otimes ((I \oplus I) \otimes A)) \oplus (I \otimes ((I \oplus I) \otimes A))} & {((I \otimes (I \oplus I)) \otimes A) \oplus ((I \otimes (I \oplus I)) \otimes A)} & {((I \oplus I)) \otimes A) \oplus ((I \oplus I)) \otimes A)} & {(A \oplus A) \oplus (A \oplus A)} \\
	{((I \oplus I) \otimes A) \oplus ((I \oplus I) \otimes A)} && {((I \otimes A) \oplus (I \otimes A)) \oplus ((I \otimes A) \oplus (I \otimes A))} & {(A \oplus A) \oplus (A \oplus A)}
	\arrow["{\mathsf{SWAPASSOC}}"', from=1-1, to=4-1]
	\arrow[""{name=0, anchor=center, inner sep=0}, "{\alpha_\otimes^{-1}}", from=1-1, to=2-2]
	\arrow["{\Swap \otimes \id}", from=2-2, to=4-2]
	\arrow[""{name=1, anchor=center, inner sep=0}, "{\alpha_\otimes^{-1}}"', from=4-1, to=4-2]
	\arrow["{\delta_R}"', from=4-1, to=6-1]
	\arrow[""{name=2, anchor=center, inner sep=0}, "{\alpha_\otimes^{-1} \oplus \alpha_\otimes^{-1}}"', from=6-1, to=6-2]
	\arrow["{\delta_R \otimes \id}"', from=4-2, to=5-2]
	\arrow["{\delta_R}"', from=5-2, to=6-2]
	\arrow["{\lambda_\otimes \oplus \lambda_\otimes}"', from=6-1, to=7-1]
	\arrow[""{name=3, anchor=center, inner sep=0}, "{(\lambda_\otimes \otimes \id) \oplus (\lambda_\otimes \otimes \id)}"', from=6-2, to=6-3]
	\arrow[""{name=4, anchor=center, inner sep=0}, "{\delta_R \oplus \delta_R}"', from=7-1, to=7-3]
	\arrow["{\delta_R \oplus \delta_R}", from=6-3, to=7-3]
	\arrow[""{name=5, anchor=center, inner sep=0}, "{(\lambda_\otimes \oplus \lambda_\otimes) \otimes \id}", from=5-2, to=5-3]
	\arrow["{\delta_R}", from=5-3, to=6-3]
	\arrow[""{name=6, anchor=center, inner sep=0}, "{\delta_R \otimes \id}", from=2-2, to=2-3]
	\arrow["{(\lambda_\otimes \oplus \lambda_\otimes) \otimes \id}", from=2-3, to=3-3]
	\arrow["{\Midswap \otimes \id}", from=3-3, to=5-3]
	\arrow[""{name=7, anchor=center, inner sep=0}, "{\delta_R}", from=2-3, to=2-4]
	\arrow[""{name=8, anchor=center, inner sep=0}, "{\alpha_\otimes^{-1} \oplus \alpha_\otimes^{-1}}"', shift left=2, from=1-5, to=2-4]
	\arrow[""{name=9, anchor=center, inner sep=0}, "{\delta_R}"', from=1-1, to=1-5]
	\arrow[""{name=10, anchor=center, inner sep=0}, "{\delta_R}"', from=3-3, to=4-4]
	\arrow["{(\lambda_\otimes \otimes \id) \oplus (\lambda_\otimes \otimes \id)}", from=2-4, to=4-4]
	\arrow[""{name=11, anchor=center, inner sep=0}, "{\lambda_\otimes \oplus \lambda_\otimes}", from=1-5, to=4-4]
	\arrow["{\delta_R \oplus \delta_R}", from=4-4, to=5-4]
	\arrow["{(\lambda_\otimes \oplus \lambda_\otimes) \oplus (\lambda_\otimes \oplus \lambda_\otimes)}", from=5-4, to=6-4]
	\arrow["\Midswap", from=6-4, to=7-4]
	\arrow[""{name=12, anchor=center, inner sep=0}, "{(\lambda_\otimes \oplus \lambda_\otimes) \oplus (\lambda_\otimes \oplus \lambda_\otimes)}"', from=7-3, to=7-4]
	\arrow["{(i)}"{description}, draw=none, from=0, to=1]
	\arrow["{(ii)}"{description}, draw=none, from=1, to=2]
	\arrow["{(ii)}"{description}, draw=none, from=9, to=2-4]
	\arrow["{(iii)}"{description}, draw=none, from=6-2, to=4]
	\arrow["{(iii)}"{description}, draw=none, from=8, to=11]
	\arrow["{(iv)}"{description}, draw=none, from=6, to=5]
	\arrow["{(v)}"{description}, draw=none, from=5, to=3]
	\arrow["{(v)}"{description}, draw=none, from=7, to=10]
	\arrow["{(vi)}"{description}, draw=none, from=10, to=12]
\end{tikzcd}}
\caption{\label{fig:diag}Diagram for proving Lem.~\ref{lem:swapassoc}.}
\end{figure}

Here (i) commutes by definition , (ii) by Laplaza (VII), (iii)
monoidal coherence for $\otimes$, (iv) by Lem.~\ref{lem:mat}, (v) by
naturality of $\delta_R$, and (vi) using Laplaza (I).
\end{proof}

\begin{theorem}[Full abstraction for Gaussian Clifford+T circuits]\label{thm:gclifft}
  Let $c_1$ and $c_2$ be \SPiLang\ terms representing Gaussian
  Clifford+T circuits. Then $\sem{c_1} = \sem{c_2}$ iff $\stsem{c_1} =
  \stsem{c_2}$.
\end{theorem}
\begin{proof}
  Let $c_1, c_2 : (I \oplus I)^{\otimes n} \to (I \oplus I)^{\otimes
    n}$. By coherence, we may assume every model of \SPiLang\ in
  sight to be bipermutative.

  As noted above, the gates of the Gaussian Clifford+T gate set are
  all representations of Gaussian dyadic rational unitaries in this
  bipermutative model: $\X$ and $\K$ are so directly, and $\Sg = \id
  \oplus i$, $\CX = \id \oplus \X$ and $\CCX = \id \oplus (\id \oplus
  \X)$ are so too by closure under direct sums. To see that the tensor
  product of two representations is also a representation, it
  suffices to show that tensoring by identities on $(I \oplus
  I)^{\otimes m}$ on either side preserves this property, since we
  have $(f \otimes \id) \circ (\id \otimes g) = f \otimes g$:
  \begin{itemize}
    \item By Lem.~\ref{lem:mat}, tensoring by $\id_{I \oplus I}$ on
      the left yields $\id_{I \oplus I} \otimes f = \Mat^{-1} \circ (f
      \oplus f) \circ \Mat$, so in the bipermutative case $\id_{I
        \oplus I} \otimes f = f \oplus f$, which is again a
      representation of a Gaussian dyadic rational unitary unitary
      when $f$ is, by closure under direct sum. But then we can repeat
      this process $m-1$ times to tensor by $\id_{(I \oplus
        I)^{\otimes m}}$.
    \item By naturality, $f \otimes \id_{(I \oplus I)^{\otimes m}} =
      \sigma_\otimes \circ \id_{(I \oplus I)^{\otimes m}} \otimes f
      \circ \sigma_\otimes$, so this reduces to the case above since
      (in the bipermutative case, using Lems.~\ref{lem:swapassoc} and
      \ref{lem:mat}) the symmetry $\sigma_\otimes$ on $(I \oplus
      I)^{\otimes p} \otimes (I \oplus I)^{\otimes q}$ is nothing but
      a series of direct sums of identities and $\oplus$-symmetries on
      $I \oplus I$ (\textit{i.e.}, $\X$ gates).
  \end{itemize}
  Finally, since representations of Gaussian dyadic rational unitaries
  are also closed under composition, it follows that any
  representation of a Gaussian Clifford+T circuit in a bipermutative
  category is directly also a representation of a Gaussian dyadic
  rational unitary.

  From this it follows for terms $c_1$ and $c_2$ representing Gaussian
  Clifford+T circuits that $\sem{c_1} = \sem{c_2}$ iff they are equal
  as representations of Gaussian dyadic rational unitaries, which in
  turn happens (by Thm.~\ref{thm:comp_unz}) iff they are equal as
  actual unitaries in $\Unitary$ (so specifically as Gaussian
  Clifford+T circuits), \textit{i.e.}, iff $\stsem{c_1} = \stsem{c_2}$.
\end{proof}

\begin{toappendix}
  \section{Definition of bipermutative category}
  \label{def:bipermutative}
\begin{definition}
  A \emph{bipermutative category} is a rig category where
  \begin{enumerate}
  \item the associators $\alpha_\otimes : (A \otimes B) \otimes C \to
    A \otimes (B \otimes C)$ and $\alpha_\oplus : (A \oplus B) \oplus
    C \to A \oplus (B \oplus C)$ and unitors $\lambda_\oplus : O
    \oplus A \to A$, $\rho_\oplus : A \oplus O \to A$,
    $\lambda_\otimes : I \otimes A \to A$, and $\rho_\otimes : A
    \otimes I \to A$ are all identities.
  \item the annihilators $\delta^0_R : A \otimes O \to O$
    and $\delta^0_L : O \otimes A \to O$ and right
    distributor $\delta_R : (A \oplus B) \otimes C \to (A \otimes C)
    \oplus (B \otimes C)$ are identities, and the following diagram
    commutes:
\[\begin{tikzcd}
	{(A \oplus B) \otimes C} & {(A \otimes C) \oplus (B \otimes C)} \\
	{(B \oplus A) \otimes C} & {(B \otimes C) \oplus (A \otimes C)}
	\arrow["{=}", from=1-1, to=1-2]
	\arrow["{=}"', from=2-1, to=2-2]
	\arrow["{\sigma_\oplus \otimes \id}"', from=1-1, to=2-1]
	\arrow["{\sigma_\oplus}", from=1-2, to=2-2]
\end{tikzcd}\]
  \item The left distributor $\delta_L : A \otimes (B \oplus C) \to (A
    \otimes B) \oplus (A \otimes C)$ makes the diagrams below commute:
\[\begin{tikzcd}
	{A \otimes (B \oplus C)} & {(B \oplus C) \otimes A} \\
	{(A \otimes B) \oplus (A \otimes C)} & {(B \otimes A) \oplus (C \otimes A)}
	\arrow["{\delta_L}"', from=1-1, to=2-1]
	\arrow["{\sigma_\otimes}", from=1-1, to=1-2]
	\arrow["{=}", from=1-2, to=2-2]
	\arrow["{\sigma_\otimes \oplus \sigma_\otimes}", from=2-2, to=2-1]
\end{tikzcd}\]
\[\begin{tikzcd}
	{(A \oplus B) \otimes (C \oplus D)} & {(A \otimes (C \oplus D)) \oplus (B \otimes (C \oplus D))} \\
	{((A \oplus B) \otimes C) \oplus ((A \oplus B) \otimes D)} & {(A \otimes C) \oplus (A \otimes D) \oplus (B \otimes C) \oplus (B \otimes D)} \\
	& {(A \otimes C) \oplus (B \otimes C) \oplus (A \otimes D) \oplus (B \otimes D)}
	\arrow["{=}", from=1-1, to=1-2]
	\arrow["{\delta_L}"', from=1-1, to=2-1]
	\arrow["{\delta_L \oplus \delta_L}", from=1-2, to=2-2]
	\arrow["{=}"', from=2-1, to=3-2]
	\arrow["{\id \oplus \sigma_\oplus \oplus \id}", from=2-2, to=3-2]
\end{tikzcd}\]
  \end{enumerate}
\end{definition}
\end{toappendix}

\section{Circuit Equivalences}
\label{sec:eq}
As a supplement to this paper, we have developed an Agda library and
used it to formalise some of our results. We discuss its use in
proving the Sleator-Weinfurter decomposition of $\CCX$ mentioned in
Sec.~\ref{sec:comb}, as well as keys aspects of the implementation.

\subsection{Decomposing $\CCX$}
In the previous section, we noted that every gate in the Gaussian
Clifford+T gate set has a ``matrix representation'', \textit{i.e.}, that it can
be written as $\Mat^{-1} \circ g \circ \Mat$ for some $g$ that only
uses $\K$, $\X$, $i$, direct sums and composition. To prove the
correctness of the Sleator-Weinfurter decomposition (see
Fig.~\ref{fig:sw} on page~\pageref{fig:sw}), we will use a common
technique: find the matrix form of each gate, compose them to form the
circuit, and use elementary reasoning to take care of the rest.

The first step seems simple given that each elementary gate has a
matrix representation, but additional work is required in the case of
multi-qubit circuits. This is because the exact positioning of the
gate alters its representation. For example, to find the matrix
representation of a $\CX$ applied to the top two qubits of a three
qubit circuit, we apply it instead to the bottom two qubits and apply
$\Swap$ gates to ``rewire'' the circuit appropriately, as in

{\scriptsize\[
\Qcircuit @C=.7em @R=0.7em @!R {
    & \ctrl{1} & \qw \\
    & \targ    & \qw \\
    & \qw      & \qw 
} \quad\raisebox{-4mm}{=}\quad
\Qcircuit @C=.7em @R=0.7em @!R {
    & \qw        & \qswap     & \qw      & \qswap     & \qw        & \qw \\
    & \qswap     & \qswap\qwx & \ctrl{1} & \qswap\qwx & \qswap     & \qw \\
    & \qswap\qwx & \qw        & \targ    & \qw        & \qswap\qwx & \qw
}
\]}

\noindent
This form allows us to use Lems.~\ref{lem:mat} and \ref{lem:swapassoc}
to find its matrix representation, which turns out (with a bit of
work) to be
$$
\Mat^{-1} \circ (\Mat^{-1} \oplus \Mat^{-1}) \circ
(\id \oplus \sigma_\oplus^{I \oplus I,I \oplus I}) \circ
(\Mat \oplus \Mat) \circ
\Mat \enspace.
$$

\noindent
We use the same technique to find the matrix representation of the
remaining gates in the circuit and compose them, yielding (after
removing a number of superfluous $\Mat^{-1} \circ \Mat$)
\begin{align*}
& \Mat^{-1} \circ (\Mat^{-1} \oplus \Mat^{-1}) \circ
(\id \oplus (\V \oplus \V)) \circ
(\id \oplus \sigma_\oplus^{I \oplus I, I \oplus I}) \circ
((\id \oplus \V^{-1}) \oplus (\id \oplus \V^{-1}))~\circ \\
& \quad
(\id \oplus \sigma_\oplus^{I \oplus I, I \oplus I}) \circ
((\id \oplus \V) \oplus (\id \oplus \V)) \circ
(\Mat \oplus \Mat) \circ \Mat
\end{align*}

\noindent
Expanding out and applying naturality of $\sigma_\oplus$,
invertibility of $\V$, and bifunctoriality a few times show that this
is equivalent to our previous definition of $\CCX$, \textit{i.e.}
$$
\Mat^{-1}\circ\left(\id \oplus
  \left(\Mat^{-1}\circ\left(\id \oplus X\right) \circ \Mat\right) 
\right)\circ \Mat.
$$

\noindent
An Agda program implementing the formal proof can be found in the
supplementary material. The equational proofs are reasonably readable
by humans (much more so than tactic proofs would be) but not so
enlightening that including them here would be warranted.

\subsection{Agda implementation}

Presented with the choice of working in the syntax
of~\SPiLang~(Sec.~\ref{sec:sqrtpi}) or in its generic models
(Def.~\ref{def:model}), we chose to work in the
latter for purely practical considerations: the library
\texttt{agda-categories} already contains a wealth of reasoning
combinators for both categories and monoidal categories that we would
have to reproduce in the syntax of the language. Furthermore, it also has proofs of useful results,
such as Kelly's various coherence lemmas, and defines useful extra
combinators like ``middle exchange'' (our \Midswap). As we would have
had to reproduce all of that, this seemed like a simple choice.

However, everything in \texttt{agda-categories} is \emph{weak}, so that
we have to worry about units and association in our formal proofs.
Doing this manually is overwhelmingly tedious. Luckily, there are a lot of
combinators already defined that make this essentially bearable.
The translation from the proofs presented in the paper, which ignore
associativity altogether, does require some care.

We have not yet had a chance to formalise everything. We did formalise
all of Sec.~\ref{sec:den}, all results in Sec.~\ref{subsec:2clifford},
Lem.~\ref{lem:ctrlh} of Sec.~\ref{subsec:nclifford},
Lem.~\ref{lem:nctrl}, and \ref{eq:a14} to \ref{eq:a17} in
Sec.~\ref{subsec:2clifft}. We foresee no additional difficulties for
other parts, except that many of the later equations are larger. Going
at ``full speed,'' a proof like that of Sleator-Weinfurter takes a
little over an hour of dedicated work.  However, identities like
\ref{eq:b1}--\ref{eq:b4} and \ref{eq:a20} are likely to take several
hours each.

We did not find any errors in any of the paper proofs while
formalising them. We did find several cross-referencing errors
(\textit{i.e.}, the wrong lemma justifying the step had been written down),
which were subsequently corrected. Interestingly, we did find an error
in \texttt{agda-categories} itself: it was missing some coherences for
\texttt{RigCategory}. This error has been fixed in the library.

We did find that some classical coherences used in the proofs of
Lem.~\ref{prop:scalars} and~\ref{lem:gates} were significantly more
work to prove than the diagrammatic sketches let on. Three of the
sub-parts of these ``preliminary lemmas'' accounted for more than a
day's work each.

Nevertheless, we conclude that doing categorical meta-theory for
quantum programming languages absolutely can be formalised at a
reasonable cost.

\section{Concluding Remarks} 
\label{sec:speed}
 
In this paper we have studied square roots from a purely axiomatic
perspective. We have shown that with a remarkably small extension to
the classical reversible programming language \PiLang{}, one can
obtain a language which is computational universal as well as sound
and complete for a variety of modes of unitary quantum computing.  A
key feature of our approach (also found in other successful calculi
such as the ZX-calculus) is the treatment of gates as white boxes that
can be decomposed and recomposed during rewriting. This is in contrast
to the circuit based approach that treat gates as black boxes. For
example, while a circuit theory will allow one to derive that $\T\T =
\Sg$, it is unable to provide justification for this in terms of the
definitions of $\Sg$ and $\T$. On the other hand, our approach reduces
this equation to the bifunctoriality of $\oplus$ and the definition of
$\Sg$ and $\T$. This style of reasoning is very close to the kind of
semi-formal reasoning used to justify matrix equalities (employed,
\text{e.g.}, in \cite{bianselinger:cliffordt} to justify their
relations).

Physically, square roots are a key feature of quantum
hardware. To understand this point, we briefly delve under the
computational abstraction to the level of energy flow. At that level,
the quantum mechanical description of a system is expressed using a
Hamiltonian that is continuous in time (and assumed here to be time
independent). Given a Hamiltonian $H$ and some initial state
$\ket{\psi(0)}$, the state of the system at a subsequent time $t$ is
given by:
\[
    \ket{\psi(t)} ~=~ e^{-iHt}\ket{\psi(0)}
\]
In the circuit model of quantum computing, the quantity $e^{-iHt}$
denotes a unitary $U$ that is implemented by a gate or collection of
gates. Mathematically, it is clearly legitimate to decompose $U =
e^{-iHt}$ into $\sqrt{U} \circ \sqrt{U} = e^{-iHt/2} \cdot
e^{-iHt/2}$.  This decomposition has a simple operational realisation:
if the application of $U$ requires an energy pulse lasting $k$ units
of time, then applying the pulse for $k/2$ units implements
$\sqrt{U}$~\cite[VII.F.2]{google:supremacy}. It turns out that the
classical computing abstraction generally does not allow such
decompositions, whereas quantum computing is distinguished by this
feature.

\begin{toappendix}
\label{ap:hask}
\begin{minted}{haskell}
module Demo where

-- A class for booleans with, possibly, 
-- a square root of negation

class Enum a => B a where
  falseB    :: a
  trueB     :: a
  notB      :: a -> a
  sqrtNotB  :: a -> a
  evenB     :: a -> Bool
  evenB     = even . fromEnum 

-- The classical instance has no square root

instance B Bool where
  falseB    = False
  trueB     = True
  notB      = not
  sqrtNotB  = error "No classical sqrt of not"

-- Now define "big" booleans: Zero and Two are the
-- classical booleans; One and Three are intermediate 
-- values along the negation trajectories

data Four = Zero | One | Two | Three

-- Create the trajectories for boolean negation:
-- Zero -> One -> Two
-- Two -> Three -> Zero

instance Enum Four where
  toEnum 0 = Zero
  toEnum 1 = One
  toEnum 2 = Two
  toEnum 3 = Three
  toEnum n = toEnum (n `mod` 4)
  fromEnum Zero = 0
  fromEnum One = 1
  fromEnum Two = 2
  fromEnum Three = 3

instance B Four where
  falseB    = Zero
  trueB     = Two
  notB      = succ . succ
  sqrtNotB  = succ

-- When boolean negation is applied to Zero, it produces Two after
-- "internally" visiting the intermediate value One. Although
-- the particular internal values are not exposed, the evenB
-- method reveals whether the underlying value is a "whole" or
-- "partial" boolean. 

data Classification = Balanced | Constant

-- An analogue of Deutsch's problem.
-- We have four functions defined on abstract booleans:
-- two constant functions (f0 and f1) and two balanced
-- functions (f2 and f3)

f0, f1, f2, f3 :: B a => a -> a
f0 a = falseB
f1 a = trueB
f2 a = a
f3 a = notB a

-- Classically the given function is applied twice to
-- classify it as Balanced or Constant

deutschC :: (Bool -> Bool) -> Classification
deutschC f = if f False == f True then Constant else Balanced

-- If we can observe the values introduced by the square
-- roots, we only need one application! 

deutschF :: (Four -> Four) -> Classification
deutschF f = if evenB (f (sqrtNotB falseB)) then Constant else Balanced
\end{minted}
\end{toappendix}
  
The fact that a function and its square root operate at different time
scales suggests evidence for the widely-believed exponential speedup
that distinguishes quantum from classical computing. Indeed the simple
Haskell module in Appendix~\ref{ap:hask} shows that, if we arrange for
boolean negation to take \emph{two steps}, then it is possible to
model the analogue of a square root of boolean negation by just taking
one of the two steps, and most importantly, this leads to the same
quantum speedup observed Deutsch's
problem~\cite{deutsch,deutschJozsa}. Taking this idea further, it is
arguably the case that more and more square roots, for example by
providing additional roots of unity, would unlock additional speedup
opportunities. We consider a formal investigation of these connections
to be an important direction of future work. 


\section*{Acknowledgements}
\label{sec:acks}
We are indebted to the reviewers for their thoughtful and detailed
comments. Jacques Carette is supported by NSERC grant RGPIN-2018-05812. 
Amr Sabry was supported by US National Science Foundation grant OMA-1936353.

\printbibliography

\appendix

\end{document}